\documentclass[11pt]{article}

\usepackage[utf8]{inputenc}
\usepackage[T1]{fontenc}
\usepackage{lmodern}
\usepackage[left=1in,right=1in,top=1in,bottom=1.25in]{geometry}

\usepackage{microtype}

\usepackage{amsmath}
\usepackage{amssymb}
\usepackage{amsthm}
\usepackage{thmtools}
\usepackage{thm-restate}

\usepackage{titlesec}
\usepackage{fancybox}
\usepackage{xcolor}
\usepackage{xspace}
\usepackage{enumerate}

\usepackage{comment}

\usepackage[
	style=alphabetic,
	backref=true,
	doi=false,
	url=false,
	maxcitenames=3,
	mincitenames=3,
	maxbibnames=10,
	minbibnames=10,
	backend=bibtex8,
	sortlocale=en_US
]{biblatex}

\usepackage[ocgcolorlinks]{hyperref} 
\usepackage{cleveref}

\allowdisplaybreaks[1]

\makeatletter
\g@addto@macro\bfseries{\boldmath}
\makeatother

\makeatletter
\g@addto@macro\mdseries{\unboldmath}
\g@addto@macro\normalfont{\unboldmath}
\g@addto@macro\rmfamily{\unboldmath}
\g@addto@macro\upshape{\unboldmath}
\makeatother

\DeclareCiteCommand{\citem}
    {}
    {\mkbibbrackets{\bibhyperref{\usebibmacro{postnote}}}}
    {\multicitedelim}
    {}

\renewcommand*{\multicitedelim}{\addcomma\space}

\AtEveryCitekey{\clearfield{note}}

\newcommand{\myhref}[1]{%
  \iffieldundef{doi}
    {\iffieldundef{url}
       {#1}
       {\href{\strfield{url}}{#1}}}
    {\href{http://dx.doi.org/\strfield{doi}}{#1}}%
}

\DeclareFieldFormat{title}{\myhref{\mkbibemph{#1}}}
\DeclareFieldFormat
  [article,inbook,incollection,inproceedings,patent,thesis,unpublished]
  {title}{\myhref{\mkbibquote{#1\isdot}}}

\makeatletter
\AtBeginDocument{%
    \newlength{\temp@x}%
    \newlength{\temp@y}%
    \newlength{\temp@w}%
    \newlength{\temp@h}%
    \def\my@coords#1#2#3#4{%
      \setlength{\temp@x}{#1}%
      \setlength{\temp@y}{#2}%
      \setlength{\temp@w}{#3}%
      \setlength{\temp@h}{#4}%
      \adjustlengths{}%
      \my@pdfliteral{\strip@pt\temp@x\space\strip@pt\temp@y\space\strip@pt\temp@w\space\strip@pt\temp@h\space re}}%
    \ifpdf
      \typeout{In PDF mode}%
      \def\my@pdfliteral#1{\pdfliteral page{#1}}
      \def\adjustlengths{}%
    \fi
    \ifxetex
      \def\my@pdfliteral #1{}
      \def\adjustlengths{\setlength{\temp@h}{-\temp@h}\addtolength{\temp@y}{1in}\addtolength{\temp@x}{-1in}}%
    \fi%
    \def\Hy@colorlink#1{%
      \begingroup
        \ifHy@ocgcolorlinks
          \def\Hy@ocgcolor{#1}%
          \my@pdfliteral{q}%
          \my@pdfliteral{7 Tr}
        \else
          \HyColor@UseColor#1%
        \fi
    }%
    \def\Hy@endcolorlink{%
      \ifHy@ocgcolorlinks%
        \my@pdfliteral{/OC/OCPrint BDC}%
        \my@coords{0pt}{0pt}{\pdfpagewidth}{\pdfpageheight}%
        \my@pdfliteral{F}
        \my@pdfliteral{EMC/OC/OCView BDC}%
        \begingroup%
          \expandafter\HyColor@UseColor\Hy@ocgcolor%
          \my@coords{0pt}{0pt}{\pdfpagewidth}{\pdfpageheight}%
          \my@pdfliteral{F}
        \endgroup%
        \my@pdfliteral{EMC}%
        \my@pdfliteral{0 Tr}
        \my@pdfliteral{Q}%
      \fi
      \endgroup
    }%
}
\makeatother

\makeatletter

\def\moverlay{\mathpalette\mov@rlay}
\def\mov@rlay#1#2{\leavevmode\vtop{%
		\baselineskip\z@skip \lineskiplimit-\maxdimen
		\ialign{\hfil$\m@th#1##$\hfil\cr#2\crcr}}}
\newcommand{\charfusion}[3][\mathord]{
	#1{\ifx#1\mathop\vphantom{#2}\fi
		\mathpalette\mov@rlay{#2\cr#3}
	}
	\ifx#1\mathop\expandafter\displaylimits\fi}
\makeatother

\newcommand{\cupdot}{\charfusion[\mathbin]{\cup}{\cdot}}

\colorlet{DarkRed}{red!50!black}
\colorlet{DarkGreen}{green!50!black}
\colorlet{DarkBlue}{blue!50!black}

\hypersetup{
	linkcolor = DarkRed,
	citecolor = DarkGreen,
	urlcolor = DarkBlue,
	bookmarksnumbered = true,
	linktocpage = true
}

\declaretheorem[numberwithin=section]{theorem}
\declaretheorem[numberlike=theorem]{lemma}

\declaretheorem[numberlike=theorem]{corollary}
\declaretheorem[numberlike=theorem]{definition}

\newcommand{\defeq}{:=}
\newcommand{\dist}{d}
\newcommand{\ww}{w}
\newcommand{\lap}{\mathcal{L}}

\newcommand{\prob}[2][]{\mathbb{P}_{#1} \left[ #2 \right]}
\newcommand{\expec}[2][]{\mathbb{E}_{#1} \left[ #2 \right]}
\newcommand{\abs}[1]{\left| #1 \right|}

\newcommand{\VC}{V\!C}
\newcommand{\MVC}{M\!V\!C}
\newcommand{\OPT}{O\!PT}

\newcommand{\tSpectral}{\lceil 12 (c+1) \alpha \epsilon^{-2} \ln{n} \rceil}
\newcommand{\tCut}{ C_\xi  c  \alpha \log W \log^2 n / \epsilon^2 }

\def\polylog{\operatorname{polylog}}
\def\poly{\operatorname{poly}}

\newenvironment{fminipage}%
  {\begin{Sbox}\begin{minipage}}%
  {\end{minipage}\end{Sbox}\fbox{\TheSbox}}

\newenvironment{algbox}[0]{\vskip 0.2in
\noindent
\begin{fminipage}{6.3in}
}{
\end{fminipage}
\vskip 0.2in
}

\ifdefined\ShowComment

\def\todo#1{{ \color{red} TODO: #1}}

\def\sebastian#1{\marginpar{$\leftarrow$\fbox{S}}\footnote{\sf \color{magenta} #1 --Sebastian}}

\fi

\title{On Fully Dynamic Graph Sparsifiers}
\author{
Ittai Abraham\thanks{VMware Research}
\and
David Durfee\thanks{Georgia Institute of Technology}
\and
Ioannis Koutis\thanks{University of Puerto Rico, Rio Piedras}
\and
Sebastian Krinninger\thanks{Max Planck Institute for Informatics, Saarland Informatics Campus, Germany
}
\and
Richard Peng\thanks{Georgia Institute of Technology}
}

\date{}

\hypersetup{
	pdftitle = {On Fully Dynamic Graph Sparsifiers},
	pdfauthor = {Ittai Abraham, David Durfee, Ioannis Koutis, Sebastian Krinninger, Richard Peng}
}

\begin{document}
\pagenumbering{roman}
\maketitle
\begin{abstract}
We initiate the study of fast dynamic algorithms for graph sparsification problems and obtain fully dynamic algorithms, allowing both edge insertions and edge deletions, that take polylogarithmic time after each update in the graph.
Our three main results are as follows.
First, we give a fully dynamic algorithm for maintaining a $ (1 \pm \epsilon) $-spectral sparsifier with amortized update time $\poly(\log{n}, \epsilon^{-1})$.
Second, we give a fully dynamic algorithm for maintaining a $ (1 \pm \epsilon) $-cut sparsifier with \emph{worst-case} update time $\poly(\log{n}, \epsilon^{-1})$.
Both sparsifiers have size $ n \cdot \poly(\log{n}, \epsilon^{-1})$.
Third, we apply our dynamic sparsifier algorithm to obtain a fully dynamic algorithm for maintaining a $(1 + \epsilon)$-approximation to the value of the maximum flow in an unweighted, undirected, bipartite graph with amortized update time $\poly(\log{n}, \epsilon^{-1})$.

\end{abstract}
\newpage

\tableofcontents
\newpage

\setcounter{page}{0}
\pagenumbering{arabic}

\section{Introduction}

Problems motivated by graph cuts are well studied
in theory and practice.
The prevalence of large graphs motivated sublinear time
algorithms for cut based problems such as
clustering~\cite{SpielmanT13,BorgsBCT12,AndersenCL06,
AndersenP09,OrecchiaV11,GharanT12}.
In many cases such as social networks or road networks, these
algorithms need to run on dynamically evolving graphs.
In this paper, we study an approach for obtaining sublinear time
algorithms for these problems based on dynamically maintaining
graph sparsifiers.

Recent years have seen a surge of interest in dynamic graph algorithms.
On the one hand, very efficient algorithms, with polylogarithmic running time per update in the graph, could be found for some key problems in the field~\cite{HenzingerK99,HolmLT01,KapronKM13,OnakR10,NeimanS13,BaswanaGS15,BhattacharyaHI15,BaswanaKS12,AbrahamCDGW16}.
On the other hand, there are polynomial conditional lower bounds for many basic graph problems~\cite{Patrascu10,AbboudW14,HenzingerKNS15}.
This leads to the question which problems can be solved with polylogarithmic update time.
Another relatively recent trend in graph algorithmics is graph sparsification where we reduce the size of graphs while approximately
preserving key properties such as the sizes of cuts~\cite{BenczurK15}.
These routines and their extensions to the spectral setting~\cite{SpielmanT11,BatsonSST13}
play central roles in a number of recent algorithmic
advances~\cite{Madry10,Sherman13,KelnerLOS14,PengS14,SpielmanTengSolver:journal,KyngLPSS16,Peng16},
often leading to graph algorithms that run in almost-linear time.
In this paper, we study problems at the intersection of dynamic algorithms and graph sparsification, leveraging ideas from both fields.

At the core of our approach are data structures that dynamically
maintain graph sparsifiers in $\polylog{n}$ time per edge
insertion or deletion.
They are motivated by the spanner based constructions
of spectral sparsifiers of Koutis~\cite{Koutis14}.
By modifying dynamic algorithms for spanners~\cite{BaswanaKS12},
we obtain data structures that spend amortized $\polylog{n}$ per update.
Our main result for spectral sparsifiers is:
\begin{restatable}[]{theorem}{mainSpectral}
\label{thm:mainSpectral}
Given a graph with polynomially bounded edge weights,
we can dynamically maintain a $(1 \pm \epsilon)$-spectral sparsifier of size $n \cdot \poly(\log{n}, \epsilon^{-1})$ with amortized update time
$\poly(\log{n}, \epsilon^{-1})$ per edge insertion / deletion.
\end{restatable}
When used as a black box, this routine allows us to run cut
algorithms on sparse graphs instead of the original,
denser network.
Its guarantees interact well with most routines that compute
minimum cuts or solve linear systems in the graph Laplacian.
Some of them include:
\begin{enumerate}
\item min-cuts, sparsest cuts, and separators~\cite{Sherman09},
\item eigenvector and heat kernel computations~\cite{OrecchiaSV12},
\item approximate Lipschitz learning on graphs~\cite{KyngRSS15} and a variety of matrix polynomials in the graph Laplacian~\cite{ChengCLPT15}.
\end{enumerate}

In many applications the full power of spectral sparsifiers is not needed, and it suffices to work with a cut sparsifier.
As spectral approximations imply cut approximations,
research in recent years has focused spectral sparsification
algorithms~\cite{KelnerL13,KoutisLP12,KapralovLMMS14,ZhuLO15,LeeS15,JindalK15}.
In the dynamic setting however we get a strictly stronger result for cut sparsifiers than for spectral sparsifiers: we can dynamically maintain cut sparsifiers with polylogarithmic \emph{worst-case} update time after each insertion / deletion.
We achieve this by generalizing Koutis' sparsification paradigm~\cite{Koutis14} and replacing spanners with approximate maximum spanning trees in the construction.
While there are no non-trivial results for maintaining spanners with worst-case update time, spanning trees can be maintained with polylogarithmic worst-case update time by a recent breakthrough result~\cite{KapronKM13}.
This allows us to obtain the following result for cut sparsifiers:
\begin{restatable}[]{theorem}{mainCut}
\label{thm:mainCut}
Given a graph with polynomially bounded edge weights,
we can dynamically maintain a $(1 \pm \epsilon)$-cut sparsifier of size $n \cdot \poly(\log{n}, \epsilon^{-1})$ with worst-case update time
$\poly(\log{n}, \epsilon^{-1})$ per edge insertion / deletion.
\end{restatable}

We then explore more sophisticated applications of dynamic
graph sparsifiers.
A key property of these sparsifiers is that they have arboricity $\polylog{n}$.
This means the sparsifier is locally sparse, and
can be represented as a union of spanning trees.
This property is becoming increasingly important in recent
works~\cite{NeimanS13,PelegS16}: 
Peleg and Solomon~\cite{PelegS16} gave data structures
for maintaining approximate maximum matchings on fully dynamic
graphs with amortized cost parameterized by the arboricity of the graphs.
We demonstrate the applicability of our data structures
for designing better data structures on the undirected
variant of the problem.
Through a two-stage application of graph sparsifiers,
we obtain the first non-separator based approach for
dynamically maintaining $(1 - \epsilon)$-approximate maximum
flow on fully dynamic graphs:
\begin{restatable}[]{theorem}{mainBipartite}
\label{thm:mainBipartite}
Given a dynamically changing unweighted, undirected, bipartite graph $ G = (A, B, E) $ with demand $ -1 $ on every vertex in $ A $ and demand $ 1 $ on every vertex in $ B $, we can maintain a $(1 - \epsilon)$-approximation to the value of the maximum flow, as well as query access to the associated approximate minimum cut, with amortized update time $\poly(\log{n}, \epsilon^{-1})$ per edge insertion / deletion.
\end{restatable}
To obtain this result we give stronger guarantees for vertex sparsification in bipartite graphs, identical to the terminal cut sparsifier question addressed by Andoni, Gupta, and Krauthgamer~\cite{AndoniGK14}.
Our new analysis profits from the ideas we develop by going back and forth
between combinatorial reductions and spectral sparsification.
This allows us to analyze a vertex sampling process via a mirror edge sampling
process, which is in turn much better understood.

Overall, our algorithms bring together a wide range of tools from
data structures,  spanners, and randomized algorithms.
We will provide more details on our routines, as well as how
they relate to existing combinatorial and probabilistic tools
in Section~\ref{sec:overview}.

\section{Background}
\label{sec:overview}

\subsection{Dynamic Graph Algorithms}

In this paper we consider undirected graphs $ G = (V, E) $ with $ n $ vertices and $ m $ edges that are either unweighted or have non-negative edge weights.
We denote the weight of an edge $ e = (u, v) $ in a graph $ G $ by $ \ww_G (e) $ or $ \ww_G (u, v) $ and the ratio between the largest and the smallest edge weight by $ W $.
The weight $ \ww_G (F) $ of a set of edges $ F \subseteq E $ is the sum of the individual edge weights.
We will assume that all weights are polynomially bounded because there
are standard reductions from the general case using minimum
spanning trees (e.g. ~\cite{SpielmanS11} Section 10.2.,~\cite{ElkinEST08} Theorem 5.2).
Also, these contraction schemes in the data structure setting introduces
another layer of complexity akin to dynamic connectivity, which we believe
is best studied separately.

A \emph{dynamic algorithm} is a data structure for dynamically maintaining the result of a computation while the underlying input graph is updated periodically.
We consider two types of updates: edge insertions and edge deletions.
An \emph{incremental} algorithm can handle only edge insertions, a \emph{decremental} algorithm can handle only edge deletions, and a \emph{fully dynamic} algorithm can handle both edge insertions and deletions.
After every update in the graph, the dynamic algorithm is allowed to process the update to compute the new result.
For the problem of maintaining a sparsifier, we want the algorithm to output the changes to the sparsifier (i.e., the edges to add to or remove from the sparsifier) after every update in the graph.

\subsection{Running Times and Success Probabilities}

The running time spent by the algorithm after every update is called \emph{update time}.
We distinguish between \emph{amortized} and \emph{worst-case} update time.
A dynamic algorithm has amortized update time $ T (m, n, W) $, if the total time spent after $ q $ updates in the graph is at most $ q T (m, n, W) $.
A dynamic algorithm has worst-case update time $ T (m, n, W) $, if the total time spent after \emph{each} update in the graph is at most $ T (m, n, W) $.
Here $ m $ refers to the maximum number of edges ever contained in the graph.
All our algorithms are randomized.

The guarantees we report in this paper (quality and size of sparsifier, and update time) will hold \emph{with high probability (w.h.p.)}, i.e. with probability at least $ 1 - 1/n^{c} $ for some arbitrarily chosen constant $ c \geq 1 $.
These bounds are against an \emph{oblivious adversary} who chooses
its sequence of updates independently
from the random choices made by the algorithm.
Formally, the oblivious adversary chooses its sequence of updates before the algorithm starts.
In particular, this means that the adversary is not allowed
to see the current edges of the sparsifier.
As our composition of routines involve $\poly(n)$ calls, we
will assume the composability of these w.h.p. bounds.

Most of our update costs have the form
$O(\log^{O(1)}{n} \epsilon^{-O(1)})$, where $\epsilon$ is the
approximation error.
We will often state these as $\poly(\log{n}, \epsilon^{-1})$ when the
exponents exceed $3$, and explicitly otherwise.

\subsection{Cuts and Laplacians}

A $ \emph{cut} $ $ U \subseteq V $ of $ G $ is a subset of vertices whose removal makes $ G $ disconnected.
We denote by $ \partial_G (U) $ the edges crossing the cut $ U $, i.e., the set of edges with one endpoint in $ U $ and one endpoint in $ V \setminus U $.
The weight of the cut $ U $ is $ \ww_G (\partial_G (U)) $.
An \emph{edge cut} $ F \subseteq E $ of $ G $ is a a subset of edges whose removal makes $ G $ disconnected and the weight of the edge cut $ F $ is $ \ww_G (F) $.
For every pair of vertices $ u $ and $ v $, the \emph{local edge connectivity} $ \lambda_G (u, v) $ is the weight of the minimum edge cut separating $ u $ and $ v $.
If $ G $ is unweighted, then $ \lambda_G (u, v) $ amounts to the number of edges that have to be removed from $ G $ to make $ u $ and $ v $ disconnected.

Assuming some arbitrary order $ v_1, \ldots v_n $ on the vertices, the \emph{Laplacian matrix} $ \lap_G $ of an undirected graph $ G $ is the $ n \times n $ matrix that in row $ i $ and column $ j $ contains the negated weight $ - w_G (v_i, v_j) $ of the edge $ (v_i, v_j) $ and in the $i$-th diagonal entry contains the weighted degree $ \sum_{j=1}^{n} w_G (v_i, v_j) $ of vertex $ v_i $.
Note that Laplacian matrices are symmetric.
The matrix $ \lap_e $ of an edge $ e $ of $ G $ is the $ n \times n $ Laplacian matrix of the subgraph of $ G $ containing only the edge $ e $.
It is $ 0 $ everywhere except for a $ 2 \times 2 $ submatrix.

For studying the spectral properties of $ G $ we treat the graph as a resistor network.
For every edge $ e \in E $ we define the \emph{resistance} of $ e $ as $ r_G (e) = 1 / \ww_G (e) $.
The \emph{effective resistance} $ R_G (e) $ of an edge $ e = (v, u) $ is defined as the potential difference that has to be applied to $ u $ and $ v $ to drive one unit of current through the network.
A closed form expression of the effective resistance is $ R_G (e) =  b_{u,v}^{\top} \lap_G^\dagger b_{u,v} $, where $ \lap_G^\dagger $ is the Moore-Penrose pseudo-inverse of the Laplacian matrix of $ G $ and $ b_{u,v} $ is the $n$-dimensional vector that is $ 1 $ at position $ u $, $ -1 $ at position $ v $, and $ 0 $ otherwise.

\subsection{Graph Approximations}

The goal of graph sparsification is to find sparse subgraphs, or similar small objects, that approximately preserve certain metrics of the graph.
We first define spectral sparsifiers where we require that Laplacian quadratic form of the graph is preserved approximately.
Spectral sparsifiers play a pivotal role in fast algorithms for solving Laplacian systems, a special case of linear systems.
\begin{definition}
	A \emph{$ (1 \pm \epsilon) $-spectral sparsifier} $ H $ of a graph $ G $ is a subgraph of $ G $ with weights $ \ww_H $ such that for every vector $ x \in \mathbb{R}^n $
	\begin{equation*}
	(1 - \epsilon) x^{\top} \lap_H x \leq x^{\top} \lap_G x \leq (1 + \epsilon) x^{\top} \lap_H x \, . \label{eq:spectral sparsifier inequality}
	\end{equation*}
\end{definition}
Using the Loewner ordering on matrices this condition can also be written as $ (1 - \epsilon) \lap_H \preceq \lap_G \preceq (1 + \epsilon) \lap_H $.
An $ n \times n $ matrix $ \mathcal{A} $ is \emph{positive semi-definite}, written as $ \mathcal{A} \succeq 0 $, if $ x^{\top} \mathcal{A} x \geq 0 $ for all $ x \in \mathbb{R}^n $.
For two $ n \times n $ matrices $ \mathcal{A} $ and $ \mathcal{B} $ we write $ \mathcal{A} \succeq \mathcal{B} $ as an abbreviation for $ \mathcal{A} - \mathcal{B} \succeq 0 $.

Note that $x^{\top} \lap_G x = \sum_{(u, v) \in E} w(u, v) (x(u) - x(v))^2$ where the vector $ x $ is treated as a function on the vertices and $ x(v) $ is the value of $ x $ for vertex $ v $.
A special case of such a function on the vertices is given by the binary indicator vector $ x_U $ associated with a cut $ U $, where $ x_U (v) = 1 $ is $ v \in U $ and $ 0 $ otherwise.
If limited to such indicator vectors, the sparsifier approximately preserves the value of every cut.
\begin{definition}
A \emph{$ (1 \pm \epsilon) $-cut sparsifier} $ H $ of a graph $ G $ is a subgraph of $ G $ with weights $ \ww_H $ such that for every subset $ U \subseteq V $
\begin{equation*}
	(1 - \epsilon) \ww_H (\partial_H (U)) \leq \ww_G (\partial_G (U)) \leq (1 + \epsilon) \ww_H (\partial_H (U)) \, . \label{eq:cut sparsifier inequality}
\end{equation*}
\end{definition}


\subsection{Sampling Schemes for Constructing Sparsifiers}

Most efficient constructions of sparsifiers are randomized,
partly because when $G$ is the complete graph, the resulting
sparsifier needs to be an expander.
These randomized schemes rely on importance sampling, which
for each edge:
\begin{enumerate}
	\item Keeps it with probability $p_e$,
	\item If the edge is kept, its weight is rescaled
		to $\frac{w_e}{p_e}$.
\end{enumerate}

A crucial property of this process is that the edge's
expectation is preserved.
As both cut and spectral sparsifiers can be viewed
as preserving sums over linear combinations of edge
weights, each of these terms have correct expectation.
The concentration of such processes can then be bounded
using either matrix concentration bounds in the spectral
case~\cite{Tropp12,SpielmanS11}, or a variety of combinatorial
arguments~\cite{BenczurK15}.

Our algorithms in this paper will use an even simpler
version of this importance sampling scheme: all of our
$p_e$'s will be set to either $1$ or $1/2$.
This scheme has a direct combinatorial interpretation:
\begin{enumerate}
	\item Keep some of the edges.
	\item Take a random half of the other edges,
		and double the weights of the edges kept.
\end{enumerate}		
Note that composing such a routine $O(\log{n})$ times gives a sparsifier,
as long as the part we keep is small.
So the main issue is to figure out how to get a small part to keep.

\subsection{Spanning Trees and Spanners}

A \emph{spanning forest} $ F $ of $ G $ is a forest (i.e., acyclic graph) on a subset of the edges of $ G $ such that every pair of vertices that is connected in $ G $ is also connected in $ F $.
A minimum/maximum spanning forest is a spanning forest of minimum/maximum total weight.

For every pair of vertices $ u $ and $ v $ we denote by $ \dist_G (u, v) $ the distance between $ u $ and $ v $ (i.e., the length of the shortest path connecting $ u $ and $ v $) in $ G $ with respect to the resistances.
The graph sparsification concept also exists with respect to distances in the graph.
Such sparse subgraphs that preserves distances approximately are called spanners.

\begin{definition}
A \emph{spanner of stretch $ \alpha $}, or short \emph{$ \alpha $-spanner}, (where $ \alpha \geq 1 $) of an undirected (possibly weighted) graph $ G $ is a subgraph $ H $ of $ G $ such that, for every pair of vertices $ u $ and $ v $, $ \dist_H (u, v) \leq \alpha \dist_G (u, v) $.
\end{definition}

\section{Overview and Related Work}
\label{sec:overview}

\subsection{Dynamic Spectral Sparsifier}\label{sec:overview spectral sparsifier}

We first develop a fully dynamic algorithm for maintaining a spectral sparsifier of a graph with polylogarithmic amortized update time.

\paragraph{Related Work.}

Spectral sparsifiers play important roles in fast numerical
algorithms~\cite{BatsonSST13}.
Spielman and Teng were the first to study these objects~\cite{SpielmanT11}.
Their algorithm constructs a $ (1 \pm \epsilon) $-spectral sparsifier of size $ O (n \cdot \poly(\log{n}, \epsilon^{-1})) $ in nearly linear time.
This result has seen several improvements in recent years~\cite{SpielmanS11,SilvaHS16,Zouzias12,ZhuLO15}.
The state of the art in the sequential model is an algorithm by Lee and Sun~\cite{LeeS15} that computes a $ (1 \pm \epsilon) $-spectral sparsifier of size $ O (n \epsilon^{-2}) $ in nearly linear time.
Most closely related to the data structural question are streaming
routines, both in one pass incremental~\cite{KelnerL13}, and turnstile~\cite{AhnGM13,KapralovW14,KapralovLMMS14}.

A survey of spectral sparsifier constructions is given in~\cite{BatsonSST13}.
Many of these methods rely on solving linear systems built on the graph, for which there approaches with a combinatorial flavor using low-stretch spanning trees~\cite{KelnerOSZ13,LeeS13} and purely numerical solvers relying on sparsifiers~\cite{PengS14} or recursive constructions~\cite{KyngLPSS16}.
We build on the spectral sparsifier obtained by a simple, combinatorial construction of Koutis~\cite{Koutis14}, which initially was geared towards parallel and distributed implementations.

\paragraph{Sparsification Framework.}
In our framework we determine `sampleable' edges by using spanners to compute a set of edges of bounded effective resistance.
From these edges we then sample by coin flipping to obtain a (moderately sparser) spectral sparsifier in which the number of edges has been reduced by a constant fraction.
This step can then be iterated a small number of times in order to compute the final sparsifier.

Concretely, we define a $t$-bundle spanner $ B = T_1 \cup \dots \cup T_t $ (for a suitable, polylogarithmic, value of~$t$) as a sequence of spanners $ T_1, \ldots, T_t $ where the edges of each spanner are removed from the graph before computing the next spanner, i.e., $ T_1 $ is a spanner of $ G $, $ T_2 $ is a spanner of $ G \setminus T_1 $, etc;
here each spanner has stretch $ O (\log{n}) $.
We then sample each non-bundle edge in $ G \setminus B $ with some constant probability $ p $ and scale the edge weights of the sampled edges proportionally.
The $t$-bundle spanner serves as a certificate for small resistance of the non-bundle edges in $ G \setminus B $ as it guarantees the presence of $t$ disjoint paths of length at most the stretch of the spanner.
Using this property one can apply matrix concentration bounds~\cite{Tropp12} to show the $t$-bundle together with the sampled edges is a moderately sparse spectral sparsifier.
We repeat this process of `peeling off' a $t$-bundle from the graph and sampling from the remaining edges until the graph is sparse enough (which happens after a logarithmic number of iterations).
Our final sparsifier consists of all $t$-bundles together with the sampled edges of the last stage.

\paragraph{Towards a Dynamic Algorithm.}
To implement
the spectral sparsification algorithm in the dynamic setting
we need to dynamically maintain a $t$-bundle spanner.
Our approach to this problem is to run $t$ different instances
of a dynamic spanner algorithm, in order to separately maintain
a spanner $ T_i $ for each graph
$ G_{i} = G \setminus \bigcup_{j=1}^{i-1} T_j$, for $ 1 \leq i \leq t $.

Baswana, Khurana, and Sarkar~\cite{BaswanaKS12} gave a fully dynamic algorithm for maintaining a spanner of stretch $ O (\log n) $ and size $ O (n \log^2 n) $ with polylogarithmic update time.\footnote{More precisely, they gave two fully dynamic algorithms for maintaing a $ (2k-1) $-spanner for any integer $ k \geq 2 $: The first algorithm guarantees a spanner of expected size $ O (k n^{1+1/k} \log{n}) $ and has expected amortized update time $ O (k^2 \log^2 n) $ and the second algorithm guarantees a spanner of expected size $ O (k^8 n^{1+1/k} \log^2{n}) $ and has expected amortized update time $ O (7^{k/2}) $.}
A natural first idea would be to use this algorithm in a black-box fashion in order to
separately maintain each spanner of a $t$-bundle.
However, we do not know how to do this because of the following obstacle.
A single update in $ G $ might lead to several changes of edges in the spanner $ T_1 $, an average of
$\Omega(\log n)$ according to the amortized upper bound.
This means that the next instance of the fully dynamic spanner algorithm which is used for maintaining $ T_2 $, not only has to deal with the deletion in $ G $ but also the artificially created updates in $G_2 = G \setminus T_1 $. This of course propagates to more updates in all graphs $G_i$.
Observe also that any given update in $ G_t $ caused by an update in $G$, can be requested {\em repeatedly}, as a result of subsequent updates in $G$.
Without further guarantees, it seems that with this approach we can only hope for an upper bound of $O (\log^{t-1} {n}) $ (on average) on the number of changes to be processed for updating $ G_t $ after a single update in $ G $.
That is too high because the sparsification algorithm
requires us to take $ t = \Omega (\log{n}) $. 
Our solution to this problem lies in a substantial modification
of the dynamic spanner algorithm in~\cite{BaswanaKS12}
outlined below.

\paragraph{Dynamic Spanners with Monotonicity.}
The spanner algorithm of Baswana et al.~\cite{BaswanaKS12} is at its
core a decremental algorithm (i.e., allowing only edge deletions in $ G $),
which is subsequently leveraged into a fully dynamic algorithm by a black-box reduction.
We follow the same approach by first designing a decremental algorithm
for maintaining a $t$-bundle spanner. This is achieved by modifying the
decremental spanner algorithm so that, in addition to its original guarantees,
it has the following \textbf{monotonicity} property:
\begin{center}
Every time an edge is added to the spanner $ T $, it stays in $ T $ until it is deleted from $ G $.
\end{center}

Recall that we initially want to maintain a $t$-bundle spanner $ T_1, \dots, T_t $ under edge deletions only.
In general, whenever an edge is added to $ T_1 $, it will cause its deletion from the graph $ G \setminus T_1 $ for which the spanner $ T_2 $ is maintained. Similarly, removing an edge from $ T_1 $ causes its insertion into $ G \setminus T_1 $, \emph{unless} the edge is deleted from $G$. This is precisely what the monotonicity property
guarantees: that an edge will not be removed from $T_1$ unless deleted from $G$. The
consequence is that no edge insertion can occur for $G_2 = G \setminus T_1$.
Inductively, no edge is ever inserted into  $ G_i $, for each $i$. Therefore
the algorithm for maintaining the spanner $ T_i $ only has to deal with edge deletions from the graph $G_i$, thus it becomes possible to run a different instance of the same decremental spanner algorithm for each $G_i$. A single deletion from $G$ can still generate many updates in the bundle. But for each~$i$, the instance of the dynamic spanner algorithm working on $G_i$ can only delete each edge {\em once}. Furthermore, we only run a small number $t$ of instances. So the total number of updates remains bounded, allowing us to claim the upper bound on the amortized update time.

In addition to the modification of the dynamic spanner algorithm, we have also deviated from Koutis' original scheme~\cite{Koutis14} in that we explicitly
`peel off' each iteration's bundle from the graph.
In this way we avoid that the $t$-bundles from different iterations share any edges, which seems hard to handle in the decremental setting we ultimately want to restrict ourselves to.

The modified spanner algorithm now allows us to maintain $t$-bundles in polylogarithmic update time, which is the main building block of the sparsifier algorithm.
The remaining parts of the algorithm, like sampling of the non-bundle edges by coin-flipping, can now be carried out in the straightforward way in polylogarithmic amortized update time.
At any time, our modified spanner algorithm can work in a purely decremental setting.
As mentioned above, the fully dynamic sparsifier algorithm is then obtained by a reduction from the decremental sparsifier algorithm.

\subsection{Dynamic Cut Sparsifier}

We then give dynamic algorithms for maintaining a $ (1 \pm \epsilon) $-cut sparsifier.
We obtain a fully dynamic algorithm with polylogarithmic worst-case update time by leveraging a recent worst-case update time algorithm for dynamically maintaining a spanning tree of a graph~\cite{KapronKM13}.
As mentioned above, spectral sparsifiers are more general than cut sparsifiers.
The big advantage of studying cut sparsification as a separate problem is that we can achieve polylogarithmic \emph{worst-case} update time, where the update time guarantee holds for each individual update and is \emph{not} amortized over a sequence of updates.

\paragraph{Related Work.}
In the static setting, Bencz{\'{u}}r and Karger~\cite{BenczurK15} developed an algorithm for computing a $ (1 \pm \epsilon) $-cut sparsifier of size $ O (n \cdot \poly(\log{n}, \epsilon^{-1})) $ in nearly linear time.
Their approach is to first compute a value called \emph{strength} for each edge and then sampling each edge with probability proportional to its strength.
Their proof uses a cut-counting argument that shows that the majority of cuts are
large, and therefore less likely to deviate from their expectation.
A union bound over these (highly skewed) probabilities then gives
the overall w.h.p. success bound.
This approach was refined by Fung et al.~\cite{FungHHP11} who show that a cut sparsifier can also be obtained by sampling each edge with probability inversely proportional to its (approximate) local edge connectivity, giving slightly better guarantees on the sparsifier.
The work of Kapron, King, and Mountjoy~\cite{KapronKM13} contains a fully dynamic approximate ``cut oracle'' with worst-case update time $ O (\log^2 n) $.
Given a set $ U \subseteq V $ as the input of a query, it returns a $2$-approximation to the number of edges in $ U \times V \setminus U $ in time $ O (|U| \log^2 n) $.
The cut sparsifier question has also been studied in the (dynamic) streaming model~\cite{AhnG09,AhnGM-SODA12,AhnGM-PODS12}.

\paragraph{Our Framework.}
The algorithm is based on the observation that the spectral sparsification scheme outlined above in \Cref{sec:overview spectral sparsifier}.
becomes a cut sparsification algorithm if we simply replace spanners by maximum weight spanning trees (MSTs).
This is inspired by sampling according to edge connectivities; the role of the MSTs is to certify lower bounds on the edge connectivities.
We observe that the framework does not require us to use exact MSTs.
For our $t$-bundles we can use a relaxed, approximate concept that we call $ \alpha $-MST that. Roughly speaking, an $ \alpha $-MST guarantees a `stretch' of $ \alpha $ in the infinity norm and, as long as it is sparse, does not necessarily have to be a tree.

Similarly to before, we define a $t$-bundle $\alpha$-MST $ B $ as the union of a sequence of $\alpha$-MSTs $ T_1, \ldots T_t $ where the edges of each tree are removed from the graph before computing the next $\alpha$-MST. The role of $\alpha$-MST is to certify uniform lower bounds on the connectivity of edges; these bounds are sufficiently large to allow uniform sampling with a fixed probability. 

This process of peeling and sampling is repeated sufficiently often and our cut sparsifier then is the union of all the $t$-bundle $\alpha$-MSTs and the non-bundle edges remaining after taking out the last bundle.
Thus, the cut sparsifier consists of a polylogarithmic number of $\alpha$-MSTs and a few (polylogarithmic) additional edges.
This means that for $\alpha$-MSTs based on spanning trees, our cut sparsifiers are not only sparse, but also have polylogarithmic \emph{arboricity}, which is the minimum number of forests into which a graph can be partitioned.

\paragraph{Simple Fully Dynamic Algorithm.}
Our approach immediately yields a fully dynamic algorithm by using a fully dynamic algorithm for maintaining a spanning forest.
Here we basically have two choices.
Either we use the randomized algorithm of Kapron, King, and Mountjoy~\cite{KapronKM13} with polylogarithmic worst-case update time.
Or we use the deterministic algorithm of Holm, de Lichtenberg, and Thorup~\cite{HolmLT01} with polylogarithmic amortized update time.
The latter algorithm is slightly faster, at the cost of providing only amortized update-time guarantees.
A $t$-bundle $2$-MST can be maintained fully dynamically by running, for each of the $ \log W $ weight classes of the graph, $t$ instances of the dynamic spanning tree algorithm in a `chain'.

An important observation about the spanning forest algorithm is that with every update in the graph, at most one edge is changed in the spanning forest:
If for example an edge is deleted from the spanning forest, it is replaced by another edge, but no other changes are added to the tree.
Therefore a single update in $G$ can only cause one update for each graph $ G_{i} = G \setminus \bigcup_{j=1}^{i-1} T_j$ and $T_i$.
This means that each instance of the spanning forest algorithm creates at most one `artificial' update that the next instance has to deal with.
In this way, each dynamic spanning forest instance used for the $t$-bundle has polylogarithmic update time.
As $ t = \polylog n $, the update time for maintaining a $t$-bundle is also polylogarithmic.
The remaining steps of the algorithm can be carried out dynamically in the straightforward way and overall give us polylogarithmic worst-case or amortized update time.

A technical detail of our algorithm is that the high-probability correctness achieved by the Chernoff bounds only holds for a polynomial number of updates in the graph.
We thus have to restart the algorithm periodically.
This is trivial when we are shooting for an amortized update time.
For a worst-case guarantee we can neither completely restart the algorithm nor change all edges of the sparsifier in one time step.
We therefore keep two instances of our algorithm that maintain two sparsifiers of two alternately growing and shrinking subgraphs that at any time partition the graph.
This allows us to take a blend of these two subgraph sparsifiers as our end result and take turns in periodically restarting the two instances of the algorithm.

\subsection{$(1 - \epsilon)$-Approximate Undirected Bipartite Flow}
\label{subsec:overviewApproxFlow}

We then study ways of utilizing our sparsifier constructions to give
routines with truly sublinear update times.
The problem that we work with will be maintaining an approximate
maximum flow problem on a bipartite graph $G_{A,B} = (A,B,E)$
with demand $-1$ and $1$ on each vertex in $A$ and $B$, respectively.
 All edges are unit weight and we dynamically insert and delete edges.
The maximum flow minimum cut theorem states that the objective here
equals to the minimum $s-t$ cut or maximum $s-t$
flow in $G$, which will be $G_{A,B}$ where we add vertices $s$ and $t$, and
connect each vertex in $A$ to $s$ and each vertex in $B$ to $t$.
The only dynamic changes in this graph will be in edges between $A$ and $B$.
As our algorithms builds upon cut sparsifiers,
and flow sparsifiers~\cite{KelnerLOS14} are more involved,
we will focus on only finding cuts.

This problem is motivated by the dynamic approximate
maximum matching problem, which differs in that the edges
are directed, and oriented from $A$ to $B$.
This problem has received much attention
recently~\cite{OnakR10,BaswanaGS15,NeimanS13,GuptaP13,PelegS16, BernsteinS16},
and led to the key definition of low arboricity graphs~\cite{NeimanS13,PelegS16}.
On the other hand, bipartite graphs are known to be difficult
to sparsify: the directed reachability matrix from $A$ to $B$
can encode $\Theta(n^2)$ bits of information.
As a result, we study the undirected variant of this problem instead,
with the hope that this framework can motivate other definitions
of sparsification suitable for wider classes of graphs.

Another related line of work are fully dynamic algorithm for maintaining the
global minimum cut~\cite{Thorup07,ThorupK00} with update time $ O (\sqrt{n} \polylog{n})$.
As there are significant differences between approximating global minimum cuts
and $st$-minimum cuts in the static setting~\cite{Karger00}, we believe that there are some
challenges to adapting these techniques for this problem.
The data structure by Thorup~\cite{Thorup07} can either maintain global edge connectivity up
to $\polylog{n}$ exactly or,
with high probability, arbitrary global edge connectivity with an approximation of $ 1 + o(1) $.
The algorithms also maintain concrete (approximate) minimum cuts, where in the
latter algorithm the update time increases to $ O (\sqrt{m} \polylog{n}) $
(and cut edges can be listed in time $ O (\log{n}) $ per edge).
Thorup's result was preceded by a randomized algorithm with worse approximation ratio for the global edge connectivity by Thorup and Karger~\cite{ThorupK00} with update time $ O (\sqrt{n} \polylog{n}) $.

At the start of Section~\ref{sec:dynamic min cut} we will show that the problem
we have formulated above is in fact different from matching.
On the other hand, our incorporation of sparsifiers for maintaining
solutions to this problem relies on several properties that hold in
a variety of other settings:
\begin{enumerate}
\item The static version can be efficiently approximated.
\item The objective can be approximated via graph sparsifiers.
\item A small answer (for which the algorithm's current approximation
may quickly become sub-optimal) means the graph also has a
small vertex cover.
\item The objective does not change much per each edge update.
\end{enumerate}

As with algorithms for maintaining high quality matchings~\cite{GuptaP13,PelegS16}, our approach aims to get a small amortized cost by keeping the same minimum $s-t$ cut for many consecutive dynamic steps.
Specifically, if we have a minimum $s-t$ cut of size $(2 + \frac{\epsilon}{2})\OPT$, then we know this cut will remain $(2 + \epsilon)$ approximately optimal for $\frac{\epsilon}{2} \OPT$ dynamic steps.
This allows us to only compute a new minimum $s-t$ cut every $\frac{\epsilon}{2} \OPT$ dynamic steps.

As checking for no edges would be an easy boundary case, we will assume throughout all the analysis that $\OPT > 0$.
To obtain an amortized $O(\poly(\log{n}, \epsilon^{-1}))$ update cost, it suffices for this computation to take $O(\OPT \cdot \poly(\log{n}, \epsilon^{-1}))$ time.
In other words, we need to solve approximate maximum flow on a graph
of size $O(\OPT \cdot \poly(\log{n}, \epsilon^{-1}))$.
Here we incorporate sparsifiers using the other crucial property used in
matching data structures~\cite{OnakR10,GuptaP13,PelegS16}: if $\OPT$ is small,
$G$ also has a small vertex cover.

\begin{restatable}[]{lemma}{smallCover}
	\label{lem:mvcOPT}
	The minimum vertex cover in $G$ has size at most $\OPT + 2$ where $\OPT$ is the size of the minimum $s-t$ cut in $G$.
\end{restatable}

We utilize the low arboricity of our sparsifiers to
find a small vertex cover with the additional property
that all non-cover vertices have small degree.
We will denote this (much) smaller set of vertices
as $VC$.
In a manner similar to eliminating vertices in numerical
algorithms~\cite{KyngLPSS16}, the graph can be reduced
to only edges on $VC$ at the cost of a
$(2 + \epsilon)$-approximation.
Maintaining a sparsifier of this routine again leads to an overall
routine that maintains a $(2 + \epsilon)$-approximation in
$\polylog{n}$ time per update, which we show in
Section~\ref{sec:dynamic min cut}.

Sparsifying vertices instead of edges inherently implies that an approximation of all cut values cannot be maintained. Instead, the sparsifier, which will be referred to as a \textit{terminal-cut-sparsifier}, maintains an approximation of all minimum cuts between any two terminal vertices, where the vertex cover is the terminal vertex set for our purposes. More specifically, given a minimum cut between two terminal vertices on the sparsified graph, by adding each independent vertex from the original graph to the cut set it is more connected to, an approximate minimum cut on the original graph is achieved. This concept of \textit{terminal-cut-sparsifier} will be equivalent to that in~\cite{AndoniGK14}, and will be given formal treatment in Section~\ref{sec:onePlusEpsilon}.

The large approximation ratio motivated us to reexamine the
sparsification routines, namely the one of reducing the
graph to one whose size is proportional to $\abs{VC}$.
This is directly related to the terminal cut sparsifiers studied
in~\cite{AndoniGK14,KoganK15}.
However, for an update time of $\poly(\log n, \epsilon^{-1})$, it is
crucial for the vertex sparsifier to have size $O(|\VC| \poly(\log{n}, \epsilon^{-1}))$.
As a result, instead of doing a direct union bound over
all $2^{|\VC|}$ cuts to get a size of $\poly(|\VC|)$ as in~\cite{AndoniGK14},
we need to invoke cut counting as with cut sparsifier constructions.
This necessitates the use of objects similar to $t$-bundles
to identify edges with small connectivity.
This leads to a sampling process motivated by the
$( 2 + \epsilon)$-approximate routine, but works on
vertices instead of edges.

By relating the processes, we are able to absorb the
factor $2$ error into the sparsifier size.
In Section~\ref{sec:vertSparsify}, we formalize
this process, as well as its guarantees on graphs
with bounded weights.
Here a major technical challenge compared to analyses of
cut sparsifiers~\cite{FungHHP11} is that
the natural scheme of bucketing by edge weights
is difficult to analyze because a sampled vertex could
have non-zero degree in multiple buckets.
We work around this issue via a pre-processing scheme
on $G$ that creates an approximation so that all vertices outside
of $\VC$ have degree $\polylog{n}$.
This scheme is motivated in part by the weighted expanders
constructions from~\cite{KyngLPSS16}.
Bucketing after this processing step ensures that each
vertex belongs to a unique bucket.
In terms of a static sparsifier on terminals, the result
 that is most comparable to results from previous works is:
\begin{restatable}[]{corollary}{terminalSparisfy}
	\label{cor:terminalSparsify}
	Given any graph $G = (V,E)$, and a vertex cover $\VC$ of $G$, where $X = V \setminus \VC$, with error $\epsilon$, we can build an $\epsilon$-approximate terminal-cut-sparsifier $H$ with $O(|\VC| \poly(\log{n},\epsilon^{-1}))$ vertices in $O(m\cdot \poly(\log{n},\epsilon^{-1}))$ work.

\end{restatable}


Turning this into a dynamic routine leads to the result
described in Theorem~\ref{thm:mainBipartite}:
a $(1 + \epsilon)$-approximate solution that can be
maintained in time $\polylog(n)$ per update. It is important to note that Theorem~\ref{thm:mainCut} plays an integral role in extending Corollary~\ref{cor:terminalSparsify} to a dynamic routine, particularly the low arboricity property that allows us to maintain a small vertex cover such that all non-cover vertices have low degree. These algorithmic extensions, as well as their incorporation
into data structures are discussed in Section~\ref{sec:onePlusEpsilon}.

\subsection{Discussion}

\paragraph{Graph Sparsification.}
We use a sparsification framework in which we `peel off' bundles of sparse subgraphs to determine `sampleable' edges, from which we then sample by coin flipping.
This leads to combinatorial and surprisingly straightforward algorithms for maintaining graph sparsifiers.
Additionally, this gives us low-arboricity sparsifiers; a property that we exploit for our main application.

Although spectral sparsification is more general than cut sparsification.
Our treatment of cut sparsification has two motivations.
First, we can obtain stronger running time guarantees.
Second, our sparsifier for the $(1 - \epsilon)$-approximate maximum flow algorithm on bipartite graphs hinges upon
improved routines for vertex sparsification,
a concept which leads to different objects in the spectral setting.

\paragraph{Dynamic Graph Algorithms.}
In our sparsification framework we sequentially remove bundles of sparse subgraphs to determine `sampleable' edges.
This leads to `chains' of dynamic algorithms where the output performed by one algorithm might result in updates to the input of the next algorithm.
This motivates a more fine-grained view on of dynamic algorithms with the goal of obtaining strong bounds on the number of changes to the output.

\paragraph{Future Work.}
The problem whether spectral sparsifiers can be maintained with polylogarithmic \emph{worst-case} update time remains open.
Our construction goes via spanners and therefore a natural question is whether spanners can be maintained with worst-case update time.
Maybe there are also other more direct ways of maintaining the sparsifier.
A more general question is whether we can find more dynamic algorithms for numerical problems.

Our dynamic algorithms cannot avoid storing the original graph, which is undesirable in terms of space consumption.
Can we get space-efficient dynamic algorithms without sacrificing fast update time?

The sparsification framework for peeling off subgraphs and uniformly sampling from the remaining edges is very general.
Are there other sparse subgraphs we could start with in the peeling process? Which properties do the sparsifiers obtained in this way have?
In particular, it would be interesting to see whether our techniques
can be generalized to flow sparsifiers~\cite{KelnerLOS14,AndoniGK14}.

The combination of sparsifiers with density-sensitive approaches for
dynamic graph data structures~\cite{NeimanS13,PelegS16} provides an
approach for obtaining $\poly(\log, \epsilon^{-1})$ update times.
We believe this approach can be generalized to other graph cut problems.
In particular, the flow networks solved for balanced cuts and graph
partitioning are also bipartite and undirected, and therefore natural
directions for future work.

\section{Dynamic Spectral Sparsifier}\label{sec:dynamic spectral sparsifier}

In this section we give an algorithm for maintaining a spectral sparsifier under edge deletions and insertions with polylogarithmic amortized update time. The main result of this section is as follows.

\begin{theorem}\label{thm:dynamic spectral sparsifier}
There exists a fully dynamic randomized algorithm with polylogarithmic update time
for maintaining a $(1 \pm \epsilon) $-spectral sparsifier $ H $
of a graph $G$, with probability at least $ 1 - 1/n^c $ for
any $0 < \epsilon \leq 1$ and $c \geq 1 $. Specifically, the amortized update time of the algorithm
is $$ O (c \epsilon^{-2} \log^3{\rho} \log^6{n}) $$
and the size of $H$ is $$ O (c n \epsilon^{-2} \log^3{\rho} \log^5{n} \log{W} + m \rho^{-1}) \, ,$$
where $ 1 \leq \rho \leq m $ is a parameter of choice.
Here, $ W $ is the ratio between the largest and the smallest edge weight in $ G $.
The ratio between the largest and the smallest edge weight in $ H $ is at most $ O (n W) $.
\end{theorem}

After giving an overview of our algorithm, we first explain our spectral sparsification scheme in a static setting and prove its properties.
Subsequently, we show how we can dynamically maintain the edges of such a sparsifier by making this scheme dynamic.

\subsection{Algorithm Overview}

\paragraph{Sparsification Framework.}
In our framework we determine `sampleable' edges by using spanners to compute a set of edges of bounded effective resistance.
From these edges we then sample by coin flipping to obtain a (moderately sparser) spectral sparsifier in which the number of edges has been reduced by a constant fraction.
This step can then be iterated a small number of times in order to compute the final sparsifier.

Concretely, we define a $t$-bundle spanner $ B = T_1 \cup \dots \cup T_t $ (for a suitable, polylogarithmic, value of~$t$) as a sequence of spanners $ T_1, \ldots, T_t $ where the edges of each spanner are removed from the graph before computing the next spanner, i.e., $ T_1 $ is a spanner of $ G $, $ T_2 $ is a spanner of $ G \setminus T_1 $, etc;
here each spanner has stretch $ O (\log{n}) $.
We then sample each non-bundle edge in $ G \setminus B $ with some constant probability $ p $ and scale the edge weights of the sampled edges proportionally.
The $t$-bundle spanner serves as a certificate for small resistance of the non-bundle edges in $ G \setminus B $ as it guarantees the presence of $t$ disjoint paths of length at most the stretch of the spanner.
Using this property one can apply matrix concentration bounds to show the $t$-bundle together with the sampled edges is a moderately sparse spectral sparsifier.
We repeat this process of `peeling off' a $t$-bundle from the graph and sampling from the remaining edges until the graph is sparse enough (which happens after a logarithmic number of iterations).
Our final sparsifier consists of all $t$-bundles together with the sampled edges of the last stage.

\paragraph{Towards a Dynamic Algorithm.}
To implement
the spectral sparsification algorithm in the dynamic setting
we need to dynamically maintain a $t$-bundle spanner.
Our approach to this problem is to run $t$ different instances
of a dynamic spanner algorithm, in order to separately maintain
a spanner $ T_i $ for each graph
$ G_{i} = G \setminus \bigcup_{j=1}^{i-1} T_j$, for $ 1 \leq i \leq t $.

Baswana, Khurana, and Sarkar~\cite{BaswanaKS12} gave a fully dynamic algorithm for maintaining a spanner of stretch $ O (\log n) $ and size $ O (n \log^2 n) $ with polylogarithmic update time.\footnote{More precisely, they gave two fully dynamic algorithms for maintaing a $ (2k-1) $-spanner for any integer $ k \geq 2 $: The first algorithm guarantees a spanner of expected size $ O (k n^{1+1/k} \log{n}) $ and has expected amortized update time $ O (k^2 \log^2 n) $ and the second algorithm guarantees a spanner of expected size $ O (k^8 n^{1+1/k} \log^2{n}) $ and has expected amortized update time $ O (7^{k/2}) $.}
A natural first idea would be to use this algorithm in a black-box fashion in order to
separately maintain each spanner of a $t$-bundle.
However, we do not know how to do this because of the following obstacle.
A single update in $ G $ might lead to several changes of edges in the spanner $ T_1 $, an average of
$\Omega(\log n)$ according to the amortized upper bound.
This means that the next instance of the fully dynamic spanner algorithm which is used for maintaining $ T_2 $, not only has to deal with the deletion in $ G $ but also the artificially created updates in $G_2 = G \setminus T_1 $. This of course propagates to more updates in all graphs $G_i$.
Observe also that any given update in $ G_t $ caused by an update in $G$, can be requested {\em repeatedly}, as a result of subsequent updates in $G$.
Without further guarantees, it seems that with this approach we can only hope for an upper bound of $O (\log^{t-1} {n}) $ (on average) on the number of changes to be processed for updating $ G_t $ after a single update in $ G $.
That is too high because the sparsification algorithm
requires us to take $ t = \Omega (\log{n}) $. 
Our solution to this problem lies in a substantial modification
of the dynamic spanner algorithm in~\cite{BaswanaKS12}
outlined below.

\paragraph{Dynamic Spanners with Monotonicity.}
The spanner algorithm of Baswana et al.~\cite{BaswanaKS12} is at its
core a decremental algorithm (i.e., allowing only edge deletions in $ G $),
which is subsequently leveraged into a fully dynamic algorithm by a black-box reduction.
We follow the same approach by first designing a decremental algorithm
for maintaining a $t$-bundle spanner. This is achieved by modifying the
decremental spanner algorithm so  so that, additional to its original guarantees,
it has the following \textbf{monotonicity} property:
\begin{center}
Every time an edge is added to the spanner $ T $, it stays in $ T $ until it is deleted from $ G $.
\end{center}

Recall that we initially want to maintain a $t$-bundle spanner $ T_1, \dots, T_t $ under edge deletions only.
In general, whenever an edge is added to $ T_1 $, it will cause its deletion from the graph $ G \setminus T_1 $ for which the spanner $ T_2 $ is maintained. Similarly, removing an edge from $ T_1 $ causes its insertion into $ G \setminus T_1 $, \emph{unless} the edge is deleted from $G$. This is precisely what the monotonicity property
guarantees: that an edge will not be removed from $T_1$ unless deleted from $G$. The
consequence is that no edge insertion can occur for $G_2 = G \setminus T_1$.
Inductively, no edge is ever inserted into  $ G_i $, for each $i$. Therefore
the algorithm for maintaining the spanner $ T_i $ only has to deal with edge deletions from the graph $G_i$, thus it becomes possible to run a different instance of the same decremental spanner algorithm for each $G_i$. A single deletion from $G$ can still generate many updates in the bundle. But for each $i$ the instance of the dynamic spanner algorithm working on $G_i$ can only delete each edge {\em once}. Furthermore, we only run a small number $t$ of instances. So the total number of updates remains bounded, allowing us to claim the upper bound on the amortized update time.

In addition to the modification of the dynamic spanner algorithm, we have also deviated from Koutis' original scheme~\cite{Koutis14} in that we explicitly
`peel off' each iteration's bundle from the graph.
In this way we avoid that the $t$-bundles from different iterations share any edges, which seems hard to handle in the decremental setting we ultimately want to restrict ourselves to.

The modified spanner algorithm now allows us to maintain $t$-bundles in polylogarithmic update time, which is the main building block of the sparsifier algorithm.
The remaining parts of the algorithm, like sampling of the non-bundle edges by coin-flipping, can now be carried out in the straightforward way in polylogarithmic amortized update time.
At any time, our modified spanner algorithm can work in a purely decremental setting.
As mentioned above, the fully dynamic sparsifier algorithm is then obtained by a reduction from the decremental sparsifier algorithm.

\subsection{Spectral Sparsification}\label{sec:spectral sparsifier general}

As outlined above, iteratively `peels off' bundles of spanners from the graph.

\begin{definition}
A \emph{$t$-bundle $ \alpha $-spanner} (where $ t \geq 1 $, $ \alpha \geq 1 $) of an undirected  graph $ G $ is the union $ T = \bigcup_{i=1}^k T_i $ of a sequence of graphs $ T_1, \dots, T_k $ such that, for every $ 1 \leq i \leq k $, $ T_i $ is an $ \alpha $-spanner of $ G \setminus \bigcup_{j=1}^{i-1} T_j $.
\end{definition}

The algorithm for spectral sparsification is presented in Figures \ref{fig:half spectral sparsify} and
\ref{fig:static spectral sparsifier}. Algorithm \textsc{Light-Spectral-Sparsify} computes a
moderately sparser $(1 \pm \epsilon)$-spectral sparsifier. Algorithm \textsc{Spectral-Sparsify} takes a parameter $ \rho $ and computes the sparsifier in $ k = \lceil \log{\rho} \rceil $ iterations of \textsc{Light-Spectral-Sparsify}.

\begin{figure} [p!]
	\begin{algbox}
	\textsc{Light-Spectral-Sparsify} $(G, \epsilon)$
		\begin{enumerate}
        \item $ t \gets \tSpectral $ for some absolute constant $c$.
        \item let $B = \bigcup_{j=1}^t T_j$ be a $t$-bundle $\alpha$-spanner of $ G $
        \item $H:=B$
        \item \textbf{for each} edge $ e \in G \setminus B $
           \begin{enumerate}
            	\item	with probability $1/4 $: add $ e $ to $ H $ with $ \ww_{H} (e) \gets 4 \ww_{G} (e) $
	  \end{enumerate}
        \item \textbf{return} $(H, B)$
		\end{enumerate}
	\end{algbox}
	\caption{\textsc{Light-Spectral-Sparsify} $(G, c, \epsilon)$. We give a dynamic implementation of this algorithm in \Cref{sec:dynamization light spectral sparsify}. In particular we dynamically maintain the $t$-bundle $\alpha$-spanner $ B $ which results in a dynamically changing graph $ G \setminus B $.}
	\label{fig:half spectral sparsify}
\end{figure}

\begin{figure} [p!]
	\begin{algbox}
	 \textsc{Spectral-Sparsify} $(G, c,\epsilon)$
		\begin{enumerate}
          \item $ k \gets \lceil \log{\rho} \rceil $
          \item $ G_0 \gets G $
          \item $ B_0 \gets (V,\emptyset) $
          \item \textbf{for} $i=1$ \textbf{to} $k$
          \begin{enumerate}
             \item $(H_i ,B_i) \gets \textsc{Light-Spectral-Sparsify} (G_{i-1}, c , \epsilon/(2k))$
             \item $ G_i \gets H_i \setminus B_i $
             \item \textbf{if} $G_i$ has less than $(c+1)\ln n$ edges \textbf{then} \textbf{break}  \qquad {\small (* break loop *) }
          \end{enumerate}
          \item $ H \gets \bigcup_{1 \leq j \leq i} B_j \cup G_i $
          \item \textbf{return} $( H , \{B_j\}_{j=1}^{i}, G_i)$
		\end{enumerate}
	\end{algbox}
	\caption{\textsc{Spectral-Sparsify} $(G, c,\epsilon)$. We give a dynamic implementation of this algorithm in \Cref{sec:dynamization spectral sparsify}. In particular we dynamically maintain each $ H_i $ and $ B_i $ as the result of a dynamic implementation of \textsc{Light-Spectral-Sparsify} which results in dynamically changing graphs $ G_i $.}
	\label{fig:static spectral sparsifier}
\end{figure}

We will now prove the properties of these algorithms. We first need the following lemma that shows how $t$-bundle spanners can be used to bound effective resistances. We highlight the main intuition of this crucial observation in our proof sketch.

\begin{lemma}[\cite{Koutis14}]\label{lem:stretch resistance correspondence}
Let $ G $ be a graph and $ B $ be a $t$-bundle $\alpha$-spanner of $ G $.
For every edge $ e $ of $ G \setminus B $, we have
\begin{equation*}
\ww_G (e) \cdot R_G (e) \leq \frac{\alpha}{t}
\end{equation*}
which implies that
\begin{equation*}
\ww_G (e) \cdot \lap_e \preceq \frac{\alpha}{t} \cdot \lap_G
\end{equation*}
where $ \lap_e $ is the $ n \times n $ Laplacian of the unweighted edge $ e $.
\end{lemma}

\begin{proof}[Sketch]
Fix some edge $ e = (u, v) $ of $ G \setminus B $ and let $ T_1, \ldots T_t $ denote the (pairwise disjoint) $\alpha$-spanners contained in $ B $.
For every $ 1 \leq i \leq t $, let $ \pi_i $ denote the shortest path from $ u $ to $ v $ in $ T_i $.
The length of the path $ \pi $ in $ T_i $ exceeds the distance from $ u $ to $ v $ in $ G \setminus \bigcup_{j=1}^{i-1} T_j $ by at most a factor of $ \alpha $ (property of the spanner $ T_i $).
Since $ e $ is contained in $ G \setminus B $, the latter distance is at most the resistance of the edge $ e $ as we have defined distances as the length of shortest paths with respect to the resistances of the edges.

Consider each path $ \pi_i $ as a subgraph of $ G $ and let $ \Pi $ be the subgraph consisting of all paths~$ \pi_i $.
Observe that $ \Pi $ consists of a parallel composition of paths, which in turn consists of a serial composition of edges, the we can view as resistors.
We can now apply the well-known rules for serial and parallel composition for computing effective resistances and get the desired bounds.
\end{proof}

Our second tool in the analysis the following variant~\cite{Harvey12} of a matrix concentration inequality by Tropp~\cite{Tropp12}.

\begin{theorem}\label{thm:matrix concentration}
Let $ Y_1, \dots, Y_k $ be independent positive semi-definite matrices of size $ n \times n $.
Let $ Y = \sum_{i=1}^k Y_i $ and $ Z = \expec{Y} $.
Suppose $ Y_i \preceq R Z $, where $ R $ is a scalar, for every $ 1 \leq i \leq k $.
Then for all $ \epsilon \in [0, 1] $
\begin{align*}
\prob{ \sum_{i=1}^k Y_i \preceq (1 - \epsilon) Z } &\leq n \cdot \exp (-\epsilon^2 / 2R) \\
\prob{ \sum_{i=1}^k Y_i \succeq (1 + \epsilon) Z } &\leq n \cdot \exp (-\epsilon^2 / 3R)
\end{align*}
\end{theorem}

Given these facts we can now prove the following Lemma which is a slight generalization
of a Lemma in~\cite{Koutis14}.
As the proof is quite standard we have moved it to Appendix~\ref{apx:proofs spectral sparsifier} (together with the proofs of the subsequent two lemmas).
For applying the lemma in our dynamic algorithm it is crucial that the input graph (which might be generated by another randomized algorithm) is independent of the random choices of algorithm \textsc{Light-Spectral-Sparsify}.

\begin{restatable}{lemma}{sparsifyingstep}\label{lem:sparsifying step}
The output $H$ of \textsc{Light-Spectral-Sparsify} is a $(1\pm \epsilon)$-spectral sparsifier with probability at least $1-n^{-(c+1)}$ for any input graph $ G $ that is independent of the random choices of the algorithm.
\end{restatable}

By iteratively applying the sparsification of \textsc{Light-Spectral-Sparsify} as done in \textsc{Spectral-Sparsify} we obtain sparser and sparser cut sparsifiers.

\begin{restatable}{lemma}{spectralsparsifiercorrectness}\label{lem:spectral sparsifier correctness}
The output $H$ of algorithm \textsc{Spectral-Sparsify} is a $(1\pm \epsilon)$-spectral
sparsifier with probability at least $ 1 - 1/n^{c+1} $ for any input graph $ G $ that is independent of the random choices of the algorithm.
\end{restatable}

\begin{restatable}{lemma}{spectralsparsifiersize}\label{lem:spectral sparsifier size}
With probability at least $1-2n^{-c}$, the number of iterations before algorithm \textsc{Spectral-Sparsify} terminates is
$$\min \{ \lceil \log \rho \rceil, \lceil \log m/((c+1)\log n) \rceil\}.$$
Moreover the size of $H$
is $$ O \left(\sum_{1 \leq j \leq i} | B_i | + m/\rho + c \log{n} \right) ,$$
and the size of the third output of the graph is
at most $\max\{ O(c\log n), O(m/\rho)\}$.
\end{restatable}

We conclude that with probability at least $1-n^{-c}$ our construction yields a $ (1 \pm \epsilon) $-spectral sparsifier that also has the properties of Lemma~\ref{lem:spectral sparsifier size}.

Typically, the $t$-bundle spanners will consist of a polylogarithmic number of spanners of size $ O (n \poly{\log{n}}) $ and thus the resulting spectral sparsifier will have size $ O (n \poly{\log{n}, \epsilon^{-1}} + m / \rho) $.
In each of the at most $ \log{n} $ iterations the weight of the sampled edges is increased by a factor of $ 4 $. Thus, the ratio between the largest and the smallest edge weight in $ H $ is at most by a factor of $ O (n) $ more than in $ G $, i.e., $ O (n W) $.

\subsection{Decremental Spanner with Monotonicity Property}\label{sec:decremental spanner}

We first develop the decremental spanner algorithm, which will give us a $ (\log{n}) $-spanner of size $ O (n \poly{(\log{n})}) $ with a total update time of $ O (m \poly{(\log{n})}) $.
Our algorithm is a careful modification of the dynamic spanner algorithm of Baswana et al.~\cite{BaswanaKS12} having the following additional monotonicity property: Every time an edge is added to $ H $, it stays in $ H $ until it is deleted from $ G $ by the adversary.
Formally, we will prove the following theorem.

\begin{lemma}\label{lem:decremental spanner}
For every $ k \geq 2 $ and every $ 0 < \epsilon \leq 1 $, there is a decremental algorithm for maintaining a $ (1 + \epsilon) (2k-1) $-spanner $ H $ of expected size $ O (k^2 n^{1 + 1/k} \log{n} \log_{1+\epsilon}{W}) $ for an undirected graph $ G $ with non-negative edge weights that has an expected total update time of $ O (k^2 m \log{n}) $, where $ W $ is the ratio between the largest and the smallest edge weight in $ G $.
Additional $ H $ has the following property: Every time an edge is added to $ H $, it stays in $ H $ until it is deleted from $ G $.
The bound on the expected size and the expected running time hold against an oblivious adversary.
\end{lemma}

It would be possible to enforce the monotonicity property for any dynamic spanner algorithm by simply overriding the algorithms' decision for removing edges from the spanner before they are deleted from $ G $.
Without additional arguments however, the algorithm's bound on the size of the spanner might then not hold anymore.
In particular, we do not know how obtain a version of the spanner of Baswana et al.\ that has the monotonicity property without modifying the internals of the algorithm.

Similar to Baswana et al.~\cite{BaswanaKS12} we actually develop an algorithm for unweighted graphs and then extend it to weighted graphs as follows.
Let $ W $ be the ratio of the largest to the smallest edge weight in~$ G $.
Partition the edges into $ \log_{1 + \epsilon} W $ subgraphs based on their weights and maintain a $(2k - 1) $-spanner ignoring the weights.
The union of these spanners will be a $ (1 + \epsilon) (2k - 1) $-spanner of $ G $ and the size increases by a factor of $ \log_{1 + \epsilon} W $ compared to the unweighted version.
The update time stays the same as each update in the graph is performed only in one of the $ \log_{1 + \epsilon} W $ subgraphs.
Therefore we assume in the following that $ G $ is an unweighted graph.

\subsubsection{Algorithm and Running Time}

We follow the approach of Baswana et al.\ and first explain how to maintain a clustering of the vertices and then define our spanner using this clustering.

\paragraph{Clustering.}

Consider an unweighted undirected graph $ G = (V, E) $ undergoing edge deletions.
Let $ S \subseteq V $ be a subset of the vertices used as cluster centers.
Furthermore, consider a permutation $ \sigma $ on the set of vertices $ V $ and an integer $ i \geq 0 $.

The goal is to maintain a clustering $ C_{S, \sigma, i} $ consisting of disjoint clusters with one cluster $ C_{S, \sigma, i} [s] \subseteq V $ for every $ s \in S $.
Every vertex within distance $ i $ to the vertices in $ S $ is assigned to the cluster of its closest vertex in $ S $, where ties are broken according to the permutation $ \sigma $.
More formally, $ v \in C_{S, \sigma, i} [s] $ if and only if
\begin{itemize}
	\item $ \dist_G (v, s) \leq i $ and
	\item for every $ s' \in S \setminus \{ s \} $ either
	\begin{itemize}
		\item $ \dist_G (v, s) < \dist_G (v, s') $ or
		\item $ \dist_G (v, s) = \dist_G (v, s') $ and $ \sigma (s) < \sigma' (s) $.
	\end{itemize}
\end{itemize}

Observe that each cluster $ C_{S, \sigma, i} [s] $ of a vertex $ s \in S $ can be organized as a tree consisting of shortest paths to $ s $.
We demand that in this tree every vertex $ v $ chooses the parent that comes first in the permutation $ \sigma $ among all candidates (i.e., among the vertices that are in the same cluster $ C_i [s] $ as $ v $ and that are at distance $ \dist (v, s) - 1 $ from $ s $).\footnote{Using the permutation to choose a random parent is not part of the original construction of Baswana et al.}
These trees of the clusters define a forest $ F_{S, \sigma, i} $ that we wish to maintain together with the clustering $ C_{S, \sigma, i} $.

Using a modification of the Even-Shiloach algorithm~\cite{EvenS81} all the cluster trees of the clustering $ C_i $ together can be maintained in total time $ O (i m \log{n}) $.

\begin{theorem}[\cite{BaswanaKS12}]\label{thm:maintaining clusters}
Given a graph $ G = (V, E) $, a set $ S \subseteq V $, a random permutation~$ \sigma $ of $ V $, and an integer $ i \geq 0 $, there is a decremental algorithm for maintaining the clustering $ C_{S, \sigma, i} $ and the corresponding forest $ F_{S, \sigma, i} $ of partial shortest path trees from the cluster centers in expected total time $ O (m i \log{n}) $.
\end{theorem}

Note that we deviate from the original algorithm of Baswana et al.\ by choosing the parent in the tree of each cluster according to the random permutation.
In the algorithm of Baswana et al.\ the parents in these trees were chosen arbitrarily.
However, it can easily be checked that running time guarantee of \Cref{thm:maintaining clusters} also holds for our modification.

The running time analysis of Baswana et al.\ hinges on the fact that the expected number of times a vertex changes its cluster is $ O (i \log{n}) $.

\begin{lemma}[\cite{BaswanaKS12}]\label{lem:number of cluster changes}
For every vertex $ v $ the expected number of times $ v $ changes its cluster in $ C_{S, \sigma, i} $ is at most $ O (i \log{n}) $.
\end{lemma}

By charging time $ O (\deg (v)) $ to every change of the cluster of $ v $ and every increase of the distance from $ v $ to $ S $ (which happens at most $i$ times), Baswana et al.\ get a total update time of $ O (i m \log{n}) $ over all deletions in $ G $.
For our version of the spanner that has the monotonicity property we additionally need the following observation whose proof is similar to the one of the lemma above.

\begin{lemma}\label{lem:number of parent changes}
For every vertex $ v $ the expected number of times $ v $ changes its parent in $ F_{S, \sigma, i} $ is at most $ O (i \log{n}) $.
\end{lemma}

\begin{proof}
Remember that we assume the adversary to be oblivious, which means that the sequence of deletions is independent of the random choices of our algorithm.
We divide the sequence of deletions into phases.
For every $ 1 \leq l \leq i $ the $l$-th phase consists of the (possibly empty) subsequence of deletions during which the distance from $ v $ to $ S $ is exactly $ l $, i.e., $ \dist_G (v, S) = l $.

Consider first the case $ l \geq 2 $.
We will argue about possible `configurations' $ (s, u) $ such that $ v $ is in the cluster of $ s $ and $ u $ is the parent of $ v $ that might occur in phase $ l $.
Let $ (s_1, u_1), (s_2, u_2), \dots, (s_{t^{(l)}}, u_{t^{(l)}}) $ (where $ t^{(l)} \leq n^2 $) be the sequence of all pairs of vertices such that, at the beginning of phase $ l $, for every $ 1 \leq j \leq t^{(l)} $, $ s_j $ is at distance $ l $ from $ v $ and $ u_j $ is a neighbor of $ v $.
The pairs $ (s_i, u_i) $ in this sequence are ordered according to the point in phase $ l $ at which they cease to be possible configurations, i.e., at which either the distance of $ s_i $ to $ v $ increases to more than $ l $ or $ u $ is not a neighbor of $ v $ anymore.

Let $ A_j^{(l)} $ denote the event that, at some point during phase $ l $, $ v $ is in the cluster of $ s_j $ and $ u_j $ is the parent of $ v $.
The expected number of times $ v $ changes its parent in $ F_{S, \sigma, i} $ during phase $ l $ is equal to the expected number of $j$'s such that event $ A_j^{(l)} $ takes place.
Let $ B_j^{(l)} $ denote the event that $ (s_j, u_j) $ is lexicographically first among all pairs $ (s_j, u_j), \ldots, (s_t, u_{t^{(l)}}) $ under the permutation $ \sigma $, i.e., for all $ j \leq j' \leq t^{(l)} $ either $ \sigma (s_j) \leq \sigma (s_{j'}) $ or $ \sigma (s_j) = \sigma (s_{j'}) $ and $ \sigma (u_j) , \sigma (u_{j'}) $.
Observe that $ \prob{ A_j^{(l)} } \leq \prob{ B_j^{(l)} } $ because the event $ A_j^{(l)} $ can only take place if the event $ B_j^{(l)} $ takes place.
Furthermore, $ \prob{ B_j^{(l)} } = 1 / (t^{(l)} - j + 1) $ as every pair of (distinct) vertices has the same probability of being first in the lexicographic order induced by $ \sigma $.
Thus, by linearity of expectation, the number of times $ v $ changes its parent in $ F_{S, \sigma, i} $ during phase $ l $ is at most
\begin{equation*}
\sum_{j=1}^{t^{(l)}} \prob{ A_j^{(l)} } \leq
	\sum_{j=1}^{t^{(l)}} \prob{ B_j^{(l)} } =
	\sum_{j=1}^{t^{(l)}} \frac{1}{t^{(l)} - j + 1} =
	\sum_{j=1}^{t^{(l)}} \frac{1}{j} =
	O (\log{t^{(l)}}) =
	O (\log{n}) \, .
\end{equation*}

In the second case $ l = 1 $, a slightly simpler argument bounds the number of times $ v $ changes its parent (which is equal to the number of times $ v $ changes its cluster) by ordering the neighbors of $ v $ in the order of deleting their edge to $ v $.
This is the original argument of Baswana et al.~\cite{BaswanaKS12} of \Cref{lem:number of cluster changes}.
We therefore also get that the number of times $ v $ changes its parent in $ F_{S, \sigma, i} $ in phase $ 1 $ is at most $ O (\log{n}) $.

We now sum up the expected number of changes during all phases, and, by linearity of expectation, get that the number of times $ v $ changes its parent in $ F_{S, \sigma, i} $ is at most $ O (i \log{n}) $.
\end{proof}

\paragraph{Spanner.}

Let $ 2 \leq k \leq \log{n} $ be a parameter of the algorithm.
At the initialization, we first create a sequence of sets $ V = S_0 \supseteq S_1 \supseteq \ldots \supseteq S_k = \emptyset $ by obtaining $ S_{i+1} $ from sampling each vertex of $ S_i $ with probability $ n^{-1/k} $.
Furthermore, we pick a random permutation $ \sigma $ of the vertices in $ V $.

We use the algorithm of \Cref{thm:maintaining clusters} to maintain, for every $ 1 \leq i \leq k $, the clustering $ C_i \defeq C_{S_i, \sigma} $ together with the forest $ F_i \defeq F_{S_i, \sigma} $.
Define the set $ V_i $ as $ V_i = \{ v \in V \mid d_G (v, S_i) \leq i \} $, i.e., the set of vertices that are at distance at most $ i $ to some vertex of $ S_i $.
Observe that the vertices in $ V_i $ are exactly those vertices that are contained in some cluster $ C_i [s] $ of the clustering~$ C_i $.
For every vertex $ v \in V_i $ (where $ C_i [s] $ is the cluster of $ v $) we say that a cluster $ C_i [s'] $ (for some $ s' \in S_i \setminus \{ s \} $) is \emph{neighboring} to $ v $ if $ G $ contains an edge $ (v, v') $ such that $ v' \in C_i [s'] $.

Our spanner $ H $ consists of the following two types of edges:
\begin{enumerate}
	\item For every $ 1 \leq i \leq k $, $ H $ contains all edges of the forest $ F_i $ consisting of partial shortest path trees from the cluster centers.
	\item For every $ 1 \leq i \leq k $, every vertex $ v \in V_i \setminus V_{i+1} $ (contained in some cluster $ C_i [s] $), and every neighboring cluster $ C_i [s'] $ of $ v $, $ H $ contains \emph{one} edge to $ C_i [s'] $, i.e., one edge $ (v, v') $ such that $ v' \in C_i [s'] $.
\end{enumerate}
The first type of edges can be maintained together with the spanning forests of the clustering algorithm of \Cref{thm:maintaining clusters}.
The second type of edges can be maintained with the following update rule:
Every time the clustering of a vertex $ v \in V_i \setminus V_{i+1} $ changes, we add to $ H $ \emph{one} edge to each neighboring cluster.
Every time such a `selected' edge is deleted from $ G $, we replace it with another edge to this neighboring cluster until all of them are used up.

We now enforce the monotonicity property mentioned above in the straightforward way.
Whenever we have added an edge to $ H $, we only remove it again from $ H $ when it is also deleted from~$ G $.
We argue below that this makes the size of the spanner only slightly worse than in the original construction of Baswana et al.

\subsubsection{Stretch and Size}

We now prove the guarantees on the stretch and size of $ H $.
The stretch argument is very similar to the ones of Baswana et al.
We include it here for completeness.
In the stretch argument we need stronger guarantees than Baswana et al.\ as we never remove edges from~$ H $, unless they are deleted from $ G $ as well.

\begin{lemma}[\cite{BaswanaKS12}]
$ H $ is a $ (2k-1) $-spanner of $ G $.
\end{lemma}

\begin{proof}
Consider any edge $ (u, v) $ of the current graph~$ G $ and the first $ j $ such that $ u $ and $ v $ are both contained in $ V_j $ and at least one of $ u $ or $ v $ is not contained in $ V_{j+1} $.
Without loss of generality assume that $ u \notin V_{j+1} $.
Since $ v \in V_j $, we know that $ v $ is contained in some cluster $ C_j [s] $ and because of the edge $ (u, v) $ this cluster is neighboring to $ u $.
Similarly, the cluster of $ u $ is neighboring to $ v $.
Consider the vertex out of $ u $ and $ v $ that has changed its cluster within $ C_i $ most recently (or take any of the two if both of them haven't changed their cluster since the initialization).
Assume without loss of generality that this vertex was $ u $.
Then $ C_i [s] $ has been a neighboring cluster of $ u $ at the time the cluster of~$ u $ changed, and thus, the spanner~$ H $ contains some edge $ (u, v') $ such that $ v' \in C_j [s] $.
Using the cluster tree of $ C_j [s] $ we find a path from $ v' $ to $ v $ via~$ s $ of length at most $ 2 i $ in $ H $.
Thus, $ H $ contains a path from $ u $ to $ v $ of length at most $ 2i + 1 \leq 2k - 1 $ as desired.
\end{proof}

\begin{lemma}
The number of edges of $ H $ is $ O(k^2 n^{1+1/k} \log{n}) $ in expectation.
\end{lemma}

\begin{proof}
Consider the first type of edges which are the ones stemming from the partial shortest path trees from the cluster centers.
We charge to each vertex $ v $ a total of $ O(k^2 \log{n}) $ edges given by all of $ v $'s parents in the partial shortest path trees from the cluster centers over the course of the algorithm.
For every $ 1 \leq i \leq k $, we know by \Cref{lem:number of parent changes} that the parent of $ v $ in $ F_i $ changes at most $ O (i \log{n}) $ times in expectation, which gives an overall bound of $ O(k^2 \log{n}) $.

We get the bound on the second type of edges by charging to each vertex~$ v $ a total of $ O (k^2 n^{1/k} \log{n}) $ edges.
Consider a vertex $ v \in V_i \setminus V_{i+1} $ for some $ 0 \leq i \leq k-1 $.
The number of neighboring clusters of $ v $ is equal to the number of vertices of $ S_i $ that are at distance exactly $ i+1 $ from $ v $.
Since $ v \notin V_{i+1} $ the number of such vertices is $ n^{1/k} $ in expectation.
Thus, whenever a vertex $ v \in V_i \setminus V_{i+1} $ changes its cluster in $ C_i $ we can charge $ n^{1/k} $ to $ v_i $ to pay for the $ n^{1/k} $ edges to neighboring clusters.
As $ v $ changes its cluster in $ C_i $ $ O (i \log{n}) $ times by \Cref{lem:number of cluster changes} and there are $ k $ clusterings, the total number of edges of the second type contained in $ H $ is $ O (k^2 n^{1 + 1/k} \log{n}) $.
Note that are allowed to multiply the two expectations because the random variables in question are independent.

The overall bound of $ O(k^2 n^{1+1/k} \log{n}) $ on the expected number of edges follows from the linearity of expectation.
\end{proof}

\subsection{Decremental Spectral Sparsifier}

In the following we explain how to obtain a decremental algorithm for maintaining a spectral sparsifier using the template of \Cref{sec:spectral sparsifier general}.
Internally we use our decremental spanner algorithm of \Cref{sec:decremental spanner}.
It is conceptually important for our approach to first develop a decremental algorithm, that is turned into a fully dynamic algorithm in \Cref{sec:decremental to fully dynamic}.
We follow the template of \Cref{sec:spectral sparsifier general} by first showing how to maintain $t$-bundle spanners under edge deletions, and then giving decremental implementations of \textsc{Light-Cut-Sparsify} and \textsc{Cut-Sparsify}.

The overall algorithm will use multiple instances of the dynamic spanner algorithm, where outputs of one instance will be used as the input of the next instance.
We will do so in a strictly hierarchical manner which means that we can order the instances in a way such that the output of instance $ i $ only affects instances $ i+1 $ and above.
In this way it is guaranteed that the updates made to instance $ i $ are independent of the internal random choices of instance $ i $, which means that each instance $ i $ is running in the oblivious-adversary setting required for \Cref{sec:decremental spanner}.

\subsubsection{Decremental $t$-Bundle Spanners}

We first show how to maintain a $t$-bundle $ \log{n} $-spanner under edge deletions for some parameter $ t $.
Using the decremental spanner algorithm of \Cref{lem:decremental spanner} with $ k = \lfloor (\log{n}) / 4 \rfloor $ and $ \epsilon = 1 $ we maintain a sequence $ H_1, \ldots H_t $ of $ \log{n} $-spanners by maintaining $ H_i $ as the spanner of $ G \setminus \bigcup_{1 \leq j \leq i-1} H_j $.
Here we have to argue that this is legal in the sense that every instance of the algorithm of \Cref{lem:decremental spanner} is run on a graph that only undergoes edge deletions.

\begin{lemma}\label{lem:decremental complement}
If no edges are ever inserted into $ G $ after the initialization, then this also holds for $ G \setminus \bigcup_{1 \leq j \leq i-1} H_j $ for every $ 1 \leq i \leq t+1 $.
\end{lemma}

\begin{proof}
The proof is by induction on $ i $.
The claim is trivially true for $ i = 1 $ by the assumption that there are only deletions in $ G $.
For $ i \geq 2 $ we the argument uses the monotonicity property of the dynamic algorithm for maintaining the spanner $ H_{i-1} $.
By the induction hypothesis we already know that no edges are ever added to the graph $ G \setminus \bigcup_{1 \leq j \leq i-2} H_j $.
Therefore the only possibility of an edge being added to $ G \setminus \bigcup_{1 \leq j \leq i-1} H_j $ would be to remove an edge $ e $ from $ H_{i-1} $.
However, by the monotonicity property, when $ e $ is removed from $ H_{i-1} $, it is also deleted from $ G $.
Thus, $ e $ will not be inserted into $ G \setminus \bigcup_{1 \leq j \leq i-2} H_j $.
\end{proof}

Our resulting $t$-bundle $ \log{n} $-spanner then is $ B = \bigcup_{1 \leq i \leq t} H_i $, the union of all these spanners.
Since the $ H_i's $ are disjoint the edges of $ B $ can be maintained in the obvious way by observing all changes to the $ H_i's $.
By our choice of parameters, $ n^{1/k} = O(1) $ and thus the expected size of $ B $ is $ O (t n \log^2{n} \log{W}) $.
Observe that \Cref{lem:decremental complement} implies that no edges will ever be inserted into the complement $ G \setminus B $, which will be relevant for our application in the spectral sparsifier algorithm.
We can summarize the guarantees of our decremental $t$-bundle spanner algorithm as follows.

\begin{lemma}\label{lem:decremental t-bundle spanner}
For every $ t \geq 1 $, there is a decremental algorithm for maintaining a $t$-bundle $ \log{n} $-spanner $ B $ of expected size $ O (t n \log^2{n} \log{W}) $ for an undirected graph $ G $ with non-negative edge weights that has an expected total update time of $ O (t m \log^3{n}) $, where $ W $ is the ratio between the largest and the smallest edge weight in $ G $.
Additional $ B $ has the following property: After the initialization, no edges are ever inserted into the graph $ G \setminus B $.
The bound on the expected size and the expected running time hold against an oblivious adversary.
\end{lemma}

\subsubsection{Dynamic Implementation of {\sc Light-Spectral-Sparsify}}\label{sec:dynamization light spectral sparsify}

We now show how to implement the algorithm \textsc{Light-Spectral-Sparsify} decrementally for a graph $ G $ undergoing edge deletions.

For this algorithm we set $ t = \lceil 12 (c+3) \alpha \epsilon^{-2} \ln{n} \rceil $. 
Note that this value is slightly larger than the one proposed in the static pseudocode of \Cref{fig:half spectral sparsify}.
For the sparsification proof in \Cref{sec:spectral sparsifier general} we have to argue that by our choice of $ t $ certain events happen with high probability.
In the dynamic algorithm we need ensure the correctness for up to $ n^2 $ versions of the graph, one version for each deletion in the graph.
By increasing the multiplicative constant in $ t $ by $ 2 $ (as compared to the static proof of \Cref{sec:spectral sparsifier general}) all desired events happen with high probability for all, up to $ n^2 $, versions of the graph by a union bound.

The first ingredient of the algorithm is to maintain a $t$-bundle $ \log{n} $-spanner $ B $ of $ G $ under edge deletions using the algorithm of \Cref{lem:decremental t-bundle spanner}.
We now explain how to maintain a graph $ H' $ -- with the intention that $ H' $ contains the sampled non-bundle edges of $ G \setminus B $ -- as follows:
At the initialization, we determine the graph $ H' $ by sampling each edge of $ G \setminus B $ with probability $ 1/4 $ and adding it to $ H' $ with weight $ 4 \ww_G (e) $.
We then maintain $ H' $ under the edge deletions in $ G $ using the following update rules:

After every deletion in $ G $ we first propagate the update to the algorithm for maintaining the $t$-bundle spanner $ B $, possibly changing $ B $ to react to the deletion.
We then check whether the deletion in $ G $ and the change in $ B $ cause an deletion in the complement graph $ G \setminus B $.
Whenever an edge $ e $ is deleted from $ G \setminus B $, it is removed from $ H' $.
Note that by \Cref{lem:decremental t-bundle spanner} no edge is ever inserted into $ G \setminus B $.
We now simply maintain the graph $ H $ as the union of $ B $ and $ H' $ and make it the first output of our algorithm; the second output is $ B $.

By the update rules above (and the increased value of $ t $ to accommodate for the increased number of events), this decremental algorithm imitates the static algorithm of \Cref{fig:half spectral sparsify} and for the resulting graph $ H $ we get the same guarantees as in \Cref{lem:sparsifying step}.
The total update time of our decremental version of \textsc{Light-Spectral-Sparsify} is $ O (t m \log^3{n}) $, as it is dominated by the time for maintaining the $t$-bundle $ \log{n} $-spanner $ B $.

As an additional property we get that no edge is ever added to the graph $ H' = H \setminus B $.
Furthermore, for all edges added to $ H' $ weights are always increased by the same factor.
Therefore the ratio between the largest and the smallest edge weight in $ H' $ will always be bounded by $ W $, which is the value of this quantity in $ G $ (before the first deletion).

\subsubsection{Dynamic Implementation of {\sc Spectral-Sparsify}}\label{sec:dynamization spectral sparsify}

Finally, we show how to implement the algorithm \textsc{Spectral-Sparsify} decrementally for a graph $ G $ undergoing edge deletions.

We set $ k = \lceil \log{\rho} \rceil $ as in the pseudocode of \Cref{fig:static spectral sparsifier} and maintain $ k $ instances of the dynamic version of \textsc{Light-Spectral-Sparsify} above.
We maintain the $ k $ graphs $ G_0, \ldots, G_k $, $ B_1, \ldots, B_k $, and $ H_1, \ldots, H_k $ as in the pseudocode.
For every $ 1 \leq i \leq k $ we maintain $ H_i $ and $ B_i $ as the two results of running the decremental version of \textsc{Light-Spectral-Sparsify} on $ G_{i-1} $ and maintain $ G_i $ as the graph $ H_i \setminus B_i $.
As argued above (for $ H' $ in \Cref{sec:dynamization spectral sparsify}), no edge is ever added to $ G_i = H_i \setminus B_i $ for every $ 1 \leq i \leq k $ and we can thus use our purely decremental implementation of \textsc{Light-Spectral-Sparsify}.

At the initialization, we additionally count the number of edges of every graph $ G_i $ and ignore every graph $ G_i $ with less than $ (c+1) \ln{n} $ edges.
Formally we set $ k $ maximal such that $ G_k $ has at least $ (c+1) \ln{n} $ edges.

The output of our algorithm is the graph $ H = \bigcup_{i=1}^k B_i \cup G_k $.
Now by the same arguments as for the static case, $ H $ gives the same guarantees as in \Cref{lem:spectral sparsifier correctness,lem:spectral sparsifier size}.
Thus, by our choices of $ k $ and $ t $, $ H $ is a $ (1 \pm \epsilon) $-spectral sparsifier of size $ O (c \epsilon^{-2} \log^3{\rho} \log^4{n} \log{W} + m \rho^{-1}) $.
As the total running time is dominated by the running time of the $ k $ instances of the decremental algorithm for \textsc{Light-Spectral-Sparsify}, the total update time is $ O (c m \epsilon^{-2} \log^3{\rho} \log^5{n}) $.
The guarantees of our decremental sparsifier algorithm can be summarized as follows.

\begin{lemma}\label{lem:decremental sparsifier}
For every $ 0 < \epsilon \leq 1 $, every $ 1 \leq \rho \leq m $, and every $ c \geq 1 $, there is a decremental algorithm for maintaining, with probability at least $ 1 - 1/n^c $ against an oblivious adversary, a $ (1 \pm \epsilon) $-spectral sparsifier $ H $ of size $ O (c \epsilon^{-2} \log^3{\rho} \log^4{n} \log{W} + m \rho^{-1}) $ for an undirected graph $ G $ with non-negative edge weights that has a total update time of $ O (c m \epsilon^{-2} \log^3{\rho} \log^5{n}) $, where $ W $ is the ratio between the largest and the smallest edge weight in $ G $.
\end{lemma}

\subsection{Turning Decremental Spectral Sparsifier into Fully Dynamic Spectral Sparsifier}\label{sec:decremental to fully dynamic}

We use a well-known reduction to turn our decremental algorithm into a fully dynamic algorithm.

\begin{lemma}\label{lem:decremental to fully dynamic}
Given a decremental algorithm for maintaining a $ (1 \pm \epsilon) $-spectral (cut) sparsifier of size $ S (m, n, W) $ for an undirected graph with total update time $ m \cdot T (m, n, W) $, there is a fully dynamic algorithm for maintaining a $ (1 \pm \epsilon) $-spectral (cut) sparsifier of size $ O (S (m, n, W) \log{n}) $ with amortized update time $ O (T (m, n, W) \log{n}) $.
\end{lemma}

Together with \Cref{lem:decremental sparsifier} this immediately implies \Cref{thm:dynamic spectral sparsifier}.
A similar reduction has been used by Baswana et al.~\cite{BaswanaKS12} to turn their decremental spanner algorithm into a fully dynamic one.
The only additional aspect we need is the lemma below on the decomposability of spectral sparsifiers.
We prove this property first and then give the reduction, which carries over almost literally from~\cite{BaswanaKS12}.

\begin{lemma}[Decomposability]\label{lem:decomposability spectral sparsifier}
Let $ G = (V, E) $ be an undirected weighted graph, let $ E_1, \dots, E_k $ be a partition of the set of edges $E$, and let, for every $ 1 \leq i \leq k $, $ H_i $ be a $ (1 \pm \epsilon) $-spectral sparsifier of $ G_i = (V, E_i) $.
Then $ H = \bigcup_{i=1}^k H_i $ is a $ (1 \pm \epsilon) $-spectral sparsifier of $ G $.
\end{lemma}

\begin{proof}
Because $H_i$ is a spectral sparsifier of $G_i$, for any vector $x$ and $i=1,\ldots,k$ we have
\begin{equation*}
 (1-\epsilon) x^T \lap_{H_i} x \leq x^T \lap_{G_i} x \leq (1+\epsilon) x^T \lap_{H_i} x
\end{equation*}
Summing these $k$ inequalities, we get that
\begin{equation*}
 (1-\epsilon) x^T \lap_{H} x \leq x^T \lap_{G} x \leq (1+\epsilon) x^T \lap_{H} x,
\end{equation*}
which by definition means that $H$ is a $ (1 \pm \epsilon) $-spectral sparsifier of $H$.
\end{proof}

\begin{proof}[Proof of \Cref{lem:decremental to fully dynamic}]
Set $ k = \lceil \log{(n^2)} \rceil $.
For each $ 1 \leq i \leq k $, we maintain a set $ E_i \subseteq E $ of edges and an instance $ A_i $ of the decremental algorithm running on the graph $ G_i = (V, E_i) $.
We also keep a binary counter $ C $ that counts the number of insertions modulo $ n^2 $ with the least significant bit in $ C $ being the right-most one.

A deletion of some edge $ e $ is carried out by simply deleting $ e $ from the set $ E_i $ it is contained in and propagating the deletion to instance $ A_i $ of the decremental algorithm.

An insertion of some edge $ e $ is carried out as follows.
Let $ j $ be the highest (i.e., left-most) bit that gets flipped in the counter when increasing the number of insertions.
Thus, in the updated counter the $j$-th bit is $ 1 $ and all lower bits (i.e., bits to the right of $ j $) are $ 0 $.
We first add the edge $ e $ as well as all edges in $ \bigcup_{i=1}^{j-1} E_i $ to $ E_j $.
Then we set $ E_i = \emptyset $ for all $ 1 \leq i \leq j-1 $.
Finally, we re-initialize the instance $ A_j $ on the new graph $ G_j = (V, E_j) $.

We know bound the total update time for each instance $ A_i $ of the decremental algorithm.
First, observe that the $i$-th bit of the binary counter is reset after every $ 2^i $ edge insertions.
A simple induction then shows that at any time $ E_i \leq 2^i $ for all $ 1 \leq i \leq k $.
Now consider an arbitrary sequence of updates of length $ \ell $.
The instance $ A_i $ is re-initialized after every $ 2^i $ insertions.
It will therefore be re-initialized at most $ \ell / 2^i $ times.
For every re-initialization we pay a total update time of $ |E_i| \cdot T (|E_i|, n, W) \leq 2^i T (m, n, W) $.
For the entire sequence of $ \ell $ updates, the total time spent for instance $ A_i $ is therefore $ (\ell / 2^i) \cdot 2^i T (m, n, W) = \ell \cdot T (m, n, W) $. Thus we spend total time $ O (\ell \cdot T (m, n, W) \log{n}) $ for the whole algorithm, which amounts to an amortized update time of $ O (T (m, n, W) \log{n}) $.
\end{proof}

\section{Dynamic Cut Sparsifier}\label{sec:dynamic cut sparsifier}

In this section we give an algorithm for maintaining a cut sparsifier under edge deletions and insertions with polylogarithmic worst-case update time. The main result of this section is as follows.

\begin{theorem}\label{thm:fully dynamic cut sparsifier}
There exists a fully dynamic randomized algorithm with polylogarithmic update time
for maintaining a $(1 \pm \epsilon) $-cut sparsifier $ H $
of a graph $G$, with probability at least $ 1 - n^{-c} $ for
any $0 < \epsilon \leq 1$ and $c \geq 1 $. Specifically, the algorithm either has worst-case update time
$$ O (c \epsilon^{-2} \log^2{\rho} \log^5{n} \log{W}) $$ or amortized update time
$$ O (c \epsilon^{-2} \log^2{\rho} \log^3{n} \log{W}) $$
and the size of $H$ is $$ O (c n \epsilon^{-2} \log^2{\rho} \log{n} \log{W} + m \rho^{-1}) \, ,$$
where $ 1 \leq \rho \leq m $ is a parameter of choice.
Here, $ W $ is the ratio between the largest and the smallest edge weight in $ G $.
The ratio between the largest and the smallest edge weight in $ H $ is at most $ O (n W) $.
\end{theorem}

By running the algorithm with basically $ \rho = m $ we additionally get that $ H $ has low arboricity, i.e., it can be partitioned into a polylogarithmic number of trees.
We will algorithmically exploit the low arboricity property in \Cref{sec:dynamic min cut,sec:onePlusEpsilon}.

\begin{corollary}\label{lem:additional properties of fully dynamic cut sparsifier}
There exists a fully dynamic randomized algorithm with polylogarithmic update time
for maintaining a $(1 \pm \epsilon) $-cut sparsifier $ H $
of a graph $G$, with probability at least $ 1 - n^{-c} $ for
any $0 < \epsilon \leq 1$ and $c \geq 1 $. Specifically, the algorithm either has worst-case update time
$ O (c \epsilon^{-2} \log^7{n} \log{W}) $ or amortized update time
$ O (c \epsilon^{-2} \log^5{n} \log{W}) $.
The arboricity of~$H$ is $ k = O (c \epsilon^{-2} \log^3{n} \log{W}) $.
Here, $ W $ is the ratio between the largest and the smallest edge weight in $ G $.
The ratio between the largest and the smallest edge weight in $ H $ is at most $ O (n W) $.
We can maintain a partition of $ H $ into disjoint forests $ T_1, \ldots, T_k $ such that every vertex keeps a list of its neighbors together with its degree in each forest $ T_i $.
After every update in $ G $ at most one edge is added to and at most one edge is removed from each forest~$ T_i $.
\end{corollary}

After giving an overview of our algorithm, we first explain our cut sparsification scheme in a static setting and prove its properties.
Subsequently, we show how we can dynamically maintain the edges of such a sparsifier with both amortized and worst-case update times by making this scheme dynamic.

\subsection{Algorithm Overview}

\paragraph{Our Framework.}
The algorithm is based on the observation that the spectral sparsification scheme outlined above in \Cref{sec:overview spectral sparsifier}.
becomes a cut sparsification algorithm if we simply replace spanners by maximum weight spanning trees (MSTs).
This is inspired by sampling according to edge connectivities; the role of the MSTs is to certify lower bounds on the edge connectivities.
We observe that the framework does not require us to use exact MSTs.
For our $t$-bundles we can use a relaxed, approximate concept that we call $ \alpha $-MST that. Roughly speaking, an $ \alpha $-MST guarantees a `stretch' of $ \alpha $ in the infinity norm and, as long as it is sparse, does not necessarily have to be a tree.

Similarly to before, we define a $t$-bundle $\alpha$-MST $ B $ as the union of a sequence of $\alpha$-MSTs $ T_1, \ldots T_t $ where the edges of each tree are removed from the graph before computing the next $\alpha$-MST. The role of $\alpha$-MST is to certify uniform lower bounds on the connectivity of edges; these bounds are sufficiently large to allow uniform sampling with a fixed probability. 

This process of peeling and sampling is repeated sufficiently often and our cut sparsifier then is the union of all the $t$-bundle $\alpha$-MSTs and the non-bundle edges remaining after taking out the last bundle.
Thus, the cut sparsifier consists of a polylogarithmic number of $\alpha$-MSTs and a few (polylogarithmic) additional edges.
This means that for $\alpha$-MSTs based on spanning trees, our cut sparsifiers are not only sparse, but also have polylogarithmic \emph{arboricity}, which is the minimum number of forests into which a graph can be partitioned.

\paragraph{Simple Fully Dynamic Algorithm.}
Our approach immediately yields a fully dynamic algorithm by using a fully dynamic algorithm for maintaining a spanning forest.
Here we basically have two choices.
Either we use the randomized algorithm of Kapron, King, and Mountjoy~\cite{KapronKM13} with polylogarithmic worst-case update time.
Or we use the deterministic algorithm of Holm, de Lichtenberg, and Thorup~\cite{HolmLT01} with polylogarithmic amortized update time.
The latter algorithm is slightly faster, at the cost of providing only amortized update-time guarantees.
A $t$-bundle $2$-MST can be maintained fully dynamically by running, for each of the $ \log W $ weight classes of the graph, $t$ instances of the dynamic spanning tree algorithm in a `chain'.

An important observation about the spanning forest algorithm is that with every update in the graph, at most one edge is changed in the spanning forest:
If for example an edge is deleted from the spanning forest, it is replaced by another edge, but no other changes are added to the tree.
Therefore a single update in $G$ can only cause one update for each graph $ G_{i} = G \setminus \bigcup_{j=1}^{i-1} T_j$ and $T_i$.
This means that each instance of the spanning forest algorithm creates at most one `artificial' update that the next instance has to deal with.
In this way, each dynamic spanning forest instance used for the $t$-bundle has polylogarithmic update time.
As $ t = \poly (\log n) $, the update time for maintaining a $t$-bundle is also polylogarithmic.
The remaining steps of the algorithm can be carried out dynamically in the straightforward way and overall give us polylogarithmic worst-case or amortized update time.

A technical detail of our algorithm is that the high-probability correctness achieved by the Chernoff bounds only holds for a polynomial number of updates in the graph.
We thus have to restart the algorithm periodically.
This is trivial when we are shooting for an amortized update time.
For a worst-case guarantee we can neither completely restart the algorithm nor change all edges of the sparsifier in one time step.
We therefore keep two instances of our algorithm that maintain two sparsifiers of two alternately growing and shrinking subgraphs that at any time partition the graph.
This allows us to take a blend of these two subgraph sparsifiers as our end result and take turns in periodically restarting the two instances of the algorithm.

\subsection{Definitions}

We will work with a relaxed notion of an MST, which will be useful when maintaining an exact maximum spanning tree is hard (as is the case for worst-case update time guarantees).

\begin{definition}
A subgraph $ T $ of an undirected graph $ G $ is an $ \alpha $-MST ($ \alpha \geq 1 $) if for every edge $ e = (u, v) $ of $ G $ there is a path $ \pi $ from $ u $ to $ v $ such that $ \ww_G (e) \leq \alpha \ww_G (f) $ for every edge $ f $ on $ \pi $.
\end{definition}
Note that in this definition we do not demand that $ T $ is a tree; any subgraph with these properties will be fine. A maximum spanning tree in this terminology is a $1$-MST.

\begin{definition}
A \emph{$t$-bundle $ \alpha $-MST} ($ t,\alpha \geq 1 $) of an undirected  graph $ G $ is the union $ B = \bigcup_{i=1}^k T_i $ of a sequence of graphs $ T_1, \dots, T_t $ such that, for every $ 1 \leq i \leq t $, $ T_i $ is an $ \alpha $-MST of $ G \setminus \bigcup_{j=1}^{i-1} T_j $.
\end{definition}

We can imagine such a $t$-bundle being obtained by iteratively peeling-off $ \alpha $-MSTs from~$ G $.

\subsection{A Simple Cut Sparsification Algorithm}\label{sec:cut sparsifier general}

We begin with algorithm \textsc{Light-Cut-Sparsify}
in Figure~\ref{fig:half cut sparsify}; this is
the core iteration used to compute a sparser
cut approximation with approximately half the edges.
Algorithm ~\textsc{Cut-Sparsify} in
Figure~\ref{fig:half cut sparsify}
is the full sparsification routine.

\begin{figure} [p!]
	\begin{algbox}
        \textsc{Light-Cut-Sparsify} $(G, c, \epsilon)$
		\begin{enumerate}
        \item $t \gets \tCut$
        \item Let $B$ be a $t$-bundle $\alpha$-MST of G
        \item $H:=B$
          \item \textbf{For each} edge $ e \in G \setminus B $
           \begin{enumerate}
            	\item	With probability $1/4$ add $ e $ to $ H $ with $ 4 \ww_{H} (e) \gets \ww_{G} (e) $
           \end{enumerate}
        \item \textbf{Return} $(H,B)$
		\end{enumerate}
	\end{algbox}
	\caption{\textsc{Light-Cut-Sparsify} $(G, c, \epsilon)$. We give a dynamic implementation of this algorithm in \Cref{sec:dynamization light cut sparsify}. In particular we dynamically maintain the $t$-bundle $\alpha$-MST $ B $ which results in a dynamically changing graph $ G \setminus B $.}
	\label{fig:half cut sparsify}
\end{figure}

\begin{figure} [p!]
	\begin{algbox}
          \textsc{Cut-Sparsify} $(G, c,\epsilon)$
		\begin{enumerate}
          \item $ k \gets \lceil \log{\rho} \rceil $
          \item $ G_0 \gets G $
          \item $ B_0 \gets (V,\emptyset) $
          \item \textbf{for} $i=1$ \textbf{to} $k$
          \begin{enumerate}
             \item $(H_i ,B_i) \gets \textsc{Light-Cut-Sparsify} (G, c+1 , \epsilon/(2k))$
             \item $ G_{i+1} \gets H_i \setminus B_i $
             \item \textbf{if} $G_{i+1}$ has less than $(c+2)\ln n$ edges \textbf{then} \textbf{break}  \qquad {\small (* break loop *) }
          \end{enumerate}
          \item $ H \gets \bigcup_{1 \leq j \leq i} B_j \cup G_{i+1} $
          \item \textbf{return} $( H , \{B_j\}_{j=1}^{i}, G_{i+1})$
		\end{enumerate}
	\end{algbox}
	\caption{\textsc{Cut-Sparsify} $(G, c,\epsilon)$ We give a dynamic implementation of this algorithm in \Cref{sec:dynamization cut sparsify}. In particular we dynamically maintain each $ H_i $ and $ B_i $ as the result of a dynamic implementation of \textsc{Light-Cut-Sparsify} which results in dynamically changing graphs $ G_i $.}
	\label{fig:alt cut sparsifier}
\end{figure}

The properties of these algorithm are given in the following lemmas.

\begin{lemma} \label{lem:half_cut}
The output $H$ of algorithm \textsc{Light-Cut-Sparsify} is a $(1\pm \epsilon)$-cut approximation of the input $G$, with probability $1-n^{-c}$.
\end{lemma}

We will need a slight generalization of a Theorem in~\cite{FungHHP11}.

\begin{lemma} (generalization of Theorem 1.1~\cite{FungHHP11}) \label{lem:connectivity_sampling}
Let $H$ be obtained from a graph $ G $ with weights in $(1/2, 1]$
by independently sampling edge edge e with probability $p_e \geq \rho/\lambda_G (e)$,
where $\rho= C_\xi c \log^2 n/ 4\epsilon^2$, and $\lambda_G (e)$
is the local edge connectivity of edge $e$, $C_\xi$ is an explicitly
known constant.  Then $H$ is a $(1\pm \epsilon)$-cut
sparsifier, with probability at least $1-n^{-c}$.
\end{lemma}

\begin{proof}
(Sketch) The generalization lies in introducing the parameter $c$
to control the probability of failure. This reflects the standard
behavior of Chernoff bounds: increasing the number
of samples by a factor of $c$ drives down the failure
probability by a factor of $n^{-c}$. Also,
the original theorem assumes that all edges are unweighted, but
a standard variant of the Chernoff bound can absorb constant ranges,
with a corresponding constant factor increase in the number
of samples. Finally, the original theorem is stated
with $p_e = \rho/\lambda_G (e)$, but all arguments remain identical
if this is relaxed to an inequality.
\end{proof}

\begin{proof}
 Suppose without loss of generality that the maximum weight in $G$ is $1$.
 We decompose~$G$ into $\log W$ edge-disjoint graphs, where $G_i$ consists of
 the edges with weights in $(2^{-(i+1)},2^{-i}]$ plus $B_i = B/2^{-(i+1)}$,
 where $B$ is the bundle returned by the algorithm.

 By definition of the $\alpha$-MST $t$-bundle, the connectivity
 of each edge of $G_i \setminus B_i$ in $G_i$ is at least $4 \rho c$,
 for $c = d\log W$ where $\rho$ is as defined in Lemma \ref{lem:connectivity_sampling}.
 Assume for a moment that all edges in $B_i$ are also in
 $(2^{-(i+1)},2^{-i}]$. Then we can set $p_e=1$ for each $e\in B_i$
 and $p_e=1/4$ for all other edges, and  apply Lemma~\ref{lem:connectivity_sampling}.
 In this way we get that $H_i$ is $(1\pm \epsilon)$-cut sparsifier
 with probability at least $1-n^{d \log W }$.

 The assumption about $B_i$ can be removed as follows. We observe that one can find
 a subgraph $B_i'$ of $B_i$ (by splitting weights when needed, and dropping
 smaller weights), such that $B_i'$ is a $t$-bundle $\alpha$-MST of $G_i$.
 This follows by the definition of the $t$-bundle $\alpha$-MST .
 We can thus apply the lemma on $G_i' = (G_i \setminus B_i) \cup B_i'$,
 and get that the sampled graph $H_i'$ is a $(1\pm \epsilon)$-cut sparsifier.
 We then observe that $G_i = G_i' \cup (B_i \setminus B_i')$ and
 $H_i = G_i' \cup (B_i \setminus B_i')$, from which it follows that
 $H_i$ is a $(1\pm \epsilon)$-cut sparsifier of $G_i$.
\end{proof}

\textbf{Note:} The number of logarithms in \textsc{Light-Cut-Sparsify} is not optimal. One can argue that the lower bounds we compute can be used in place of the {\em strong connectivities} used in~\cite{BenczurK15} and reduce by one the number of logarithms. It is also possible to replace $\log W$ with $\log n$ by carefully re-working some of the details  in~\cite{BenczurK15}.

We finally have the following Lemmas. The proofs are identical to those for the corresponding
Lemmas in \Cref{sec:dynamic spectral sparsifier}, so we omit them.

\begin{lemma}\label{lem:cut sparsifier correctness}
The output $H$ of algorithm \textsc{Cut-Sparsify} is a $(1\pm \epsilon)$-spectral
sparsifier of the input $G$, with probability at least $ 1 - 1/n^{c+1} $.
\end{lemma}

\begin{lemma}\label{lem:cut sparsifier size}
With probability at least $1-2n^{-c}$, the number of iterations before algorithm \textsc{Cut-Sparsify} terminates is
$$\min \{ \lceil \log \rho \rceil, \lceil \log m/((c+1)\log n) \rceil\}.$$
Moreover the size of $H$
is $$ O \left(\sum_{1 \leq j \leq i} | B_i | + m/\rho + c \log{n} \right) ,$$
and the size of the third output of the graph is
at most $\max\{ O(c\log n), O(m/\rho)\}$.
\end{lemma}

\subsection{Dynamic Cut Sparsifier}\label{sec:dynamic cut sparsifier}

We now explain how to implement the cut sparsifier algorithm of \Cref{sec:cut sparsifier general} dynamically.
The main building block of our algorithm is a fully dynamic algorithm for maintaining a spanning forest with polylogarithmic update time.
We either use an algorithm with worst-case update time, or a slightly faster algorithm with amortized update time.
In both algorithms, an insertion might join two subtrees of the forest and after a deletion the forest is repaired by trying to find a single replacement edge.
This strongly bounds the number of changes in the forest after each update.

\begin{theorem}[\cite{KapronKM13,GibbKKT15}]\label{thm:worst case connectivity}
	There is a fully dynamic deterministic algorithm for maintaining a spanning forest $ T $ of an undirected graph $ G $ with worst-case update time $ O (\log^4{n}) $.
	Every time an edge $ e $ is inserted into $ G $, the only potential change to $ T $ is the insertion of $ e $.
	Every time an edge $ e $ is deleted from $ G $, the only potential change to $ T $ is the removal of $ e $ and possibly the addition of at most one other edge to $ T $.
	The algorithm is correct with high probability against an oblivious adversary.
\end{theorem}

\begin{theorem}[\cite{HolmLT01}]\label{thm:amortized connectivity}
	There is a fully dynamic deterministic algorithm for maintaining a minimum spanning forest $ T $ of a weighted undirected graph $ G $ with amortized update time $ O (\log^2{n}) $.
	Every time an edge $ e $ is inserted into $ G $, the only potential change to $ T $ is the insertion of $ e $.
	Every time an edge $ e $ is deleted from $ G $, the only potential change to $ T $ is the removal of $ e $ and possibly the addition of at most one other edge to $ T $.
\end{theorem}

We first explain how to use these algorithms in a straightforward way to maintain a $2$-MST.
Subsequently we show how to dynamically implement the procedures \textsc{Light-Cut-Sparsify} and \textsc{Cut-Sparsify}.
The overall algorithm will use multiple instances of a dynamic spanning forest algorithm, where outputs of one instance will be used as the input of the next instance.
We will do so in a strictly hierarchical manner which means that we can order the instances in a way such that the output of instance $ i $ only affects instances $ i+1 $ and above.
In this way it is guaranteed that the updates made to instance $ i $ are independent of the internal random choices of instance $ i $, which means that each instance $ i $ is running in the oblivious-adversary setting required for \Cref{thm:worst case connectivity}.

\subsubsection{Dynamic Maintenance of $2$-MST}

For every $ 0 \leq i \leq \lfloor \log W \rfloor $, let $ E_i $ be the set of edges of weight between $ 2^i $ and $ 2^{i+1} $, i.e., $ E_i = \{ e \in E \mid 2^i \leq \ww_G (e) < 2^{i+1} \} $, and run a separate instance of the dynamic spanning forest algorithm for the edges in $ E_i $.
For every $ 0 \leq i \leq \lfloor \log W \rfloor $, let $ F_i $ be the spanning forest of the edges in $ E_i $ maintained by the $i$-th instance.
We claim that the union of all these trees is a $2$-MST of $ G $.

\begin{lemma}
$ T = \bigcup_{i=0}^{\lfloor \log W \rfloor} F_i $ is a $2$-MST of $ G $.
\end{lemma}

\begin{proof}
Consider some edge $ e = (u, v) $ of $ G $ and let $ i $ be the (unique) index such that $ 2^i \leq \ww_G (e) < 2^{i+1} $.
Since $ F_i $ is spanning tree of $ G $, there is a path $ \pi $ from $ u $ to $ v $ in $ F_i $ (and thus also in $ T $).
Every edge $ f $ of $ \pi $ is in the same weight class as $ e $, i.e., $ 2^i \leq \ww_G (f) < 2^{i+1} $.
Thus, $ \ww_G (e) < 2^{i+1} \leq 2 \ww_G (f) $ as desired.
\end{proof}

Every time an edge $ e $ is inserted or deleted, we determine the weight class $ i $ of $ e $ and perform the update in the $i$-th instance of the spanning forest algorithm.
This $2$-MST of size $ O (n \log{W}) $ can thus be maintained with the same asymptotic update time as the dynamic spanning forest algorithm.

We now show how to maintain a $t$-bundle $2$-MST and consequently a $ (1 \pm \epsilon) $-cut sparsifier $ H $ according to the construction presented in \Cref{sec:cut sparsifier general}.
For the $t$-bundle $2$-MST $ B = \bigcup_{1 \leq i \leq k} T_i $ we maintain, for every $ 1 \leq i \leq t $, a $2$-MST of $ G \setminus \bigcup_{j=1}^{i-1} T_j $.
We now analyze how changes to $ G \setminus \bigcup_{j=1}^{i-1} T_j $ affect $ G \setminus \bigcup_{j=1}^i T_j $ (for every $ 1 \leq i \leq k $):
\begin{itemize}
	\item Whenever an edge $ e $ is inserted into $ G \setminus \bigcup_{j=1}^{i-1} T_j $, the $2$-MST algorithm either adds $ e $ to $ T_i $ or not.
	\begin{itemize}
		\item If $ e $ is added to $ T_i $, then $ G \setminus \bigcup_{j=1}^i T_j $ does not change.
		\item If $ e $ is not added to $ T_i $, then $ e $ is added to $ G \setminus \bigcup_{j=1}^i T_j $.
	\end{itemize}
	\item Whenever an edge $ e $ is deleted from $ G \setminus \bigcup_{j=1}^{i-1} T_j $, either $ e $ is contained in $ T_i $ or not.
	\begin{itemize}
		\item If $ e $ is contained in $ T_i $, then $ e $ is removed from $ T_i $ and some other edge $ f $ is added to $ T_i $.
		This edge $ f $ is removed from $ G \setminus \bigcup_{j=1}^i T_j $.\footnote{The edge $ e $ will not be added to $ G \setminus \bigcup_{j=1}^i T_j $ because it is removed from both $ G \setminus \bigcup_{j=1}^{i-1} T_j $ and $ T_i $.}
		\item If $ e $ is not contained in $ T_i $, then $ e $ is removed from $ G \setminus \bigcup_{j=1}^i T_j $.
	\end{itemize}
\end{itemize}
Thus, every change to $ G \setminus \bigcup_{j=1}^{i-1} T_j $ results in at most one change to $ G \setminus \bigcup_{j=1}^i T_j $.
Consequently, a single update to~$ G $ results to at most one update in each instance of the dynamic MST algorithm.
For every update in $ G $ we therefore incur an amortized update time of $ O (t \log^4{n}) $.
Thus, we can summarize the guarantees for maintaining a $t$-bundle $2$-MST as follows.

\begin{corollary}\label{lem:maintaining 2MST}
There are fully dynamic algorithms for maintaining a $t$-bundle $2$-MST $ B $ (where $ t \geq 1 $ is an integer) of size $ O (t n \log{W}) $ with worst-case update time $ O (t \log^4{n}) $ or amortized update time $ O (t \log^2{n}) $, respectively.
After every update in $ G $, the graph $ G \setminus B $ changes by at most one edge.
\end{corollary}

\subsubsection{Dynamic Implementation of {\sc Light-Cut-Sparsify}}\label{sec:dynamization light cut sparsify}

For this algorithm we set $ t = (C_\xi + 2) d \alpha \epsilon^{-2} \log W \log^2 n $. 
Note that this value is slightly larger than the one proposed in \Cref{fig:half cut sparsify}.
For the sparsification proof in \Cref{sec:cut sparsifier general} we have to argue that by our choice of $ t $ certain events happen with high probability.
In the dynamic algorithm we need ensure the correctness for a polynomial number of versions of the graph, one version for each update made to the graph.
We show in \Cref{lem:arbitrarily long sequences} that it is sufficient to be correct for up to $ 4 n^2 $ updates to the graph, as then we can extend the algorithm to an arbitrarily long sequence of updates.
By making $ t $ slightly large than in the static proof of \Cref{sec:cut sparsifier general} all the desired events happen with high probability for all $ 4 n^2 $ versions of the graph by a union bound.

The first ingredient of the algorithm is to dynamically maintain a $t$-bundle $2$-MST $ B $ using the algorithm of \Cref{lem:maintaining 2MST} above.
We now explain how to maintain a graph $ H' $ -- with the intention that $ H' $ contains the sampled non-bundle edges of $ G \setminus B $ -- as follows:
After every update in $ G $ we first propagate the update to the algorithm for maintaining the $t$-bundle $2$-MST $ B $, possibly changing $ B $ to react to the update.
We then check whether the update in $ G $ and the change in $ B $ cause an update in the complement graph $ G \setminus B $.
\begin{itemize}
	\item Whenever an edge is inserted into $ G \setminus B $, it is added to $ H' $ with probability $ 1/4 $ and weight $ 4 \ww_G (e) $.
	\item Whenever an edge $ e $ is deleted from $ G \setminus B $, it is removed from $ H' $.
\end{itemize}
We now simply maintain the graph $ H $ as the union of $ B $ and $ H' $ and make it the first output of our algorithm; the second output is $ B $.

By the update rules above (and the increased value of $ t $ to accommodate for the increased number of events), this dynamic algorithm imitates the static algorithm of \Cref{fig:half cut sparsify} and for the resulting graph $ H $ we get the same guarantees as in \Cref{lem:half_cut}.
The update time of our dynamic version of \textsc{Light-Spectral-Sparsify} is $ O (t \log^4{n}) $ worst-case and $ O (t \log^2{n}) $ worst-case, as it is dominated by the time for maintaining the $t$-bundle $2$-MST $ B $.

As an additional property we get that with every update in $ G $ at most one change is performed to $ H' = H \setminus B $.
Furthermore, for all edges added to $ H' $ weights are always increased by the same factor.
Therefore the ratio between the largest and the smallest edge weight in $ H' $ will always be bounded by $ W $, which is the value of this quantity in $ G $ (before the first deletion).

\subsubsection{Dynamic Implementation of {\sc Cut-Sparsify}}\label{sec:dynamization cut sparsify}

We set $ k = \lceil \log{\min(\rho, m/((c+2) \log{n}))} \rceil $ and maintain $ k $ instances of the dynamic version of \textsc{Light-Cut-Sparsify} above, using the other parameters just like in the pseudo-code of \Cref{fig:alt cut sparsifier}.
By this choice of $ k $ we ensure that we do not have to check the breaking condition in the pseudo-code explicitly, which is more suited for a dynamic setting where the number of edges in the maintained subgraphs might grow and shrink.

We maintain the $ k $ graphs $ G_0, \ldots, G_k $, $ B_1, \ldots, B_k $, and $ H_1, \ldots, H_k $ as in the pseudocode.
For every $ 1 \leq i \leq k $ we maintain $ H_i $ and $ B_i $ as the two results of running the dynamic version of \textsc{Light-Cut-Sparsify} on $ G_{i-1} $ and maintain $ G_i $ as the graph $ H_i \setminus B_i $.

The output of our algorithm is the graph $ H = \bigcup_{i=1}^k B_i \cup G_k $.
Note that, by our choice of $ k $, $ G_k $ has at most $ \max (m / \rho, (c+2) \log{n}) $ edges.
Now by the same arguments as for the static case, $ H $ gives the same guarantees as in \Cref{lem:cut sparsifier correctness,lem:cut sparsifier size} for up to a polynomial number of updates (here at most $ 4 n^2 $) in the graph.

As argued above (for $ H' $ in \Cref{sec:dynamization light cut sparsify}), every update in $ G_{i-1} $ results in at most one change to $ G_i = H_i \setminus B_i $ for every $ 1 \leq i \leq k $.
By an inductive argument this means that every update in $ G $ results in at most one change to $ G_i $ for every $ 1 \leq i \leq k $.
As each instance of the dynamic \textsc{Light-Cut-Sparsify} algorithm has update time $ O (t \log^4{n}) $ worst-case or $ O (t \log^2{n}) $ amortized, this implies that our overall algorithm has update time $ O (k t \log^4{n}) $ or $ O (k t \log^2{n}) $, respectively.
Together with \Cref{lem:extending sparsifier to long sequence} in \Cref{sec:cut sparsifier general}, we have proved \Cref{thm:fully dynamic cut sparsifier} stated at the beginning of this section.

In \Cref{lem:additional properties of fully dynamic cut sparsifier} we additionally claim that for $ \rho = m $ we obtain a sparsifier with polylogarithmic arboricity.
This is true because the cut sparsifier $ H $ mainly consists of a collection of bundles, which in turn consists of a collection of trees.
In total, $ H $ consists of $ O (t k \log{W}) $ trees and $ O (c \log{n}) $ remaining edges in $ G_k $, each of which can be seen as a separate tree.
Furthermore we can maintain the collection of trees explicitly with appropriate data structures for storing them.

\subsection{Handling Arbitrarily Long Sequences of Updates}\label{lem:arbitrarily long sequences}

The high-probability guarantees of the algorithm above only holds for a polynomially bounded number of updates.
We now show how to extend it to an arbitrarily long sequence of updates providing the same asymptotic update time and size of the sparsifier.
We do this by concurrently running two instances of the dynamic algorithm that periodically take turns in being restarted, which is a fairly standard approach for such situations.
The only new aspect necessary for our purposes is that both instances explicitly maintain a sparsifier and when taking turns we cannot simply replace all the edges of one sparsifier with the edges of the other sparsifier as processing all these edges would violate the worst-case update time guarantee.
For this reason we exploit the decomposability of graph sparsifiers and maintain a `blend' of the two sparsifiers computed by the concurrent instances of the dynamic algorithm.
This step is not necessary for other dynamic problems such as connectivity where we only have to make sure that the query is delegated to the currently active instance. 

\begin{lemma}[Decomposability]\label{lem:decomposability cut sparsifier}
Let $ G = (V, E) $ be an undirected weighted graph, let $ E_1, \dots, E_k $ be a partition of the set of edges $E$, and let, for every $ 1 \leq i \leq k $, $ H_i $ be a $ (1 \pm \epsilon) $-cut sparsifier of $ G_i = (V, E_i) $.
Then $ H = \bigcup_{i=1}^k H_i $ is a $ (1 \pm \epsilon) $-cut sparsifier of $ G $.
\end{lemma}

\begin{proof}
Let $ U $ be a cut in $ G $.
First observe that
\begin{equation*}
\ww_G (\partial_G (U)) = \ww_G (\bigcup_{i=1}^k \partial_{G_i} (U)) = \sum_{i=1}^k \ww_G (\partial_{G_i} (U)) = \sum_{i=1}^k \ww_{G_i} (\partial_{G_i} (U))
\end{equation*}
and similarly $ \ww_H (\partial_H (U)) = \sum_{i=1}^k \ww_{H_i} (\partial_{H_i} (U)) $.
Now since
\begin{equation*}
	(1 - \epsilon) \ww_{H_i} (\partial_{H_i} (U)) \leq \ww_{G_i} (\partial_{G_i} (U)) \leq (1 + \epsilon) \ww_{H_i} (\partial_{H_i} (U))
\end{equation*}
for every $ 1 \leq i \leq k $, we have
\begin{multline*}
	(1 - \epsilon) \ww_H (\partial_H (U)) = (1 - \epsilon) \sum_{i=1}^k \ww_{H_i} (\partial_{H_i} (U)) \leq  \sum_{i=1}^k \ww_{G_i} (\partial_{G_i} (U)) = \ww_G (\partial_G (U)) \\ = \dots \leq (1 + \epsilon) \ww_H (\partial_H (U)) \, .
\end{multline*}
\end{proof}

\begin{lemma}\label{lem:extending sparsifier to long sequence}
Assume there is a fully dynamic algorithm for maintaining a $ (1 \pm \epsilon) $-cut (spectral) sparsifier of size at most $ S (m, n, W) $ with worst-case update time $ T (m, n, W) $ for up to $ 4 n^2 $ updates in $ G $.
Then there also is a fully dynamic algorithm for maintaining a $ (1 \pm \epsilon) $-cut (spectral) sparsifier of size at most $ O (S (m, n, W)) $ with worst-case update time $ O (T (m, n, W)) $ for an arbitrary number of updates.
\end{lemma}

\begin{proof}
We exploit the decomposability of cut sparsifiers.
We maintain a partition of $ G $ into two disjoint subgraphs $ G_1 $ and $ G_2 $ and run two instances $ A_1 $ and $ A_2 $ of the dynamic algorithm on $ G_1 $ and $ G_2 $, respectively.
These two algorithms maintain a $ (1 \pm \epsilon) $-sparsifier of $ H_1 $ of $ G_1 $ and a $ (1 \pm \epsilon) $-sparsifier $ H_2 $ of $ G_2 $.
By the decomposability stated in \Cref{lem:decomposability cut sparsifier,lem:decomposability spectral sparsifier}, the union $ H \defeq H_1 \cup H_2 $ is a $ (1 \pm \epsilon) $-sparsifier of $ G = G_1 \cup G_2 $.

We divide the sequence of updates into phases of length $ n^2 $ each.
In each phase of updates one of the two instances $ A_1 $, $ A_2 $ is in the state \emph{growing} and the other one is in the state \emph{shrinking}.
$ A_1 $ and $ A_2 $ switch their states at the end of each phase.
In the following we describe the algorithm's actions during one phase.
Assume without loss of generality that, in the phase we are fixing, $ A_1 $ is growing and $ A_2 $ is shrinking.

At the beginning of the phase we restart the growing instance $ A_1 $.
We will orchestrate the algorithm in such a way that at the beginning of the phase $ G_1 $ is the empty graph and $ G_2 = G $.
After every update in $ G $ we execute the following steps:
\begin{enumerate}
	\item If the update was the insertion of some edge $ e $, then $ e $ is added to the graph $ G_1 $ and this insertion is propagated to the \emph{growing} instance $ A_1 $.
	\item If the update was the deletion of some edge $ e $, then $ e $ is removed from the graph $ G_i $ it is contained in and this deletion is propagated to the corresponding instance $ A_i $.
	\item In addition to processing the update in $ G $, if $ G_2 $ is non-empty, then one arbitrary edge~$ e $ is first removed from $ G_2 $ and deleted from instance $ A_2 $ and then added to $ G_1 $ and inserted into instance $ A_1 $.
\end{enumerate}
Observe that these rules indeed guarantee that $ G_1 $ and $ G_2 $ are disjoint and together contain all edges of $ G $.
Furthermore, since the graph $ G_2 $ of the shrinking instance has at most $ n^2 $ edges at the beginning of the phase, the length of $ n^2 $ updates per phase guarantees that $ G_2 $ is empty at the end of the phase.
Thus, the growing instance always starts with an empty graph $ G_1 $.

As both $ H_1 $ and $ H_2 $ have size at most $ S (n, m, W) $, the size of $ H = H_1 \cup H_2 $ is $ O (S (n, m, W)) $.
With every update in $ G $ we perform at most $ 2 $ updates in each of $ A_1 $ and $ A_2 $.
It follows that the worst-case update time of our overall algorithm is $ O (T (m, n, W)) $.
Furthermore since each of the instances $ A_1 $ and $ A_2 $ is restarted every other phase, each instance of the dynamic algorithm sees at most $ 4 n^2 $ updates before it is restarted.
\end{proof}

\section{Application of Dynamic Cut Sparsifier: Undirected Bipartite Min-Cut}
\label{sec:dynamic min cut}

We now utilize our sparsifier data structure to  maintain a
$(2 + \epsilon)$-approximate $st$-min-cut in amortized
$O(\poly(\log{n}, \epsilon^{-1}))$ time per update.
In this section, we will define several tools that are crucial
for the better analyses in Sections~\ref{sec:vertSparsify}
and~\ref{sec:onePlusEpsilon}.

This result is a weaker form of Theorem~\ref{thm:mainBipartite}
with an approximation factor of $2 + \epsilon$ instead of $1 + \epsilon$.
The main result that we will show in this section is:

\begin{theorem}\label{thm:dynamic min cut approximation}
For every $ 0 < \epsilon \leq 1 $, there is a fully dynamic algorithm for
maintaining a $(2+\epsilon)$-approximate minimum cut in an unweighted
undirected graph that's a bipartite graph with source/sink $s$ and $t$ attached
to each of the partitions with amortized update time $ O(\poly(\log{n},\epsilon^{-1})) $.
\end{theorem}

To add motivation for solving this problem, we would like to point out that there are examples in which the maximum $s-t$ flow is much larger than the minimum vertex cover, and we cannot simply consider the problem as finding a maximum matching in $G_{A,B}$. Specifically, let $A = A_{k^2} \cup A_k$ and $B = B_{k^2} \cup B_k$, where $|A_k|,|B_k| = k$ and $|A_{k^2}|,|B_{k^2}| = k^2$, then construct a complete bipartite graph on $(A_{k^2},B_k), (A_k,B_k), (A_k,B_{k^2})$, while having no edges between $A_{k^2}$ and $B_{k^2}$. A vertex cover would be $A_k \cup B_k \cup \{s,t\}$, but we can achieve a max-flow in $G$ of $\Omega(k^2)$.

Accordingly, the objective cannot be approximated using matching routines even in the static case.
However,  the solution can still be approximated using recent developments in flow algorithms~\cite{Sherman13,KelnerLOS14,Peng16}.
Below we will show that these routines can be sped up on dynamic graphs using multiple layers of sparsification.
Specifically, the cut sparsifiers from Section~\ref{sec:dynamic cut sparsifier} allow us to
dynamically maintain a $(1 + \epsilon)$-approximation of the solution value,
as well as some form of query access to the minimum cut, in $O(\poly(\log{n},\epsilon^{-1}))$ per update.

The section is organized as follows.
Section~\ref{subsec:critcalCut} will give some of the high level ideas and critical observations on which our dynamic algorithm will hinge. Section~\ref{subsec:dyn2Algo} will present the dynamic algorithm for maintaining a $(2+\epsilon)$-approximate minimum $s-t$ cut, prove that the approximation factor is correct, and show that the dynamic update time is $ O(\poly(\log{n},\epsilon^{-1})) $ if we can dynamically update all data structures necessary for the algorithm in $O(\poly(\log{n},\epsilon^{-1}))$ time. Finally, Section~\ref{subsec:dataStructures} will present all of the necessary data structures and show how we can dynamically maintain them in $O(\poly(\log{n},\epsilon^{-1}))$ time.

\subsection{Key Observations and Definitions}
\label{subsec:critcalCut}

Our starting point is the observation that a small solution value
implies a small vertex cover.

\smallCover*

\begin{proof}
	Denote the minimum vertex cover as $\MVC$, and the minimum $s-t$ cut in $G$ as $(S,\bar{S})$ where $S = \{s\} \cup A_s \cup B_s$ and $\bar{S} = \{t\} \cup A_t \cup B_t$. Hence, we must have $\OPT \geq |A_t| + |B_s| + |E(A_s,B_t)|$ where $E(A_s,B_t)$ are all of the edges between $A_s$ and $B_t$.
	
	Let $V_A(A_s,B_t)$ denote all of the vertices in $A$ that are incident to an edge in $E(A_s,B_t)$, so $|V_A(A_s,B_t)| \leq |E(A_s,B_t)|$. We know $G_{A,B}$ is bipartite, so $A_t \cup B_s \cup V_A(A_s,B_t)$ must be a vertex cover in $G_{A,B}$, which implies $\abs{\MVC} \leq \OPT + 2$ by adding $s$ and $t$ to the cover.
	
\end{proof}

Our goal, for the rest of this section, is to show ways of
reducing the graph onto a small vertex cover, while
preserving the flow value.
The first issue that we encounter is that the minimum vertex
cover can also change during the updates.
However, in our case, the low arboricity property of the sparsifier
given in Corollary~\ref{lem:additional properties of fully dynamic cut sparsifier}
gives a more direct way of obtaining a small cover:

\begin{lemma}
	\label{lem:treeApprox}
	For any tree $T$, the vertex cover of all vertices other than the leaves is within a $2$-approximation of the minimum vertex cover.
	
\end{lemma}

This is proven in Appendix~\ref{sec:min_cut_proofs}.
We suspect that this is a folklore result, but it was difficult to find
a citation of it, as there exist far better algorithms for maintaining
vertex covers on dynamic trees~\cite{GuptaS09}.
Since there are at most $O(\poly(\log{n},\epsilon^{-1}))$ trees, and the overall
vertex cover needs to be at least the size of any cover in one of the trees,
we can set the cover as the set of all non-leaf vertices in the trees.

\begin{definition}
	Given a set of disjoint spanning forests $F = F_1 \cupdot \ldots \cupdot F_{K}$, we say that $\VC = \bigcup_{i \in [K]} \VC_i$ is a $\textbf{branch vertex cover}$ of $F$, if each $\VC_i$ is the set of all vertices other than the leaves in $F_i$
\end{definition}

\begin{corollary}
	\label{cor:lowdegInd}
	For any graph $G = (V,E)$ and corresponding sparsified graph $\tilde{G}= F_1 \cupdot \ldots \cupdot F_{K}$. If $\VC$ is a $\textbf{branch vertex cover}$ of $\tilde{G}$, then, $\VC$ is a $2K$-approximate vertex cover of $\tilde{G}$.
	Furthermore, any $x \in V \setminus \VC$ has degree at most $K$ in $\tilde{G}$
\end{corollary}

\begin{proof}
	Since the size of a minimum vertex cover in subgraph can only be smaller,
	we have
	\[
	|\MVC_{F_i}| \leq |\MVC_G|.
	\]
	Coupling this the choice of $|\VC|$ gives $|\VC_i| \leq 2 |\MVC_{F_i}|$,
	and summing over all $K$ trees gives the bound.
	The bound on the degree of $x$ follows from all leaves having degree $1$.
\end{proof}

We will ensure that $s$ and $t$ are placed in the cover,
and use $X$ to denote the non-cover vertices.
If we let the neighborhood of $x$ be $N(x)$, its
interaction with various partitions of $S$ can be described as:
\begin{definition}
	For a cut on $\VC$, $S \subseteq \VC$, and a non-cover vertex
	$x \in X$ with neighborhood $N(x)$, let
	\begin{enumerate}
		\item $\ww(x,S) \defeq \sum_{u \in S \cap N(x)} \ww (x,u)$,
		\item $\ww^{(x)}(S) \defeq \min\{\ww(x,S),\ww(x,\VC\setminus S)\}$.
	\end{enumerate}
\end{definition}

\begin{definition}
	Given a graph $G= (V,E)$ and some $\widehat{V} \subseteq V$ such that $\widehat{V}$ is a vertex cover of $G$, and $X = V \setminus \widehat{V}$
	
	\begin{enumerate}
		\item For any $S \subset V$, let $\Delta_G(S)$ be the weight of cut $S$ on $G$
		
		\item For any $S_{\widehat{V}} \subset \widehat{V}$, let the weight of a cut that is minimally extended from $S_{\widehat{V}}$ then be given by 
		
		\[
		\Delta_{G}(S_{\widehat{V}}) \defeq \Delta_{G \setminus X}(S_{\widehat{V}}) + \sum_{x \in X} \ww^{(x)}(S_{\widehat{V}}), 
		\]
		
	\end{enumerate}
	
\end{definition}

\begin{definition}
	Given $G= (V_G,E_G)$ and $H=(V_H,E_H)$ such that $V_H \subseteq V_G$ and $\widehat{V}$ is a vertex cover of both graphs
	\begin{enumerate}
		\item If $V_H = V_G$, then we say $H \approx_{\epsilon} G$ if for any $S \subset V_G$
		
		\[(1-\epsilon)\Delta_H(S) \leq \Delta_G(S) \leq (1 + \epsilon)\Delta_H(S) \]
		
		\item If $V_H \subset V_G$, then we say $H \approx_{\epsilon}^{\widehat{V}} G$, if for any $S_{\widehat{V}} \subset \widehat{V}$
		
		\[(1-\epsilon)\Delta_H(S_{\widehat{V}}) \leq \Delta_G(S_{\widehat{V}}) \leq (1 + \epsilon)\Delta_H(S_{\widehat{V}}) \]
		
	\end{enumerate}
	
\end{definition}

Note that if some $x \in X$ has degree $1$, it will always belong
to the same side as its neighbor in a minimum $s-t$ cut;
while if $x$ is incident to two neighbors $u$ and $v$, it will
always go with the neighbor with smaller weight.
That means that if $\ww(x, u) \leq \ww(x, v)$, then this is equivalent to 
an edge of weight $\ww(x, u)$ between $u$ and $v$.
This suggests that we can reduce the star out of $x$, $N_x$,
to a set of edges on its neighborhood.
We formalize the construction of this graph, $K_x$, as well
as the resulting graph by removing all of $X$ below:

\begin{definition}
	\label{def:kX}
	Given a weighted graph $G=(V,E)$ and $\ww (u,v) \rightarrow \mathbb{R}_{+}$,
	and any $S$, let:
	$K_x$ be the clique generated by running $\textsc{VertexElimination}$:
		for any two neighbors $u$ and $v$ of $x$, the edge weight of $(u,v)_x$ is 
		\[
		\frac{\ww (x,v) \ww (x,u)}{\sum_{i \in N(x)} \ww (x,i)}.
		\]		
	For some vertex cover $\VC$ and independent set $X = V \setminus \VC$, we let $G_{\VC} = \left( G \setminus X \right) \cup \bigcup_{x \in X} K_x$

\end{definition}

Note that we're using a subscript $_x$ to denote the origin of the edge.
Specifically, an edge $e_x \in G_{\VC}$ implies that $e_x \in K_x$,
and an edge $e_{\emptyset} \in G_{\VC}$ means it's from $\VC$, i.e.
$e_{\emptyset} \in G \setminus X$.
Note that $G_{\VC}$ also defines a weight for each cut $S_{\VC} \subset \VC$,
	where $\Delta_{G_{\VC}}(S_{\VC})$.
The crucial property of Definition~\ref{def:kX} is that it preserves
the values all cuts within a factor of $2$.
We prove the following in Appendix~\ref{sec:min_cut_proofs}.

\begin{theorem}
	\label{thm:schurComplement}
	Given a weighted graph $G=(V,E)$ and $\ww (u,v) \rightarrow \mathbb{R}_{+}$, with some vertex cover $\VC$ and independent set $X = V \setminus \VC$. For any $S_{VC} \subset \VC$
	
	\[\frac{1}{2}\Delta_{G}(S_{VC}) \leq \Delta_{G_{\VC}}(S_{VC}) \leq \Delta_{G}(S_{VC})\]
\end{theorem}

\begin{lemma}
	\label{lem:elimWeight}
	Given $G=(V,E)$ with all weights in $[\gamma,\gamma U]$, along with vertex cover $\VC$ and independent set $X$, such that any $x \in X$ has degree at most $d$. Then the weight of any edge in $G_{\VC}$ is in $[\gamma(dU)^{-1}, \gamma U]$
\end{lemma}

\subsection{Dynamic Algorithm for Maintaining a Minimum $s-t$ Cut on Bipartite Graphs}
\label{subsec:dyn2Algo}

Our algorithm can then be viewed as dynamically maintaining this cover
using two layers of dynamic graph sparsifiers intermixed with elimination routines.
Its main steps are shown in Figure~\ref{fig:dynamic2approx}.

\begin{figure}

\begin{algbox}

\begin{enumerate}
	\item Dynamically maintain a sparsified $G$, which we will denote $\tilde{G}$ 
	
	\item Dynamically maintain a $\textit{branch vertex cover}$, $\VC$, of $\tilde{G}$, where we ensure $s,t \in \VC$
	
	\item Dynamically maintain multi-graph $\tilde{G}_{\VC}$
	
	\item Dynamically maintain a sparsified $\tilde{G}_{\VC}$, which we will denote as $H$ with vertex set $V$
	
	\item Every $\frac{\epsilon}{2}\Delta_H(\widehat{S}_{\VC})$ dynamic steps, recompute $\widehat{S}_{\VC} \subset \VC$, an approximate minimum $s-t$ cut on $H$, ignoring all degree zero vertices \label{step:compute}
	
\end{enumerate}

\end{algbox}

	\caption{Dynamic $(2 + \epsilon)$-approximate Minimum $s-t$ Cut}
	\label{fig:dynamic2approx}

\end{figure}

One issue with maintaining a cut is that its two sides could have size $O(n)$,
which cannot be returned in amortized $O(\poly(\log{n},\epsilon^{-1}))$ time.
Instead, we will maintain the cut $\widehat{S}_{\VC} \subset \VC$ with $s \in \widehat{S}_{\VC}$, and allow querying of any vertex. For a vertex $ v \in \VC$, return $v$ is with $s$ iff $v \in \widehat{S}_{\VC}$, which takes $O(1)$ time. For a vertex $x \notin \VC$, return that $x$ is with $s$ iff $\ww(x,\widehat{S}_{\VC}) = \ww^{(x)}(\widehat{S}_{\VC})$ in $\tilde{G}$, taking $O(\poly(\log{n},\epsilon^{-1}))$ time to compute $\ww(x,\widehat{S}_{\VC})$ and $\ww(x,\VC \setminus \widehat{S}_{\VC})$. Specifically, the cut will be 	\[
\widehat{S} = \widehat{S}_{VC} \cup
\{x \in V \setminus \VC: \ww(x,\widehat{S}_{VC}) = \ww^{(x)}(\widehat{S}_{VC}) \},
\] the extension of $\widehat{S}_{\VC}$ on $\tilde{G}$ which allows for the $O(\poly(\log{n},\epsilon^{-1}))$ query computation by Corollary~\ref{lem:additional properties of fully dynamic cut sparsifier} and Corollary~\ref{cor:lowdegInd}.

We first establish the quality of this cut on $H$ that we maintain:

\newcommand{\epsilonhat}{\widehat{\epsilon}}

\begin{theorem}
	\label{thm:approxDyn2} Computing a $(1 + \epsilonhat)$-approximate
	minimum $s-t$ cut in $H$ as in Step~\ref{step:compute} of
	Figure~\ref{fig:dynamic2approx} takes $O(\OPT\cdot \poly(\log{n},\epsilon^{-1}))$ time for $\widehat{\epsilon} = \frac{\epsilon}{O(1)}$,
	and cut $\widehat{S}_{\VC} \subset \VC$ can be extended to $\widehat{S}$ a $2(1 + \epsilonhat)^5$-approximate
	minimum $s-t$ cut in $G$ with high probability
	
\end{theorem}

\begin{proof}
	$\tilde{G} = F_1 \cupdot \ldots \cupdot F_K$ for some $K = O(\poly(\log{n},\epsilon^{-1}))$ by Corollary~\ref{lem:additional properties of fully dynamic cut sparsifier}, so from Lemma~\ref{lem:mvcOPT} and Corollary~\ref{cor:lowdegInd}, we know $|\VC| = O(\OPT\cdot \poly(\log{n},\epsilon^{-1}))$.
	From Corollary~\ref{lem:additional properties of fully dynamic cut sparsifier}, the weights of $\tilde{G}$ are in $[1,O(n)]$, and Lemma~\ref{lem:elimWeight} implies that the weights of $\tilde{G}_{\VC}$ are in $[O(n^{-1}\poly(\log{n},\epsilon^{-1}))^{-1},O(n)]$.
	Further, each $K_x$ of $\tilde{G}_{\VC}$ has at most $K^2 = O(\poly(\log{n},\epsilon^{-1}))$ edges, so $\tilde{G}_{\VC}$ has $O(n \cdot \poly(\log{n},\epsilon^{-1}))$ edges.
	Corollary~\ref{lem:additional properties of fully dynamic cut sparsifier} then tells us that $H$ has $O(|\VC|\cdot \poly(\log{n}, \epsilon^{-1})) = O(\OPT\cdot \poly(\log{n}, \epsilon^{-1}))$ edges, and that we can find
	a $(1 + \epsilonhat)$ approximate minimum $s-t$ cut in $H$, $\widehat{S}_{VC}$
	in $O(\OPT\cdot \poly(\log{n}, \epsilon^{-1}))$ time.

	From Corollary~\ref{lem:additional properties of fully dynamic cut sparsifier}, we assume that $H \approx_{\epsilonhat} \tilde{G}_{\VC}$ and $G \approx_{\epsilonhat} \tilde{G}$ with high probability.
	
	
	Suppose $\widehat{S}_{VC} \subset \VC$ is returned as a $(1 + \epsilonhat)$-approximate minimum $s-t$ cut in $H$, and let
	\[
	\widehat{S} = \widehat{S}_{VC} \cup
		\{x \in V \setminus \VC: \ww(x,\widehat{S}_{VC}) = \ww^{(x)}(\widehat{S}_{VC}) \}
	\]
	be its extension onto $\tilde{G}$.
	The left-hand side of Theorem~\ref{thm:schurComplement} implies
	\[
		\Delta_{\tilde{G}}(\widehat{S}_{VC})
			\leq 2 \Delta_{\tilde{G}_{VC}}(\widehat{S}_{VC}),
	\]
	which along with the approximations $G \approx_{\epsilonhat} \tilde{G}$
	and $\tilde{G}_{VC} \approx_{\epsilonhat} H$ 	gives
	\[
	\Delta_{G}(\widehat{S}) \leq (1 + \epsilonhat )\Delta_{\tilde{G}}(\widehat{S})
		\leq 2  (1 + \epsilonhat) \Delta_{\tilde{G}_{\VC}}(\widehat{S}_{VC})
		\leq 2 (1 + \epsilonhat )^2 \Delta_{H}(\widehat{S}_{VC}).
	\]

	On the other hand, let $\overline{S} \subset V$ be the minimum
	$s-t$ cut in $G$, and $\overline{S}_{VC} \subset \VC$
	be its restriction to $\VC$.
	Since right-hand side of Theorem~\ref{thm:schurComplement} is over optimum choices
	of $V \setminus S_{VC}$, we have
	\[
		\Delta_{\tilde{G}}(\overline{S})
			\geq \Delta_{\tilde{G}}(\overline{S}_{VC})
				\geq \Delta_{\tilde{G}_{VC}}(\overline{S}_{VC}),
	\]
	which when combined with the approximations $G \approx_{\epsilonhat} \tilde{G}$
	and $\tilde{G}_{VC} \approx_{\epsilonhat} H$ 	gives
	\[
		\Delta_{G}(\overline{S}) \geq (1 - \epsilonhat ) \Delta_{\tilde{G}}(\overline{S})
		\geq (1 - \epsilonhat ) \Delta_{\tilde{G}_{\VC}}(\overline{S}_{VC})
		\geq (1 - \epsilonhat )^2 \Delta_{H}(\overline{S}_{VC}).
	\]	
	The result then follows from the near-optimality of $\widehat{S}_{VC}$
	on $H$, $ \Delta_{H}(\overline{S}_{VC}) \geq
		(1 - \epsilonhat) \Delta_{H}(\widehat{S}_{VC})$.
		
%
%
%
%
%
%
%

\end{proof}

\begin{corollary}
	\label{cor:dyn2Correct}
	The dynamic algorithm maintains a $(2 + \epsilon)$-approximate minimum $s-t$ cut in $G$, and will only compute an approximate minimum $s-t$ cut on $H$ every $O({\epsilon}\OPT)$ dynamic steps.
	
\end{corollary}

\begin{proof}
	Choosing $\epsilonhat = \frac{\epsilon}{O(1)}$ in Theorem~\ref{thm:approxDyn2} can give a $(2 + \frac{\epsilon}{2})$-approximate minimum $s-t$ cut in $G$. Borrowing notation from the proof of Theorem~\ref{thm:approxDyn2}, an approximate minimum $s-t$ cut on $H$ will be re-computed in $\frac{\epsilon}{2}\Delta_H(\widehat{S}_{\VC})$ dynamic steps. $\OPT = \Delta_G(\bar{S}) \leq \Delta_G(\widehat{S}) \leq 2(1 + \epsilonhat)^2\Delta_H(\widehat{S}_{\VC})$, so $\Delta_H(\widehat{S}_{\VC}) = O(\OPT)$
	
\end{proof}

\subsection{Dynamically Updating Data Structures}
	\label{subsec:dataStructures}

As was shown in Corollary~\ref{cor:dyn2Correct}, the dynamic algorithm maintains a $(2 + \epsilon)$-approximate minimum $s-t$ cut of $G$, an approximate minimum $s-t$ cut of $H$ is computed every $O(\epsilon \OPT)$, and that computation takes $O(\OPT\cdot \poly(\log{n}, \epsilon^{-1}))$ time from Theorem~\ref{thm:approxDyn2}. Therefore, in order to establish that the amortized dynamic update time is $O( \poly(\log{n}, \epsilon^{-1}))$, it suffices to show that all data structures can be maintained in $O(\poly(\log{n},\epsilon^{-1}))$ time per dynamic update, thereby finishing the proof of Theorem~\ref{thm:dynamic min cut approximation}. As a result of Corollary~\ref{lem:additional properties of fully dynamic cut sparsifier}, it suffices to show the following

\begin{theorem}
	\label{thm:dyn2Maintenance}
	For each addition/deletion of an edge in $\tilde{G}$, data structures for $\tilde{G}$, $\VC$, $\tilde{G}_{\VC}$, and $H$ can be maintained in $O(\poly(\log{n},\epsilon^{-1}))$ time.
	
\end{theorem}

Bounds on the dynamic update time of each data structure will all ultimately follow from the $O(\poly(\log{n},\epsilon^{-1}))$ degree bound for $\tilde{G}$ of all vertices not in the $\textit{branch vertex cover}, \ \VC$. This is a direct result of the $O(\poly(\log{n},\epsilon^{-1}))$ arboricity of $\tilde{G}$ from Corollary~\ref{lem:additional properties of fully dynamic cut sparsifier}, and the properties of a $\textit{branch vertex cover}$ of $\tilde{G}$ in Corollary~\ref{cor:lowdegInd}.

\

\noindent\textbf{Data structure for $\tilde{G}$}: A list of $O(\poly(\log{n},\epsilon^{-1}))$ spanning forests, which we will denote $\textsc{SPANNERS}_{G}$.

\noindent\textbf{Data structure for adjacency lists of $\tilde{G}$}:, We will denote it as $\textsc{ADJ-LIST}_{\tilde{G}}$, and it will have, for each vertex $v$, two lists $\textsc{LEAF}_v$ and $\textsc{BRANCH}_v$:
\begin{itemize}
\item The list $\textsc{LEAF}_v$ will have the adjacency list of $v$ for each spanning forest in $\textsc{SPANNERS}_G$ in which $v$ is a leaf.
\item Similarly, the list $\textsc{BRANCH}_v$ will have the adjacency list of $v$ edge for each spanning forest in $\textsc{SPANNERS}_G$ in which $v$ is not a leaf. 
\end{itemize}

\

\noindent\textbf{Data structure for $\VC$}: We will denote it as $\textsc{VC}_{\tilde{G}}$, which will be a list of all vertices $v$ whose list $\textsc{BRANCH}_v$ is non-empty.

\

\noindent\textbf{Data structure for $\tilde{G}_{\VC}$}: We will denote it as $\textsc{GRAPH}_{\VC}$, and it will contain an adjacency list, $\textsc{ADJ}_v$, for each vertex $v \in \VC$. Assume that each $\textsc{ADJ}_v$ has a data structure such that deletion and insertion of any edge takes $O(\log{n})$ time.

\

\noindent\textbf{Data structure for $H$}: We will denote it as $\textsc{SPARSE}_{\VC}$, and it will be the sparsified multi-graph.

\

We first show that moving a vertex in / out of the vertex cover
can be done in $O(\poly(\log{n},\epsilon^{-1}))$ time, assuming that the degree
of the vertex added/removed is small. Note that the small number of forests in $\tilde{G}$ and the choice of $ \VC $ allow us to meet this requirement.

\begin{figure}
\begin{algbox}
$\textsc{InsertVC}(G, \VC, v)$
\begin{enumerate}
\item Delete all edges $e_v \in K_v$ in $\textsc{GRAPH}_{\VC}$.
\item For all edges $e$ adjacent to $v$ in $\textsc{ADJ-LIST}_{\tilde{G}}$, insert $e_{\emptyset}$ into $\textsc{GRAPH}_{\VC}$.
\end{enumerate}
\end{algbox}

\caption{Moving a Vertex into $\VC$}
\label{fig:insertVC}

\end{figure}

\begin{lemma}
\label{lem:insertVC}
	If $v$ is not in $\textsc{VC}_{\tilde{G}}$, then running $\textsc{INSERT}_{\VC}(v)$ on $\textsc{GRAPH}_{\VC}$, using $\textsc{ADJ-LIST}_{\tilde{G}}$, will output $\textsc{GRAPH}_{\VC}$ equivalent to ${\tilde{G}}_{\VC \cup v}$ of $\textsc{ADJ-LIST}_{\tilde{G}}$ in $O(\poly(\log{n},\epsilon^{-1}))$ time.	
\end{lemma}

\begin{proof}
	Costs of the two steps are:
	\begin{enumerate}
		\item  Delete all edges $e_v \in K_v$ in $\textsc{GRAPH}_{\VC}$. This requires finding all incident vertices to $v$ in $\textsc{LEAF}_v$ and $\textsc{BRANCH}_v$, which is at most $O(\poly(\log{n},\epsilon^{-1}))$ because $\textsc{BRANCH}_v$ is empty due to $v$ not in $\textsc{VC}_{\tilde{G}}$. Every pair of vertices has a corresponding edge $e_v$ in $\textsc{GRAPH}_{\VC}$, so this takes $O(\poly(\log{n},\epsilon^{-1}))$ time.
		
		\item There are at most $O(\poly(\log{n},\epsilon^{-1}))$ edges adjacent to $v$ in $\textsc{ADJ-LIST}_{\tilde{G}}$, so adding all these edges into $\textsc{GRAPH}_{\VC}$ takes $O(\poly(\log{n},\epsilon^{-1}))$ time.
	\end{enumerate}
	
	If $v$ is not in $\textsc{VC}_G$, then $v$ must only be incident to $\VC$ in $\textsc{ADJ-LIST}_{\tilde{G}}$. Therefore in ${\tilde{G}}_{\VC \cup v}$, $v$ will only be incident to edges $e_{\emptyset}$ for each $e$ incident to $v$ in $\textsc{ADJ-LIST}_{\tilde{G}}$, and no edges $e_v$ will be in ${\tilde{G}}_{\VC \cup v}$. $\textsc{INSERT}_{\VC}(v)$ will perform exactly these operations on $\textsc{GRAPH}_{\VC}$.
	
\end{proof}

\begin{figure}
	\begin{algbox}
		$\textsc{RemoveVC}(G, \VC, v)$
		\begin{enumerate}
			\item For all edges $e$ adjacent to $v$ in $\textsc{ADJ-LIST}_{\tilde{G}}$, delete $e_{\emptyset}$ from $\textsc{GRAPH}_{\VC}$.
			\item Use all incident edges to compute $K_v$ and insert all $e_v \in K_v$ into $\textsc{GRAPH}_{\VC}$
		\end{enumerate}
	\end{algbox}
	
	\caption{Removing a Vertex from $\VC$}
	\label{fig:removeVC}
	
\end{figure}

\begin{lemma}
\label{lem:removeVC}
	If $\textsc{BRANCH}_v$ is empty, then running $\textsc{REMOVE}_{\VC}(v)$ on $\textsc{GRAPH}_{\VC}$, using $\textsc{ADJ-LIST}_{\tilde{G}}$, will output $\textsc{GRAPH}_{\VC}$ equivalent to ${\tilde{G}}_{\VC \setminus v}$ of $\textsc{ADJ-LIST}_{\tilde{G}}$ in $O(\poly(\log{n},\epsilon^{-1}))$ time.
	
\end{lemma}

\begin{proof} Costs of the two steps are:
	\begin{enumerate}
		\item At most $O(\poly(\log{n},\epsilon^{-1}))$ edges are adjacent to $v$ in $\textsc{ADJ-LIST}_{\tilde{G}}$, so deleting all these edges from $\textsc{GRAPH}_{\VC}$ takes $O(\poly(\log{n},\epsilon^{-1}))$ time.
		
		\item $v \notin \VC$, so $v$ has $O(\poly(\log{n},\epsilon^{-1}))$ neighbors, and using all incident edges to compute each $e_v \in K_v$ and insert $e_v$ into $\textsc{GRAPH}_{\VC}$ takes $O(\poly(\log{n},\epsilon^{-1}))$ time.
	\end{enumerate}

	If $\textsc{BRANCH}_v$ is empty, then $v$ must only be incident to $\VC$ in $\textsc{ADJ-LIST}_{\tilde{G}}$. Therefore in ${\tilde{G}}_{\VC \setminus v}$, $v$ will never be incident to any edges $e_{\emptyset}$, and for any of its neighbors $w$ and $z$, $(w,z)_v$ will be in $\textsc{ADJ-LIST}_{\tilde{G}}$. $\textsc{INSERT}_{\VC}(v)$ will perform exactly these operations on $\textsc{GRAPH}_{\VC}$
	
\end{proof}

We now consider updating $\textsc{ADJ-LIST}_{\tilde{G}}$ given the addition/deletion of some edge. This process is simple in terms of time complexity, but has a small wrinkle in maintaining the correct $\textsc{LEAF}$ and $\textsc{BRANCH}$ structure. Specifically, for each forest, we can consider all of the degree one vertices to be leaves, except for when there is a disjoint edge in the forest. Accordingly, steps 3, 4, and 5 of the algorithm in Figure~\ref{fig:updateADJ} will take care of this edge case.

\begin{figure}
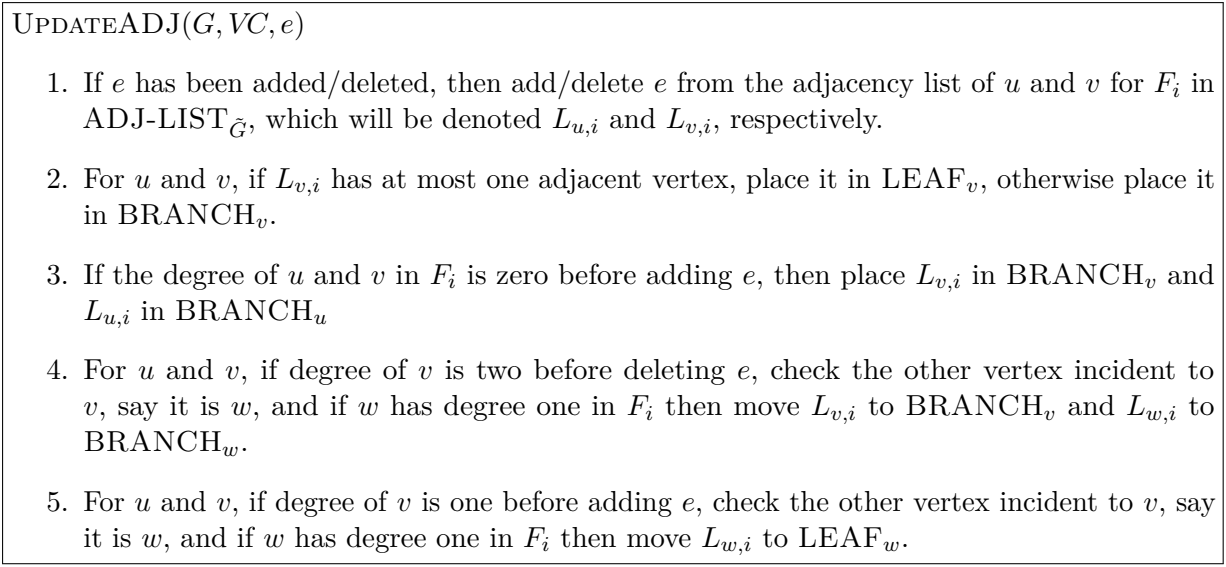

	\begin{algbox}
		$\textsc{UpdateADJ}(G, \VC, e)$
		\begin{enumerate}
			\item If $e$ has been added/deleted, then add/delete $e$ from the adjacency list of $u$ and $v$ for $F_i$ in $\textsc{ADJ-LIST}_{\tilde{G}}$, which will be denoted $L_{u,i}$ and $L_{v,i}$, respectively.
			\item For $u$ and $v$, if $L_{v,i}$ has at most one adjacent vertex, place it in $\textsc{LEAF}_v$, otherwise place it in $\textsc{BRANCH}_v$.
			\item If the degree of $u$ and $v$ in $F_i$ is zero before adding $e$, then place $L_{v,i}$ in $\textsc{BRANCH}_v$ and $L_{u,i}$ in $\textsc{BRANCH}_u$
			\item For $u$ and $v$, if degree of $v$ is two before deleting $e$, check the other vertex incident to $v$, say it is $w$, and if $w$ has degree one in $F_i$ then move $L_{v,i}$ to $\textsc{BRANCH}_v$ and $L_{w,i}$ to $\textsc{BRANCH}_w$.
			\item For $u$ and $v$, if degree of $v$ is one before adding $e$, check the other vertex incident to $v$, say it is $w$, and if $w$ has degree one in $F_i$ then move $L_{w,i}$ to $\textsc{LEAF}_w$.
		\end{enumerate}
	\end{algbox}
	
	\caption{Update $\textsc{ADJ-LIST}_{\tilde{G}}$}
	\label{fig:updateADJ}
	
\end{figure}

\begin{lemma}
	\label{lem:updateADJ}
	$\textsc{UpdateADJ}(G, \VC, e)$ takes $O(\log{n})$ time and all vertices $v$ such that $L_{v,i}$ are in $\textsc{BRANCH}_v$, maintain a 2-approximate vertex cover of $F_i$.
\end{lemma}

\begin{proof}
	Finding the adjacency list of $u$ and $v$ for $F_i$ in $\textsc{ADJ-LIST}_{\tilde{G}}$ takes $O(\log{n})$ time. The rest of the steps all take $O(1)$ time, as they are just there to ensure we maintain the 2-approximate vertex cover of $F_i$.
	
	For all trees, other than a single edge, it suffices to put all vertices with degree $\geq 2$ in the vertex cover, and 2-approx tree theorem tells us that this is a 2-approximate vertex cover. Step 3 and 4 of Update $\textsc{ADJ-LIST}_{\tilde{G}}$ ensure that in the single edge case, $e = (u,v)$ that $L_{v,i}$ is in $\textsc{BRANCH}_v$ and $L_{u,i}$ is in $\textsc{BRANCH}_u$, which is still a 2-approximate vertex cover. Further, step 5 ensures that anytime an edge is added to a tree that just contains a single edge, all vertices of degree one have their adjacency list moved to the $\textsc{LEAF}$ list.
\end{proof}

\subsubsection{Full Dynamic Update Process}

Finally, we consider the addition/deletion of an edge in $\textsc{SPANNERS}_{G}$. Specifically, let the edge $e = (u,v)$ be added/deleted from forest $F_i$. The above two operations allow us to reduce it to the simpler
case of both $u$ and $v$ being in $\VC$.
The update process will occur as follows:

\begin{enumerate}	
	\item For $u$ and $v$, if $v \notin \textsc{\VC}_{\tilde{G}}$, then run $\textsc{InsertVC}$ on $\textsc{GRAPH}_{\VC}$, $\textsc{VC}_{\tilde{G}}$, and $v$
	
	\item Update $\textsc{ADJ-LIST}_{\tilde{G}}$
	
	\item If $e$ was added/deleted from $\tilde{G}$, insert/delete edge $e_{\emptyset}$ from $\textsc{GRAPH}_{\VC}$ and insert $u$ and $v$ into $\textsc{VC}_{\tilde{G}}$
	
	\item For $u$ and $v$, if $\textsc{BRANCH}_v$ is empty, then run $\textsc{RemoveVC}$ on $\textsc{GRAPH}_{\VC}$, $\textsc{VC}_{\tilde{G}}$, and $v$, and delete $v$ from $\textsc{VC}_{\tilde{G}}$
\end{enumerate}

By Lemma~\ref{lem:insertVC}, $\textsc{GRAPH}_{\VC}$ is equivalent to ${\tilde{G}}_{\VC \cup \{u,v\}}$ on updated $\textsc{ADJ-LIST}_{\tilde{G}}$ after step 3 because $u$ and $v$ are in $\VC$. Similarly, the moving of $u$ and $v$ outside of $\VC$ ensures our final state is good.

\paragraph{Proof of Theorem~\ref{thm:dyn2Maintenance}}: The full update process for $\textsc{ADJ-LIST}_{\tilde{G}}$, $\textsc{VC}_{\tilde{G}}$, and $\textsc{GRAPH}_{\VC}$ only calls $\textsc{InsertVC}$, $\textsc{RemoveVC}$, and $\textsc{UpdateADJ}$ a constant number of times. Therefore, by Lemma~\ref{lem:insertVC}, Lemma~\ref{lem:removeVC}, and Lemma~\ref{lem:updateADJ} this process takes $O(\poly(\log{n},\epsilon^{-1}))$ time. This also implies that at most $O(\poly(\log{n},\epsilon^{-1}))$ edges can be added/deleted from ${\tilde{G}}_{VC}$, and by Corollary~\ref{lem:additional properties of fully dynamic cut sparsifier} maintaining $H$ will take at most $O(\poly(\log{n},\epsilon^{-1}))$ time.

\section{Vertex Sampling in Bipartite Graphs}
\label{sec:vertSparsify}

We now design an improved method for
reducing a graph onto one whose vertex size is
$O(|\VC|\poly(\log{n}, \epsilon^{-1})) + |X|/2$.
Instead of sampling edges of $G_{\VC}$, it samples vertices
in $X = V \setminus \VC$ using $G_{\VC}$ as a guide.
This question that we're addressing, and the vertex sampling
scheme, is identical to the terminal cut sparsifier question
addressed in~\cite{AndoniGK14}.
In the next section we will apply this sampling scheme to obtain
a vertex sparsification routine that will reduce onto a graph of
size proportional to $O(|\VC|\poly(\log{n}, \epsilon^{-1}))$
without losing a factor of 2 approximation.

We will reuse the notation from Section~\ref{subsec:critcalCut}, and we encourage the reader to revisit the definitions in that subsection. For this section, we will exclusively be dealing with subsets of $\VC$, and we will drop the $\VC$ subscript from each $S_{\VC}$. So, formally our goal is to find $H$ so that for all
$S \subset \VC$,
\[
(1-\epsilon)\Delta_{G}(S) \leq \Delta_{H}(S) \leq (1-\epsilon)\Delta_{G}(S).
\]

This sampling scheme allows us to keep expectation of the cuts
on $\VC$ to be exactly the same, instead of having a factor
$2$ error from the conversion from $G$ to $G_{\VC}$.
The connection to $G_{\VC}$ on the other hand allows
us to bound the variance of this sampling process as before.

In our application of this sampling routine to vertex sparsification, we will consider sparsifying $G \setminus X$ separately, so for simplicity, we assume here that $(VC,X)$ is a bipartition and 

\[G = \bigcup_{x \in X}N_x \]

Further, we first focus on the case where
all vertices in $X$ have degree $d$, and all edge weights
in $X$ are within a factor of $U$ from each other.
We will show reductions from general cases to ones
meeting these assumptions in Subsection~\ref{subsec:reduction}.

As before, let $G_{\VC}$ be the multigraph generated
by the clique edges from Theorem~\ref{thm:schurComplement}:
\[
G_{\VC} = \bigcup_{x \in X} K_x.
\]
Lemma~\ref{lem:elimWeight} implies that the weights of every (multi) edge $e_x \in G_{\VC}$
are within a factor of $O(U^2 d)$ from each other.

As mentioned, we ultimately want to obtain a vertex sparsification scheme that reduces to size $O(|\VC|\poly(\log{n}, \epsilon^{-1}))$ for further application.
As a result, instead of doing a direct union bound over
all $2^{|\VC|}$ cuts to get a size of $\poly(|\VC|)$ as in~\cite{AndoniGK14},
we need to invoke cut counting as with cut sparsifier constructions.
This necessitates the use of objects similar to $t$-bundles
to identify edges with small connectivity.

Our proof will use a similar structure to that of Fung et al.~\cite{FungHHP11}, particularly the cut-counting based analysis
of cut sparsifiers.
We will follow their definitions, which are in turn
based on the definition of edge strength by
Benczur and Karger~\cite{BenczurK15}.

\begin{definition}
\label{dfn:heavy}
In a graph $G$, an edge is $e$ $k-heavy$ if the connectivity of its endpoints is at least $k$ in $G$.
Furthermore, for a cut $S$, its $k-projection$ is the set of $k-heavy$ edges in the edges cut, $\partial(S)$.
\end{definition}

We will refer to edges that we cannot certify to be heavy
as light.
These edges are analogous to the bundle edges from
the cut sparsifier routine from Section~\ref{sec:dynamic cut sparsifier}.

Before we continue, we remark that these definitions of
heavy/strong edges in~\cite{FungHHP11, BenczurK15} is almost
the opposite of definitions in spectral sparsification.
In spectral sparsification, the edges with high leverage scores
are kept, and the low leverage score ones are sampled.
This issue can also be reflected in the robustness of this definition
in the presence of weights:
a natural way of generalizing heaviness is to divide the connectivity
of $uv$ by the weight $w(u, v)$.
This leads to a situation where halving the weight of an edge actually
makes it heavier.
In fact, these definitions of heaviness / strength are measuring
the connectivity in the graph between the endpoints of $e$,
instead of the strength of $e$ itself.
As our routines are in the cut-sparsification setting, we will
use these definitions in this version in order to be
consistent with previous works~\cite{FungHHP11, BenczurK15}, 
but may switch to a different set of notations in a future edit.

The main result of~\cite{FungHHP11}, when restricted to
graphs with bounded edge weights, states that we can
sample the $O(\log{n} \epsilon^{-2})$-heavy edges by a factor of $2$.
Our goal is to prove the analogous statement for sampling
heavy vertices, which we define as follows:

\begin{definition}
	\label{def:heavyVertex}
	A subset of $X$, $X^{heavy}$ is a $k$-heavy subset
	if every pair of vertices $u, v$ in some $N_x$ for some
	$x \in X^{heavy}$ is $k$-connected in the graph
	\[
		G^{light}_{\VC} = \cup_{x \notin X^{heavy}} K_x.
	\]
\end{definition}

We will show in Section~\ref{sec:onePlusEpsilon}, these heavy/light
subsets can be found by taking pre-images of more restricted
versions of $t$-bundles on $G_{\VC}$.
Our main structural result is that a heavy subset can be sampled
uniformly while incurring $\epsilon$-distortion.

\begin{figure}
	
	\begin{algbox}
		$\textsc{Sample}(G, \VC, X^{heavy})$
		
		\textbf{Input:} Bipartite graph $G$ with one bipartition $\VC$, heavy subset $X^{heavy}$ of the other bipartition.
		
		\textbf{Output:} Bipartite graph $H$ with bipartition $(\VC, XH)$.
		
		\begin{enumerate}
			\item Initialize $H \leftarrow \emptyset$, $XH = \emptyset$.
			\item For every $x \in X^{heavy}$, flip fair coin with probability $1/2$, if returns heads:
				\begin{enumerate}
					\item $H \leftarrow H + 2 N_{x}$.
					\item $XH \leftarrow XH \cup \{ x \}$
				\end{enumerate}
			\item Return $(H, XH)$.
		\end{enumerate}
	\end{algbox}
	
	\caption{Sampling Heavy Vertices}
	\label{fig:sampleHeavy}
	
\end{figure}

\begin{lemma}
	\label{lem:sampleHeavy}
	Given a bipartite graph $G$ between $\VC$ and $X$ such that
	$X$ has maximum degree $d$ and
	all edge weights are in some range $[\gamma, U \gamma]$, with $U = O(\poly(n))$ and any non-negative $\gamma$. 
	For any $\epsilon$, there is a parameter
	$t_{\min} = O(d U \log{n} \epsilon^{-2})$ such that if we're given
	a subset $X^{light}$ of $X$ so that $X^{heavy} = X \setminus X^{light}$
	is  $\gamma (dU) t$-heavy with $t \geq t_{\min} $ then the graph consisting of the light vertices and sampled heavy vertices,
	\[
		H = N(X^{light}) \cup \textsc{Sample}(G, \VC, X^{heavy})
	\]
	meets the condition:
	\[
	\abs{\Delta_G(S) - \Delta_{H}(S)} \leq \epsilon \Delta_G(S) 
	\]
	for all subsets $S \subseteq \VC$ w.h.p.
	Here the constants in $t_{\min}$ depends on the failure probability
	in the w.h.p.
\end{lemma}


The cut-counting proof of cut-sparsifiers from~\cite{FungHHP11}
essentially performs a union bound over distinct sets of $k$-heavy
projections over all cuts.
We will perform the same here, but over distinct partitions of $N_x$
over all $x$ in $X^{heavy}$.
We can first define the partition of a single vertex by a cut $S \subseteq \VC$ as:
\[
N_x(S) = \left\{S \cap N(x), N(x) \setminus S \right\}.
\]
Then we can define an equivalence relation on cuts as:
\begin{definition}
\label{def:equiv}
$S_1 \equiv_G S_2$ if for any $x \in X^{heavy}$,
$N_x(S_1) = N_x(S_2)$ 
\end{definition}

Note that this equivalence ignores the presence of edges
in $X^{light}$.
So we need to further take representatives of each equivalence class:

\begin{definition}
\label{def:sRep}
	Define $\mathcal{S}^{rep}$ to be the set of subsets $S \subset \VC$ such that
	\begin{enumerate}
		\item For every $S \in \mathcal{S}^{rep}$, there is some $x \in X^{heavy}$ s.t. $N_x(S) \neq \left\{ N_x, \emptyset \right\}$, i.e.\ $N_x$ is not entirely on one side of the cut.
		\item For any $S_1,S_2 \in \mathcal{S}^{rep}$, $S_1 \not\equiv_G S_2$
		\item For any $S \subset \VC$ such that $S \notin \mathcal{S}^{rep}$, there exists $\overline{S} \in \mathcal{S}^{rep}$ such that
		\begin{itemize}
		\item $S \equiv_{G} \overline{S}$, and
		\item $\Delta_G(\overline{S}) \leq \Delta_{G}(S)$.
		\end{itemize}
\end{enumerate} 	
\end{definition}

An immediate consequence of condition 1 is that for any $S \in \mathcal{S}^{rep}$ we have $\Delta_{G}(S) > \gamma t (dU)^{-1}$.
This set plays the same role as the unique
$k$-projections in cut sparsifiers.

\begin{lemma}
\label{lem:repGood}
Let $H$ be obtained from $G$ by sampling on $X^{heavy}$, then
for any element of $\mathcal{S}^{rep}$, $\overline{S}$ we have:
\[
\prob[H]{\bigcup_{S, S \equiv_{G} \overline{S}}
	\abs{\Delta_G(S) - \Delta_{H}(S)} > \epsilon \Delta_G(S) }
= \prob[H]{\abs{\Delta_G(\overline{S}) - \Delta_{H}(\overline{S})}
	> \epsilon \Delta_G(\overline{S})}	.
\]	
\end{lemma}

\begin{proof}
Let $G_{sample} = G \setminus \bigcup_{x \in X^{light}} N_x$
and $H_{sample} = H \setminus \bigcup_{x \in X^{light}} N_x$
be the graphs being sampled.

By construction of $H$, for any $S \subset \VC$,
\[
\abs{\Delta_G(S) - \Delta_{H}(S)}
= \abs{\Delta_{G_{sample}}(S) - \Delta_{H_{sample}}(S)}.
\]
By construction of our equivalence relation, if $S \equiv_G 
\overline{S}$,
\[
\abs{\Delta_{G_{sample}}(S) - \Delta_{H_{sample}}(S)}
= \abs{\Delta_{G_{sample}}(\overline{S}) - \Delta_{H_{sample}}(\overline{S})}.
\]
due to them having the same part that's not in $H$.
Therefore, the failure probability is limited by the element in
the equivalence class with the smallest $\Delta_{G}(S)$,
i.e.\ $\overline{S}$.
\end{proof}

\begin{corollary}
	\[ \prob[H]{\bigcup_{S \subset \VC}\abs{\Delta_G(S) - \Delta_{H}(S)} > \epsilon \Delta_G(S)} =  \prob[H]{\bigcup_{S \in \mathcal{S}^{rep}}\abs{\Delta_G(S) - \Delta_{H}(S)} > \epsilon \Delta_G(S)} \]
\end{corollary}

The key observation is that the sizes of subsets of $\mathcal{S}^{rep}$
of certain sizes can be bounded using cut-counting on $G_{\VC}$.
For any $S \subset \VC$, define
\[
K_x(S) = E_{K_x}(S \cap N(x),N(x)\setminus S),
\]
which are the edges in $K_x$ crossing $S$.
Similar to $S_1 \equiv_{G} S_2$, we can define
$S_1 \equiv_{G_{\VC}} S_2$ if for any
$x \in X^{heavy}$, $K_x(S_1) = K_x(S_2)$

\begin{lemma}
\label{lem:equivEquiv}
	For any $S_1,S_2 \subseteq \VC$, $N_x(S_1) = N_x(S_2)$ iff $K_x(S_1) = K_x(S_2)$.
	Therefore $S_1,S_2 \subseteq \VC$, $S_1 \equiv_G S_2$ iff $S_1 \equiv_{G_{\VC}} S_2$.
\end{lemma}

\begin{proof}
We construct $K_x$ as a clique, so $K_x(S_1) = K_x(S_2)$ iff $S_1 \cap N(x) = S_2 \cap N(x)$ or $S_1 \cap N(x) = N(x) \setminus S_2$
\end{proof}

\begin{lemma}
	$|\{S\in \mathcal{S}^{rep}| \Delta_{G_{\VC}}(S) \leq K\}|$ is less than or equal to the number of distinct $\gamma(dU)^{-1} t$-projections in cuts of weight at most $K$ 
\end{lemma}

\begin{proof}
	 Lemma~\ref{lem:equivEquiv} gives that $\mathcal{S}^{rep}$
	 has the following properties for $G_{\VC}$
	 \begin{enumerate}
	 	\item For every $S \in \mathcal{S}^{rep}$, there is some $x \in X^{heavy}$ s.t.  $K_x(S) \neq \emptyset$.
	 	\item For any $S_1,S_2 \in \mathcal{S}^{rep}$, $S_1 \not\equiv_{G_{\VC}} S_2$
	 	
	 \end{enumerate} 
	 
	For any $S\in \mathcal{S}^{rep}$, let $E_{heavy}(S)$ denote all the
	$\gamma(dU)^{-1} t$-heavy edges crossing $S$ in $G_{\VC}$.
	The property above gives:
	\[
	\bigcup_{x \in X^{heavy}}K_x(S)
	\]
	 is a non-empty subset of $E_{heavy}(S)$, and
	 \[
	 \bigcup_{x \in X^{heavy}}K_x(S_1) \neq \bigcup_{x \in X^{heavy}}K_x(S_2) \qquad \forall S_1,S_2 \in \mathcal{S}^{rep}.
	 \]
	 Therefore, each $S\in \mathcal{S}^{rep}$ such that $\Delta_{G_{\VC}}(S) \leq K$, must be a distinct $\gamma(dU)^{-1}t-projection$ of weight at most $K$
	
\end{proof}

It remains to combine this correspondence with
cut counting to show the overall success probability
of the vertex sampling routine.

Proving this requires using Chernoff bounds.
The bound that we will use is below, it can be viewed
as a scalar version of Theorem 1 of~\cite{Tropp12}.

\begin{lemma}
\label{lem:chernoff}
	Let $Y_1 \ldots Y_n$ be random variables s.t.
	\begin{enumerate}
		\item $0 \leq Y_i \leq 1$.
		\item $\mu_i = \expec[Y_{i}]{Y_{i}} $
		\item $\mu = \sum_{i} \mu_i$
		
	\end{enumerate}
	Then for any $\epsilon \geq 0$
	\[
	\prob[Y_1 \ldots Y_n]{{\sum_{i} Y_i} > (1 +\epsilon) \mu }
	\leq \exp \left( -\frac{\epsilon^2 \mu}{2} \right).
	\]
	
	\[
	\prob[Y_1 \ldots Y_n]{{\sum_{i} Y_i} < (1  - \epsilon) \mu }
	\leq \exp \left( -\frac{\epsilon^2 \mu}{2} \right).	\]
\end{lemma}

This bound can be invoked in our setting on a single
cut $S$ as follows:

\begin{lemma}
\label{lem:singelCut}
For each cut $S$, we have
\[ 
\prob[H]{\abs{\Delta_H(S) - \Delta_G(S)}  > \epsilon \Delta_G(S)}
\leq 2 \exp \left(- \frac{\epsilon^2 \Delta_G(S) }{4\gamma} \right).
\]
\end{lemma}

\begin{proof}
Let
\[
\ww^{\max} (S) = \max_{x \in X} \{\ww^{(x)}(S) \}.
\]
We will only consider $S \in \mathcal{S}^{rep}$, so we know $\ww^{\max}(S) > 0$, which implies $\ww^{\max} (S) \geq \gamma $. For each $S \subseteq \mathcal{S}^{rep}$ and for all $x \in X$, let $Y_{x}(S)$ be the random variable such that either 
\begin{enumerate}
	\item $Y_{x}(S) = \frac{\ww^{(x)}(S)}{2\ww^{\max}(S)}$ if $x \in X^{light}$
	
	\item $Y_{x}(S)$ equals $\frac{\ww^{(x)}(S)}{\ww^{\max}(S)}$ w.p. $1/2$, and $0$ w.p. $1/2$.
\end{enumerate}

Accordingly, we have $\sum_{x\in X} \expec[Y_{x}(S)]{Y_{x}(S)} = \frac{1}{2\ww^{\max}(S)} \sum_{x \in X} \ww^{(x)}(S) = \frac{1}{2\ww^{\max}(S)} \Delta_{G}(S)$.
The bound then follows from invoking Lemma~\ref{lem:chernoff}.
\end{proof}

\begin{proof}(Of Lemma~\ref{lem:sampleHeavy})

Let $\Delta_{G_{\VC}}(S)$ be the weight of cutting $S \subset \VC$ in $G_{\VC}$. From Theorem~\ref{thm:schurComplement} for any $S \in \mathcal{S}^{rep}$ that $\Delta_G(S) \geq \Delta_{G_{\VC}}(S)$.
Therefore,
\[\sum_{S\subseteq \mathcal{S}^{rep}} 2 \exp \left(- \frac{\epsilon^2 \Delta_G(S)}{4\gamma}  \right)
\leq \sum_{S\subseteq \mathcal{S}^{rep}} 2 \exp \left(- \frac{\epsilon^2 \Delta_{G_{\VC}}(S) }{4\gamma} \right)
\]

The main cut-counting bound follows from Theorem 1.6~\cite{FungHHP11}
on multi-graphs, and by our construction of $\mathcal{S}^{rep}$ gives:
\[
|\{S\in \mathcal{S}^{rep}| \Delta_{G_{\VC}}(S) \leq K\}| \leq
\begin{cases}
n^{2 K d U (\gamma t)^{-1}} & \qquad \text{if } K \geq {\gamma (dU)^{-1}t},\\
0 & \qquad \text{otherwise}.
\end{cases}
\]

Each vertex adds weight at most $\gamma  d U $ for any cut, so we can upper bound $K$ by $n^2 \gamma U$ because $d \leq n$.
Invoking cut counting for intervals of length $\gamma$ from $K \geq \gamma (dU)^{-1}t$ to $K \leq n^2 \gamma  U$
allows us to bound the overall failure probability by:

\begin{multline}
\leq \sum_{i = (dU)^{-1} t}^{n^2U}  \left( \sum_{\substack{S \in \mathcal{S}^{rep} \\ \gamma i \leq \Delta_{G_{\VC}}(S) \leq \gamma(i+1) }}
2 \exp \left(-\frac{\epsilon^2 \Delta_{G_{\VC}}(S)}{4\gamma} \right) \right)\\
\leq \sum_{i = (dU)^{-1} t}^{n^2U} 2 n^{2 (i+1) d U t^{-1}} \exp \left( - \frac{\epsilon^{2} i}{4 } \right)
= \sum_{i = (dU)^{-1} t}^{n^2U} 2n^{2 (i+1) d U t^{-1} - \frac{\epsilon^{2} i}{4 \log{n} }}.
\end{multline}
Note that we're free to choose $t$, and it can be checked
that for $U \leq n^{c_1}$, setting $t \geq (28 + 4c_1)c_2 d U \log{n} \epsilon^{-2}$ bounds this by $n^{-c_2}$ for any $c_2 \geq 1$. Note that if $U$ is larger than $O(\poly(n))$, we could set $t = O(d U^2 \log{n} \epsilon^{-2})$ and still achieve w.h.p., but for our practical purposes assuming $U = O(\poly(n))$ is more than sufficient because $U$ will always be $O(\poly(\log{n}, \epsilon^{-1}))$.

\end{proof}

\section{Maintaining $(1 + \epsilon)$-Approximate Undirected Bipartite Min-Cut}
\label{sec:onePlusEpsilon}

In this section, we will again consider the bipartite minimum $s-t$ cut problem of Section~\ref{sec:dynamic min cut}, and will improve the approximation guarantee to $(1 + \epsilon)$. This improvement will require many of the techniques from Section~\ref{sec:dynamic min cut}, but we will bypass the loss of a factor 2 approximation by utilizing the vertex sampling scheme presented in Section~\ref{sec:vertSparsify}.
A high level overview of these techniques is in Section~\ref{subsec:overviewApproxFlow}.
The dynamic algorithm given in this section will rely heavily on the definitions and observations of Subsection~\ref{subsec:critcalCut}, which we encourage the reader to revisit.

Lemma~\ref{lem:sampleHeavy}, along with the framework
from Section~\ref{sec:dynamic min cut} allow us  sample a
large set of vertices if the optimal minimum $s-t$ cut is small,
and will guarantee that the sampled vertices have $O(\poly(\log{n},\epsilon^{-1}))$ degree.
However, Lemma~\ref{lem:sampleHeavy} as stated
require incident edges of all sampled vertices to have weight
within factor $O(\poly(\log{n},\epsilon^{-1}))$ of one another.
In this section, we integrate this subroutine into the data structure
framework, leading to our main result for approximating undirected
bipartite maximum flows:

\mainBipartite*


Section~\ref{subsec:vertexSpars} will show how the vertex sampling scheme given in Section~\ref{sec:vertSparsify} can be iteratively applied, reducing to a graph with $O(|VC|\poly(\log{n},\epsilon^{-1}))$ vertices and $O(|VC|\poly(\log{n},\epsilon^{-1}))$ edges. This section will first present the full vertex sparsification scheme, and then examine the two primary components of this scheme. Section~\ref{subsec:reduction} will show how we can pre-process a graph to ensure that all edge weights of each sampled vertex are close to each other, which will be necessary for bucketing sampled vertices. Section~\ref{subsec:boundedSparsify} will utilize these bounded properties and the vertex sampling of Section~\ref{sec:vertSparsify} to give a vertex sparsification scheme for each bucket, culminating in a proof of correctness for the full scheme in terms of approximation guarantees and bounds on the number of edges and vertices. Section~\ref{subsec:generalGraph} will extend vertex sparsification to general graphs without bounds on degree for the static case, proving Corollary~\ref{cor:terminalSparsify}.

Section~\ref{subsec:dynMinCut1} will then use this vertex sparsification scheme along with many of the components from Section~\ref{sec:dynamic min cut} to give a fully dynamic algorithm for maintaining a minimum $s-t$ cut on a bipartite graph. The correctness of this algorithm will follow from the correctness of the dynamic algorithm in Section~\ref{sec:dynamic min cut} and the correctness of vertex sparsification. Accordingly, it will then only be necessary to establish that we can dynamically update all necessary data structures in $O(\poly(\log{n}, \epsilon^{-1}))$ time.

\subsection{Vertex Sparsification in Quasi-Bipartite Graphs}
\label{subsec:vertexSpars}

The general framework of the routine is shown in Figure~\ref{fig:vertexSparsify}.

\begin{figure}
	
	\begin{algbox}
		$\textsc{VertexSparsify}(G, \VC, XG, d, \epsilon)$
		
		\textbf{Input:} Graph $G$ with vertex cover $\VC$ and $XG = V \setminus \VC$, such that the degree of each vertex in $XG$ is bounded by $d$.
		
		\begin{enumerate}
			\item Build $\widehat{G}$ on the same vertex set as $G$ s.t. $G \approx_{\epsilon/2} \widehat{G}$ and for each $x$ in $X\widehat{G}$, the weights are within a factor of $O( d / \epsilon)$ of each other.
			\item Bucket $\widehat{G}$ by maximum edge weights in each $N_x$
			into $\widehat{G}_1 \ldots \widehat{G}_{L}$, along with $\widehat{G} \setminus XG$
			\item Set $t = O(d^2\log^3{n}\epsilon^{-3})$, initialize $H = (\VC, \emptyset)$.
			\item With error $\epsilon / 2$, sparsify $\widehat{G}\setminus XG$ and $\textsc{BoundedVertexSparsify}$ each $\widehat{G}_{i}$, giving $H_i$
			\item Return the union of each sparsified graph, $H = H \setminus XG \cupdot H_1 \cupdot \ldots \cupdot H_L$.
		\end{enumerate}
	\end{algbox}
	
	\caption{Vertex Sampling in $G$}
	\label{fig:vertexSparsify}
	
\end{figure}

\begin{theorem}
	\label{thm:vertexSparsify}
	Given any graph $G$, vertex cover $\VC$ and $XG = V \setminus \VC$, such that the degree of each vertex in $XG$ is bounded by $d$, with weights in $[\gamma, O(\gamma W)]$ where $\log{W} = O(poly(\log{n}))$,
	and error $\epsilon$. Then there is a $t = O(d^2\log^3{n}\epsilon^{-3})$ whereby
	$\textsc{VertexSparsify}(G, \VC, XG, d, \epsilon)$ returns $H$
	s.t. w.h.p.
	\begin{enumerate}
		\item $H \setminus XG$ is a multi-graph on $\VC$ with $O(|VC|\poly(\log{n},\epsilon^{-1}))$ edges, and each $H_i$ is a bipartition with $\VC$ on one side, and at most $O(|\VC| t \log{n} )$ vertices of $XG$ on the other.
		\item $H \approx_{\epsilon}^{\VC} G$.
		\item All edge weights of $H$ are in $[\gamma, O(\gamma n W)]$ 
	\end{enumerate}
	Here the constant in front of $t$ depends on $W$ as well as in the w.h.p. condition.
\end{theorem}

A proof of Theorem~\ref{thm:vertexSparsify} will be given at the end of Section~\ref{subsec:boundedSparsify}.



\subsubsection{Reduction to Bounded Weight Case}
\label{subsec:reduction}

The idea here will be to look at each $N_x$ and move the low weight edges into $G \setminus X$ , thereby ensuring that the remaining edges in $N_x$ have weight within a $O(\poly(\log{n},\epsilon^{-1}))$ factor. This will create a multi-graph in $G \setminus X$, where will use the normal notation $(u,v)_x$ to denote an edge added by $N_x$.

\begin{figure}
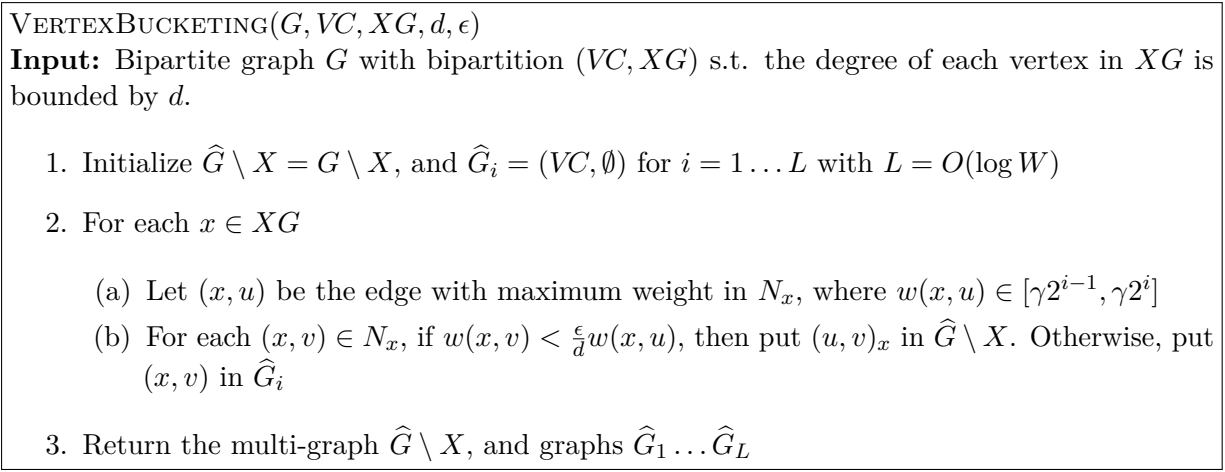

	
	\begin{algbox}
		$\textsc{VertexBucketing}(G, \VC, XG, d, \epsilon)$
		
		\textbf{Input:} Bipartite graph $G$ with bipartition $(\VC, XG)$ s.t. the degree of each vertex in $XG$ is bounded by $d$.
		
		\begin{enumerate}
			\item Initialize $\widehat{G} \setminus X = G \setminus X$, and $\widehat{G}_{i} = (\VC, \emptyset)$ for $i = 1 \ldots L$ with $L = O(\log{W})$
			
			\item For each $x \in XG$
			
			\begin{enumerate}
				\item Let $(x,u)$ be the edge with maximum weight in $N_x$, where $\ww (x,u) \in [\gamma 2^{i-1}, \gamma 2^i]$
				
				\item For each $(x,v) \in N_x$, if $\ww (x,v) < \frac{\epsilon}{d} \ww (x,u)$, then put $(u,v)_x$ in $\widehat{G}\setminus X$. Otherwise, put $(x,v)$ in $\widehat{G}_i$
				
			\end{enumerate}

			\item Return the multi-graph $\widehat{G}\setminus X$, and graphs $\widehat{G}_{1} \ldots \widehat{G}_{L}$
		\end{enumerate}
	\end{algbox}
	
	\caption{Vertex Bucketing in $G$}
	\label{fig:vertexBucket}
	
\end{figure}

\begin{theorem}
	\label{thm:bucketing}
	Given $G$ with bipartition $(VC,XG)$ with weights in $[\gamma,\gamma W]$, such that the degree of each vertex in $XG$ is bounded by $d$, for any $\epsilon$, $\textsc{VertexBucketing}(G,VC,XG,d,\epsilon)$ will return $\widehat{G} = \widehat{G} \setminus X \cupdot \widehat{G}_1 \cupdot \ldots \cupdot \widehat{G}_L$ such that 
	\begin{enumerate}
		\item $G \approx_{\epsilon} \widehat{G}$
		
		\item For each $\widehat{G}_i$, the weights of $\widehat{G}_i$ are in $[\gamma,2\gamma d \epsilon^{-1}]$ for some $\gamma$
		
		\item Any edge $e_{\emptyset} \in \widehat{G}$ must be in $\widehat{G} \setminus X$
		
		\item If $x \in X$ has non-zero degree in $\widehat{G}_i$, then $x$ has zero degree in $\widehat{G} \setminus \widehat{G}_i$, and the degree of $x$ in $\widehat{G}_i$ is bounded by $d$
		
	\end{enumerate}	
	
\end{theorem}

\begin{proof}
	Items 2, 3, and 4 follow from construction, and because $XG$ is an independent set in $G$, we can conclude that $G \approx_{\epsilon} \widehat{G}$ from Lemma~\ref{lem:threePoint} and Lemma~\ref{lem:reduce} below.
	
\end{proof}

\begin{lemma}
	\label{lem:threePoint}
	Consider a graph on three vertices, $x$, $u$, and $v$
	with edges between $xu$ and $xv$.
	If $\ww (x, v) \leq \epsilon \ww (x, u)$,
	then the graph with edges  $xu$ with weight $\ww (x, u)$
	and $uv$ with weight $\ww (x, v)$ is an $\epsilon$-approximation
	on all cuts.
\end{lemma}

\begin{proof}
	The only interesting cuts are singletons:
	\begin{enumerate}
		\item Removing $v$ has $\ww (x, v)$ before and after.
		\item Removing $x$ has $\ww (x, u) + \ww (x, v)$ before,
		and $\ww (x, u)$ after, a factor of $\epsilon$ difference since
		\[
		\ww (x, u) + \ww (x, v) \leq (1 + \epsilon) \ww (x, u).
		\]
		\item Removing $u$ has $\ww (x, u)$ before, and $\ww (x, u) + \ww (x, v)$ after, same as above.
	\end{enumerate}
	
\end{proof}

Invoking this repeatedly on small stars gives:
\begin{lemma}
	\label{lem:reduce}
	A star $x$ with degree $d$
	can be reduced to one whose maximum and minimum
	weights is within a factor of $O(d \epsilon^{-1})$
	while only distorting cuts by a factor of $1 + \epsilon$. 
\end{lemma}

\begin{proof}
	Let the neighbors of $x$ be $v_1 \ldots v_d$ s.t.
	$\ww (x, v_1) \geq \ww (x, v_2) \geq \ldots \geq \ww (x, v_d)$.
	Suppose $\ww (x, v_i) < \epsilon /d \ww (x, v_1)$,
	then applying Lemma~\ref{lem:threePoint} gives a
	multiplicative error of $1 + \epsilon / d$.
	Applying this at most $d$ times gives the approximation
	ratio, and moves all the light edges onto $v_1$.
\end{proof}

\subsubsection{Bounded Weight Vertex Sparsification}
\label{subsec:boundedSparsify}

\begin{figure}
	
	\begin{algbox}
		$\textsc{BoundedVertexSparsify}(G, \VC, XG,t)$
		
		\textbf{Input:} Bipartite graph $G$ with bipartition $(\VC, XG)$
		
		\begin{enumerate}
			\item Initialize 	$ G_0 \gets G $, $ XG_0 \gets XG$, and $H \gets \emptyset$
			
			\item For each $i = 0$ to $l-1$
			
			\begin{enumerate}
				\item Compute a $t$-bundle vertex set $ XG_{i}^{light} \subseteq XG_{i}$ of $ G_{i} $
				
				\item $ (G_{i+1}, XG_{i+1}) \gets Sample(G_i, \VC, XG_i, XG_i^{light})$
				
				\item Add $\bigcup_{x \in XG_i^{light}} N^i_x$ to $H$
				
			\end{enumerate}

			\item Return $H = H \cup G_{l }$
		\end{enumerate}
	\end{algbox}
	
	\caption{Bounded Weight Vertex Sparsification in $G$}
	\label{fig:boundedVertexSparsification}
	
\end{figure}

The bucketing of vertices in the independent set ensures that all the weights in each bucket are within a factor $O(d/\epsilon)$, which will allow us to iteratively reduce the number of vertices by applying the $\textsc{Sample}$ algorithm given in Section~\ref{sec:vertSparsify} $O(\log{n})$ times. Note that our $\textsc{Sample}$ algorithm doubles the weights of each sampled star, so $N^i_x$ will denote the star $x$ in $i$th iteration graph $G_i$ with updated weights for that graph.

\begin{theorem}
	\label{thm:fullBounded}
	Given a bipartite graph $G$ with bipartition $(VC,XG)$, and weights in $[\gamma, U \gamma]$ where $U= O(poly(n))$, with degree of $x \in XG$ bounded by $d$, and error $\epsilon$. Then there is a $t = O(dU\log^3{n} \epsilon^{-2})$ whereby $\textsc{BoundedVertexSparsify}(G,VC,XG,t)$ returns $H$, s.t. w.h.p.
	\begin{enumerate}
		\item $H$ is a bipartition with $\VC$ on one side and at most $O(|VC|t \log{n})$ vertices on the other
		
		\item $H \approx_{\epsilon}^{VC} G$
	\end{enumerate}
	
\end{theorem}

\begin{proof}
	\textbf{(1)}: Set $l = O(\log{n})$ and note that $|XG| \leq n$, so $G_l$ is unlikely to have many remaining vertices after sampling $O(\log{n})$ times by a standard argument using concentration bounds. Then, Lemma~\ref{lem:lightVertex} will show $|XG_i^{light}| \leq t|VC|$ for all $i$, giving the desired size.
	
	\noindent\textbf{(2)}: By construction of $\textsc{BoundedVertexSparsify}(G,VC,XG)$, the weights of each $G_i$ are in $[2^{i}\gamma,2^{i}\gamma U]$. We will show in Lemma~\ref{lem:lightVertex} that for each $G_i$ and $XG_i$, we can find a $t$-bundle vertex set $XG_i^{light}$ of $XG_i$, such that  $XG_i^{heavy} = XG_i \setminus XG_i^{light}$ is a $2^i\gamma (dU)^{-1} t - heavy$ vertex subset. Assuming that this is the case, from Lemma~\ref{lem:sampleHeavy}, if we set $\widehat{\epsilon} = \frac{\epsilon}{l}$, then with high probability \[G_i \approx_{\widehat{\epsilon}}^{\VC} G_{i+1} \cup \bigcup_{x \in XG_i^{light}} N^i_x \]
	
	By construction, for all $j < i$, $XG_j^{light} \cap XG_i = \emptyset$, so adding each $\bigcup_{x \in XG_j^{light}} N^j_x$ to both sides will still preserve the relation above. Applying this argument inductively and using $\widehat{\epsilon} = \frac{\epsilon}{l}$ gives $H \approx_{\epsilon}^{\VC} G$ with high probability.
	
\end{proof}

In order to complete the proof of Theorem~\ref{thm:fullBounded}, it is now necessary to show that for each $G_i$ and $XG_i$, we can construct $XG_i^{light}$ such that $XG_i \setminus XG_i^{light}$ is a $2^i \gamma (dU)^{-1} t-heavy$ subset of $XG_i$. The idea will simply be to construct $XG^{light}_i$ from $t$ disjoint spanning forests in $G_{\VC}^i$ with some additional properties that will allow $O(\poly(\log{n},\epsilon^{-1}))$ dynamic maintenance in the following subsection. 

\begin{definition}
	Given $G$ with vertex bipartition $(\VC,X)$, we say that $F = F_1 \cupdot \ldots \cupdot F_t$ is a $t$-clique forest if 
	
	\begin{enumerate}
		\item Each $F_{i}$ is a forest of $G_{\VC}$ and all are disjoint. 
		
		\item For any $x \in X$, at most one edge $e_x \in K_x$ is in $F$. 
		
		\item For all $x \in X$ such that $ F \cap K_x = \emptyset$, for any $e_x = (u,v)_x \in K_x$, $u$ and $v$ are connected in all $F_i$
		
	\end{enumerate}
	
\end{definition}

\begin{lemma}
	Given $G$ with vertex bipartition $(\VC,X)$ such that all $x \in X$ have maximum degree $d$, weights in $[\gamma,\gamma U]$ and a $t$-clique forest $F$, if $X^{light} = \{x \in X | F \cap K_x = \emptyset \}$, then $X^{heavy} = X \setminus X^{light}$ is an $\gamma(dU)^{-1} t -heavy$ subset of $X$
	
\end{lemma}

\begin{proof}
	For some $(u,v)_x \in K_x$ with $x \in X^{heavy}$, suppose $(u,v)_x$ is in a cut $S_{\VC} \subset \VC$ such that $\Delta_{G_{\VC}}(S_{\VC}) < {\gamma(dU)^{-1} t}$. From Lemma~\ref{lem:elimWeight}, all edges in $G_{\VC}$ have weight at least $\gamma(dU)^{-1}$. Therefore, there must exist some $F_{j}$ such that $u$ and $v$ are not connected, giving a contradiction.
	
\end{proof}

\begin{figure}
	
	\begin{algbox}
		$\textsc{LightVertices}(G_i, \VC, XG_i)$
		
		\textbf{Input:} Bipartite graph $G_i$ with bipartition $(\VC, XG_i)$
		
		\begin{enumerate}
			\item Initialize $ XG^{light}_i \gets \emptyset$ and $F_i = \bigcup_{j \in [t]} F_{i,j}$ with $F_{i,j} \gets \emptyset$ for all $j$
			
			\item For each $j = 1$ to $t$
			
			\begin{enumerate}
				\item While some edge $e_x \in G_{\VC}^{i}$ can be added to forest $F_{i,j}$
				
				\item Place $e_x$ in $F_{i,j}$, place $x$ in $XG^{light}_i$, and remove $K_x$ from $G_{\VC}^i$
				
			\end{enumerate}

			\item Return $XG^{light}_i$
		\end{enumerate}
	\end{algbox}
	
	\caption{Light Vertex Set of $XG$}
	\label{fig:lightVertices}
	
\end{figure}

Note that after the algorithm terminates $G_{\VC}^i = \bigcup_{x \in XG^{heavy}_i} K_x$, which will be necessary for the dynamic maintenance. The following lemma follows by construction and the fact that each forest has at most $|VC| - 1$ edges.

\begin{lemma}
	\label{lem:lightVertex}
	$F_i$ is a $t$-clique forest of $G_{\VC}^i$, $|XG_i^{light}| \leq t|VC|$, and $XG^{heavy}_i$ is a $2^i \gamma (dU)^{-1} t-heavy$ subset of $XG_i$	
\end{lemma}

\paragraph{Proof of Theorem~\ref{thm:vertexSparsify}}

\noindent\textbf{(1)} The first part follows from Theorem~\ref{thm:fully dynamic cut sparsifier} and the second part follows from Theorem~\ref{thm:fullBounded}

\noindent\textbf{(2)} Property (1) of Theorem~\ref{thm:bucketing} gives us $G \approx_{\epsilon /2} \widehat{G}$ with $U = 4d\epsilon^{-1}$ for each $\widehat{G}_i$ from property (2). Then, property (3) implies that each $\widehat{G}_i$ is bipartite, and property (4) implies that each vertex in $X\widehat{G}_i$ is bounded by $d$. We can then apply Theorem~\ref{thm:fullBounded} to each $\widehat{G}_i$, with $U = 4d\epsilon^{-1}$ to get $\widehat{G}_i \approx_{\epsilon/2}^{VC} H_i$ with high probability. Note that we are implicitly assuming $U = O(poly(n))$, aka $\epsilon^{-1} = O(poly(n))$. As was discussed at the end of Section~\ref{sec:vertSparsify}, we could avoid this assumption by adding an extra $\epsilon^{-1}$ factor to the $t$-bundle, but any $\epsilon^{-1} = \omega(poly(n))$ loses any practical value. $L = O(\log{W}) = O(poly(\log{n}))$ by assumption, and property (4) of Theorem~\ref{thm:bucketing} ensures that a vertex is only sampled in one $\widehat{G}_i$, so taking the union over $O(poly(\log{n}))$ buckets preserves $\widehat{G} \approx_{\epsilon /2}^{VC} H$ w.h.p. for sufficient constants in $t$.  $G \approx_{\epsilon/2} \widehat{G}$ is a stronger statement than $G \approx_{\epsilon/2}^{VC} \widehat{G}$, implying $G \approx_{\epsilon}^{VC} H$

\noindent\textbf{(3)} Edge weights are only changed in $\textsc{Sample}$ where they are either doubled or left alone. $\textsc{VertexSparsify}$ calls $\textsc{Sample}$ at most $O(\log{n})$ times for each bucket of $\widehat{G}$, giving the appropriate bound.

\subsubsection{Improved Static Algorithm for General Graphs}
\label{subsec:generalGraph}

Composing this routine $O(\log{n})$ times,
along with spectral sparsifiers, leads to a static routine:


\terminalSparisfy*

Now that we have sufficient notation in place, by $terminal-cut-sparsifier$, we mean that $G \approx_{\epsilon}^{\VC} H$ with high probability. Note that this is almost equivalent to Theorem~\ref{thm:vertexSparsify}, but we make no assumptions on the degree of vertices in $X$. Also, we will specify $\poly(\log^{n},\epsilon^{-1})$ as $\log^{18}{n}\epsilon^{-7}$.

\begin{proof}

	Consider running the following routine iteratively:
	\begin{enumerate}
		\item Sparsify $G$ with error $\epsilonhat = \frac{\epsilon}{O(\log{n})}$ and output $\tilde{G}$
		\item Find the bipartite subgraph $\widehat{G}$ containing $\VC$ and vertices $X\widehat{G} \subseteq X$
		whose degree are less than $O( \log^{2}{n} \epsilon^{-2})$.
		Run \textsc{VertexSparsify} on $\widehat{G}$, $\VC$, $X\widehat{G}$, with $d = O( \log^{2}{n} \epsilon^{-2})$ and with error $\epsilonhat = \frac{\epsilon}{O(\log{n})}$, returning $\widehat{H}$
		\item $G \gets \tilde{G} \setminus \widehat{G}$ and $H \gets H \cup \widehat{H}$
	\end{enumerate}
	
	If at any point, we have $|X| < |\VC|\log^{17}{n} \epsilon^{-7}$, then return $H \cup G$.
	
	From \cite{SpielmanS11}, and the number of edges in $\tilde{G}$ is $O(n\log{n}\epsilonhat^{-2})$ with high probability.
	Therefore, at least half of $|X|$ have degree less than $O( \log^{2}{n} \epsilonhat^{-2})$ because otherwise the number of edges in $\tilde{G}$ would be $O(|X|\log^2{n} \epsilonhat^{-2}) = O(n\log^2{n} \epsilonhat^{-2})$ by the assumption $|X| \geq |\VC|$.
	This eliminates half the vertices in $X$ with high probability for every run of the routine, so the process can continue at most $O(\log{n})$ times.
	From Theorem~\ref{thm:vertexSparsify} each bucket of $\widehat{H}$ will have at most $O(|\VC| t\log{n})$ vertices with $t = O(d^2\log^3{n}\epsilonhat^{-3})$ and $d = O(\log^2{n}\epsilonhat^{-2})$, giving $O(|\VC| \log^{15}{n} \epsilon^{-7})$. We run sparsification on $G$ and $\textsc{VertexSparsify}$ on $\widehat{G}$ $O(\log{n})$ times, so from the guarantees of Theorem~\ref{thm:fully dynamic cut sparsifier} and property (3) of Theorem~\ref{thm:vertexSparsify}, the weights are within a factor $O(n^{O(\log{n})})$. Therefore, there are at most $O(\log^2{n})$ buckets of $\widehat{H}$, and at most $O(|\VC| \log^{17}{n} \epsilon^{-7})$ vertices which has the appropriate size requirement.
	
	Sparsification gives $G \approx_{\widehat{\epsilon}} \tilde{G}$ with high probability, which is a stronger statement than $G \approx_{\widehat{\epsilon}}^{\VC} \tilde{G}$. Theorem~\ref{thm:vertexSparsify}, which is still applicable for weight within a factor $O(n^{O(\log{n})})$, gives $\widehat{G} \approx_{\widehat{\epsilon}}^{\VC} \widehat{H}$ with high probability. Therefore $(\tilde{G} \setminus \widehat{G}) \cup \widehat{H} \approx_{2\widehat{\epsilon}}^{\VC} G$ with high probability. Applying this inductively for $O(\log{n})$ steps gives the desired relation by setting $\widehat{\epsilon} = \frac{\epsilon}{O(\log{n})}$ as was done in the iterative routine above.
	
	Sparsifying $G$ requires $O(m \cdot \poly(\log{n}, \epsilon^{-1}))$ work \cite{SpielmanS11}.
	Furthermore, in Section~\ref{subsec:dynMinCut1} we will show that \textsc{VertexSparsify} can be maintained dynamically in worst-case update time of $O(\poly(\log{n}, \epsilon^{-1}))$, so it's static runtime must be $O(m \cdot \poly(\log{n}, \epsilon^{-1}) )$.
	
\end{proof}

\subsection{Dynamic Minimum Cut of Bipartite Graphs}
\label{subsec:dynMinCut1}

Now that we have the full process of $\textsc{VertexSparsify}$, we will give the dynamic algorithm for maintaining a $(1 + \epsilon)$-approximate minimum cut in amortized $O(\poly(\log{n},\epsilon^{-1}))$ time. The algorithm in Figure~\ref{fig:dynamic1approx} will be analogous to the one given in Section~\ref{sec:dynamic min cut}, but will replace sparsification of $G_{\VC}$ with $\textsc{VertexSparsify}$, improving the approximation by a factor of $2$.

\begin{figure}
	
	\begin{algbox}
		
		\begin{enumerate}
			\item Dynamically maintain a sparsified $G$, which we will denote $\tilde{G}$ 
			
			\item Dynamically maintain a $\textit{branch vertex cover}$,  $\VC$, on $\tilde{G}$, where we ensure $s,t \in \VC$
	
			\item Dynamically maintain a vertex sparsified $\tilde{G}$ using $\VC$ and $X\tilde{G} = V \setminus \VC$ which we will denote $H$
	
			\item Every $\frac{\epsilon}{2}\Delta_H(\widehat{S}_{V_H})$ dynamic steps, recompute $\widehat{S}_{V_H} \subset V_H$, an approximate minimum $s-t$ cut on $H$, ignoring all degree zero vertices	
		\end{enumerate}
		
	\end{algbox}
	
	\caption{Dynamic $(1 + \epsilon)$-approximate Minimum $s-t$ Cut}
	\label{fig:dynamic1approx}
	
\end{figure}

In this algorithm we run into the same issue of returning a cut of size $O(n)$ in amortized $O(\poly(\log{n},\epsilon^{-1}))$ time, and will allow a similar querying scheme. Let $V_H$ be the non-zero degree vertex set of $H$. Our vertex sparsification process ensures that $\VC \subseteq V_H$, so for the computed $\widehat{S}_{V_H} \subset V_H$, we will maintain the cut $\widehat{S}_{V_H} \cap \VC \subset \VC$ with $s \in \widehat{S}_{V_H} $. For a vertex $v \in \VC$, return $v$ is with $s$ iff $v \in \widehat{S}_{V_H} \cap \VC$, which takes $O(1)$ time. For a vertex $x \notin \VC$, note that all of $N(x)$ must be in $\VC$, and return that $x$ is with $s$ iff $\ww(x,\widehat{S}_{V_H} \cap \VC) = \ww^{(x)}(\widehat{S}_{V_H} \cap \VC)$ in $\tilde{G}$, taking $O(\poly(\log{n},\epsilon^{-1}))$ time to compute $\ww(x,\widehat{S}_{V_H} \cap \VC)$  and $\ww^{(x)}(\widehat{S}_{V_H} \cap \VC)$, by Corollary~\ref{lem:additional properties of fully dynamic cut sparsifier} and Corollary~\ref{cor:lowdegInd}. Note that by restricting to $\VC$ we will be able take advantage of the approximation guarantees of vertex sparsification in the corollary below.

\begin{corollary}
	\label{cor:dynCorrect}
	The dynamic algorithm maintains a $(1 + \epsilon)$-approximate minimum $s-t$ cut in $G$, and will only compute an approximate minimum $s-t$ cut on $H$ every $O(\epsilon \OPT)$ dynamic steps, taking $O(\OPT \cdot \poly(\log{n}, \epsilon^{-1}))$ time each computation 
	
\end{corollary}

\begin{proof}
	$\tilde{G} = F_1 \cupdot \ldots \cupdot F_K$ for some $K = O(\poly(\log{n},\epsilon^{-1}))$ by Corollary~\ref{lem:additional properties of fully dynamic cut sparsifier}, so from Lemma~\ref{lem:mvcOPT} and Corollary~\ref{cor:lowdegInd}, we know $|\VC| = O(\OPT\cdot \poly(\log{n},\epsilon^{-1}))$ and the degree of all vertices in $X\tilde{G}$ is $O(\poly(\log{n},\epsilon^{-1}))$.
	From Corollary~\ref{lem:additional properties of fully dynamic cut sparsifier}, the weights of $\tilde{G}$ are in $[1,O(n)]$, and so property (1) of Theorem~\ref{thm:vertexSparsify} implies that $H$ has $O(\OPT \cdot \poly(\log{n},\epsilon^{-1}))$ edges.
	Therefore, we can find
	a $(1 + \epsilonhat)$ approximate minimum $s-t$ cut in $H$,
	in $O(\OPT\cdot \poly(\log{n}, \epsilon^{-1}))$ time.
	
	Assume $\widehat{S}_{V_H} \subset V_H$ is returned as a $(1 + \epsilonhat)$-approximate minimum $s-t$ cut in $H$, with $\epsilonhat = \frac{\epsilon}{O(1)}$. Let $\widehat{S}_{\VC} = \widehat{S}_{V_H}  \cap \VC$ be its restriction to $\VC$, and let
	\[
	\widehat{S} = \widehat{S}_{VC} \cup
	\{x \in X\tilde{G}: \ww(x,\widehat{S}_{VC}) = \ww^{(x)}(\widehat{S}_{VC}) \}
	\]
	be the extension of $\widehat{S}_{\VC}$ onto $\tilde{G}$, which is the cut returned by our vertex querying scheme.	From Corollary~\ref{lem:additional properties of fully dynamic cut sparsifier} and Theorem~\ref{thm:vertexSparsify}, we have $G \approx_{\epsilonhat} \tilde{G}$ and $\tilde{G} \approx_{\epsilonhat}^{\VC} H$, respectively, which gives

	\[
	\Delta_{G}(\widehat{S}) \leq (1 + \epsilonhat )\Delta_{\tilde{G}}(\widehat{S})
	=  (1 + \epsilonhat) \Delta_{\tilde{G}}(\widehat{S}_{VC})
	\leq (1 + \epsilonhat )^2 \Delta_{H}(\widehat{S}_{VC}).
	\]

	On the other hand, let $\overline{S} \subset V$ be the minimum
	$s-t$ cut in $G$, and $\overline{S}_{VC} \subset \VC$
	be its restriction to $\VC$. Using the fact that $\Delta_{\tilde{G}}(\overline{S}_{\VC})$ is the weight of the minimal extension of $\overline{S}_{\VC}$ in $\tilde{G}$, along with the approximations $G \approx_{\epsilonhat} \tilde{G}$
	and $\tilde{G} \approx_{\epsilonhat}^{\VC} H$ gives

	\[
	\Delta_{G}(\overline{S}) \geq (1 - \epsilonhat ) \Delta_{\tilde{G}}(\overline{S})
	\geq (1 - \epsilonhat ) \Delta_{\tilde{G}}(\overline{S}_{VC})
	\geq (1 - \epsilonhat )^2 \Delta_{H}(\overline{S}_{VC}).
	\]	
	
	The near-optimality of $\widehat{S}_{V_H}$
	on $H$ and setting $\widehat{S}_{\VC} = \widehat{S}_{V_H} \cap \VC$, gives, \[ \Delta_{H}(\overline{S}_{VC}) \geq
	(1 - \epsilonhat) \Delta_{H}(\widehat{S}_{V_H}) \geq
	(1 - \epsilonhat) \Delta_{H}(\widehat{S}_{\VC})\]
	Therefore, $\Delta_{G}(\widehat{S}) \leq (1 + \epsilonhat)^5\Delta_{G}(\overline{S})$, and by choosing $\epsilonhat = \frac{\epsilon}{O(1)}$ we maintain a $(1 + \frac{\epsilon}{2})$-approximate minimum $s-t$ cut in $G$.

	An approximate minimum $s-t$ cut on $H$ will be re-computed in $\frac{\epsilon}{2}\Delta_H(\widehat{S}_{V_H})$ dynamic steps. $\OPT = \Delta_{G}(\overline{S}) \leq (1 + \epsilon)\Delta_H(\widehat{S}_{V_H})$, so $\Delta_H(\widehat{S}_{V_H}) = O(\OPT)$
	
\end{proof}

All that is left to be shown is that data structures can be maintained in $O(\poly(\log{n},\epsilon^{-1}))$ time per dynamic update. As a result of Corollary~\ref{lem:additional properties of fully dynamic cut sparsifier}, it suffices to show the following

\begin{theorem}
	\label{thm:dynMaintenance}
	For each addition/deletion of an edge in $\tilde{G}$, maintaining $\widehat{G}$, $H$, and $\VC$ takes $O(\poly(\log{n},\epsilon^{-1}))$ time.
	
\end{theorem}

As in Section~\ref{subsec:dataStructures}, most of the necessary analysis for Theorem~\ref{thm:dynMaintenance} will follow from the fact that all $x \in X\tilde{G}$ have degree $O(\poly(\log{n},\epsilon^{-1}))$, and the only substantial changes made to the data structures in one dynamic step, are done within the neighborhood of some $x \in X\tilde{G}$. We will also assume all of the dynamic data structure analysis of Section~\ref{subsec:dataStructures} with regards to maintaining a corresponding $G_{\VC}$ of some $G$. 

In the rest of this section, we will first examine dynamically maintaining the pre-processing routine, particularly when vertices are moved in and out of the vertex cover. Then we will consider dynamically maintaining our vertex sparsification routine. Most of the time complexity analysis will follow from Section~\ref{subsec:dataStructures}, and the only tricky part will be ensuring that dynamic changes do not multiply along iterations of the sparsification routine.

\paragraph{Maintaining $\widehat{G}$}
	\label{subsec:bucketMaintenance}

As with the multi-graph $G_{\VC}$, for $\widehat{G}\setminus X\tilde{G}$, an edge $e_{\emptyset}$ denotes an edge originally in $\tilde{G}$ and $e_x$ denotes an edge that was moved into $\widehat{G}\setminus X\tilde{G}$ from $N_x$. For each $x \in X\tilde{G}$, let $x_{\max}$ denote the vertex such that $(x,x_{\max})$ has the maximum weight in $N_x$. Let $bucket(x)$ be the $i \in [L]$ such that $\ww (x,x_{\max}) \in [2^{i-1},2^i]$. We can use 1 as our scalar here because all weights of $G$ are 1, so from Corollary~\ref{lem:additional properties of fully dynamic cut sparsifier}, all weights of $\tilde{G}$ are in $[1, O(n)]$. In order to maintain each $x_{\max}$, we will assume that the data structure of $\tilde{G}$ is such that the adjacency list of each $x$ is sorted by edge weight. Consequently, edge insertions/deletions in $\tilde{G}$ will require $O(\log{n})$ time.

Maintaining each bucket for an edge insertion/deletion in $\tilde{G}$ will be analogous to maintaining $G_{\VC}$ in Section~\ref{subsec:dataStructures}. We will first show that moving a vertex in and out of $X\tilde{G}$ can be done in $O(\poly(\log{n},\epsilon^{-1}))$ time, then give the overall update process, which will primarily just be composed of these two operations.

\begin{figure}
	\begin{algbox}
		$\textsc{RemoveXG}(\tilde{G}, X\tilde{G}, v)$
		\begin{enumerate}
			\item Delete all edges $e_v$ incident to $v_{\max}$ from $\widehat{G} \setminus X\tilde{G}$
			\item Delete all edges incident to $v$ from  $\widehat{G}_{bucket(v)}$
			\item For all edges $e$ incident to $v$ in $\tilde{G}$, add $e_{\emptyset}$  into $\widehat{G} \setminus X\tilde{G}$
		\end{enumerate}
	\end{algbox}
	
	\caption{Removing a Vertex from $X\tilde{G}$}
	\label{fig:removeXG}
	
\end{figure}

\begin{lemma}
	\label{lem:removeXG}
	If $v$ is not in $\VC$, then running $\textsc{RemoveXG}(\widehat{G},X\tilde{G},v)$ will output $\widehat{G}$ with $v \in \VC$ in $O(deg_v\log{n})$ time, where $deg_v$ is the degree of $v$ in $\tilde{G}$	
\end{lemma}

\begin{proof}
	Costs of the three steps are:
	\begin{enumerate}
		\item  Deleting all edges $e_v$ incident to $v_{\max}$ from $\widehat{G} \setminus X\tilde{G}$ takes $O(\log{n})$ time per deletion and $O(deg_v)$ deletions.
		
		\item  Deleting all edges incident to $v$ from  $\widehat{G}_{bucket(v)}$ takes $O(\log{n})$ time per deletion and $O(deg_v)$ deletions.
		
		\item Adding $e_{\emptyset}$  into $\widehat{G} \setminus X\tilde{G}$ takes $O(\log{n})$ time and is done for all edges $e$ incident to $v$ in $\tilde{G}$, so $O(deg_v)$ times
	\end{enumerate}
	
	If $v$ is not in $\VC$, then $v$ cannot be incident to any vertices in $X\tilde{G}$. Therefore, placing $v$ in $\VC$ implies that $v$ cannot be incident to any edges in all $\widehat{G}_k$ and no edges $e_v$ exist in $\widehat{G} \setminus X\tilde{G}$. $\textsc{RemoveXG}(\widehat{G},X\tilde{G},v)$ performs exactly these removals and inserts all necessary $e_{\emptyset}$ incident to $v$ into $\widehat{G}\setminus X\tilde{G}$

\end{proof}

\begin{figure}
	\begin{algbox}
		$\textsc{InsertXG}(\widehat{G}, X\tilde{G}, v)$
		\begin{enumerate}
			\item Delete all edges $e_{\emptyset}$ incident to $v$ in $\widehat{G} \setminus X\tilde{G}$
			\item For all edges $e = (v,w) \in \tilde{G}$ incident to $v$
			\begin{enumerate}
				\item If $\ww (v,w) < \frac{\epsilon}{d} \ww (v,v_{\max})$: insert $(w,v_{\max})_v$ into $\widehat{G} \setminus X\tilde{G}$
				\item Otherwise: insert $(v,w)$ into $\widehat{G}_{bucket(v)}$
			\end{enumerate}
		\end{enumerate}
	\end{algbox}
	
	\caption{Inserting a Vertex into $X\tilde{G}$}
	\label{fig:insertXG}
	
\end{figure}

\begin{lemma}
	\label{lem:insertXG}
	If $v$ is not in $\VC$, but was placed in $\VC$ for $\widehat{G}$, then running $\textsc{InsertXG}(\widehat{G},X\tilde{G},v)$ will output $\widehat{G}$ with $v \notin \VC$ in $O(deg_v \log{n})$ time, where $deg_v$ is the degree of $v$ in $\tilde{G}$	
\end{lemma}

\begin{proof}
	Costs of the two steps are:
	\begin{enumerate}
		\item  Deleting all edges $e_{\emptyset}$ incident to $v$ in $\widehat{G} \setminus X\tilde{G}$ takes $O(\log{n})$ time per deletion and $O(deg_v)$ deletions.
		
		\item  Checking if $\ww (v,w) < \frac{\epsilon}{d} \ww (v,v_{\max})$ and inserting $(w,v_{\max})_v$ into $\widehat{G} \setminus X$ or inserting $(v,w)$ into $\widehat{G}_{bucket(v)}$ takes $O(\log{n})$ time. This is done for all edges $e = (v,w) \in \tilde{G}$ incident to $v$, so $O(deg_v)$ times
		
	\end{enumerate}
	
	If $v$ is not in $\VC$, but was placed in $\VC$ for $\widehat{G}$, then only edges $e_{\emptyset}$ are incident to $v$ in $\widehat{G}$. Removing $v$ from $\VC$ requires deleting all of these edges. Further, all edges $e$ in $N_v$ of sufficiently small weight must be moved to $\widehat{G} \setminus X\tilde{G}$ as $e_v$, and the rest of $N_v$ must be placed in the appropriate $\widehat{G}_i$. $\textsc{InsertXG}(\widehat{G},X\tilde{G},v)$ performs exactly these operations.
	
\end{proof}

The full dynamic update process of $\widehat{G}$ for each $e = (u,v)$ insertion/deletion in $\tilde{G}$ will then be as follows.

\begin{enumerate}
	\item For $u$ and $v$, $\textsc{RemoveXG}(\widehat{G},X\tilde{G},v)$ if $v \notin \VC$
	
	\item Update $\VC$ and $\tilde{G}$ as done in section 5
	
	\item Add/delete $(u,v)_{\emptyset}$ from $\widehat{G} \setminus X\tilde{G}$
	
	\item Update $u_{\max}$ and $v_{\max}$, which will simply require looking at the first edge incident to $u$ and $v$ in $\tilde{G}$, as the list is sorted by weight
	
	\item For $u$ and $v$, $\textsc{InsertXG}(\widehat{G},X\tilde{G},v)$ if $v \notin \VC$
	
\end{enumerate}

\begin{lemma}
	\label{lem:widehatMaintanence}
	For each edge addition/deletion in $\tilde{G}$, maintaining $\widehat{G} = \widehat{G} \setminus X\tilde{G} \cupdot \widehat{G}_1 \cupdot \ldots \cupdot \widehat{G}_L$ takes $O(\poly(\log{n},\epsilon^{-1}))$ time.
	
\end{lemma}

\begin{proof}
	Note that $\textsc{InsertXG}(\widehat{G},X\tilde{G},v)$ and $\textsc{RemoveXG}(\widehat{G},X\tilde{G},v)$ are only performed if $v \notin \VC$, which implies that the degree of $v$ in $\tilde{G}$ is $O(\poly(\log{n},\epsilon^{-1}))$. Updating $\VC$ and $\tilde{G}$ is known to take $O(\poly(\log{n},\epsilon^{-1}))$ time. Steps 3 and 4 clearly take $O(\log{n})$ time. Therefore, the full runtime of this update process is $O(\poly(\log{n},\epsilon^{-1}))$.
	
\end{proof}

\paragraph{Maintaining BoundedVertexSparsify}
\label{subsec: boundedMaintenance}

We will dynamically sparsify the multi-graph $\widehat{G}\setminus X\tilde{G}$ as per usual, so each edge insertion/deletion requires $O(\poly(\log{n},\epsilon^{-1}))$ update time for $\widehat{G} \setminus X\tilde{G}$. Accordingly, we will only consider maintaining the necessary data structures for $\textsc{BoundedVertexSparsify}$ of each $\widehat{G}_k$, which we will simply denote as $G$ with bipartition $(\VC,XG)$.

Alterations to $G$ are made by the dynamic update process in the previous section, which implies that we only need to consider the following changes to $G$. Add/Delete a vertex $x$ from $X$, and add/delete $N_x$ from $G$. Add/Delete an edge within $N_x$ for some $x \in X$. If an edge is added/deleted from $N_x$, we will simply delete $N_x$ from $G$, and then add $N_x$ with the edge added/deleted to $G$. Accordingly, in order to establish that our data structures can be maintained in $O(\poly(\log{n},\epsilon^{-1}))$ update time, we just need to show that adding/deleting any $N_x$ from $G$ can be done in $O(\poly(\log{n},\epsilon^{-1}))$ update time.

For each level $i$ of computing a light vertex set and running $\textsc{Sample}$, we need to maintain $G_i$, $G_{\VC}^i$, $XG^{light}_i$ and all $F_{i,j}$ in $F_i$. The data structures for $G_i$ and $G_{\VC}^i$ will be as in Subsection~\ref{subsec:dataStructures}. Assume that the data structure of each $F_{i,j}$ is such that we can search for edges in $O(\log{n})$-time, either by search trees or linked lists with back pointers (see e.g.~\cite{CormenLRS09}, Chapters 10.2, 10.3, and 13). The data structure each $XG^{light}_i$ will just be a list of vertices with insertion/deletion taking $O(\log{n})$ time.

We will still assume edge additions/deletions in $G_i$, $G_{\VC}^i$ can be maintained in $O(\poly(\log{n},\epsilon^{-1}))$, as was shown in Subsection~\ref{subsec:dataStructures}. Most of the time complexity analysis will then follow from this, and we just need to establish that the additions/deletions will not multiply as we move down the pipeline. This will ultimately follow from our construction of the $t$-clique forests.

\paragraph{Adding some $N_x$ to $G_i$}

The algorithm in Figure~\ref{fig:insertStar} will add a vertex $x$ to $G_i$, along with the corresponding $N_x$. 

\begin{figure}
	\begin{algbox}
		$\textsc{InsertStar}(G_i, XG_i, XG_i^{light}, N_x)$
		\begin{enumerate}
			\item Update $G_i \gets G_i \cup N_x$, $XG_i \gets XG_i \cup x$, and insert $K_x$ into $G_{\VC}^i$
			\item For the first $e_x \in K_x$ that can be added to some $F_{i,j}$: Update $F_{i,j} \gets F_{i,j} \cup e_x$, $XG^{light}_i \gets XG^{light}_i \cup x$, and remove $K_x$ from $G_{\VC}^i$
			\item If no $e_x \in K_x$ can be added to any $F_{i,j}$, with probability $\frac{1}{2}$: run $\textsc{InsertStar}(G_{i+1}, XG_{i+1}, XG_{i+1}^{light}, 2N_x)$
		\end{enumerate}
	\end{algbox}
	
	\caption{Add $N_x$ to $G_i$}
	\label{fig:insertStar}
	
\end{figure}

\begin{lemma}
	\label{lem:insertStar}
	$\textsc{InsertStar}(G_i, XG_i, XG_i^{light}, N_x)$ adds $N_x$ to $G_i$ while maintaining $t$-clique forest $F_i$
	
\end{lemma}

\begin{proof}
If some $e_x \in K_x$ can be added to some $F_{i,j}$, then by construction,  $F_i \cap K_x = e_x$ and $x \in XG^{light}_i$. Therefore, $F_i$ is still a $t$-clique forest, and $x \in XG^{light}_i$ implies $x \notin XG_{i+1}$, so it is only necessary to add $e_x$ to $F_{i,j}$ and $x$ to $XG^{light}_i$.

If no $e_x \in K_x$ can be added to any $F_{i,j}$, then $F_i \cap K_x = \emptyset$ and $x \in XG^{heavy}_i$. Therefore, $F_i$ is still a $t$-clique forest, and $x \in XG^{heavy}_i$ implies a coin must be flipped to determine whether $x$ is added to $XG_{i+1}$ and $2N_x$ is added to $G_{i+1}$.

\end{proof}

Furthermore, we still maintain $G_{\VC}^i = \bigcup_{x \in XG^{heavy}_i} K_x$

\paragraph{Deleting some $N_x$ from $G_i$}

The algorithm in Figure~\ref{fig:removeStar} will delete a vertex $x$ from $G_i$, along with the corresponding $N_x$.

\begin{figure}
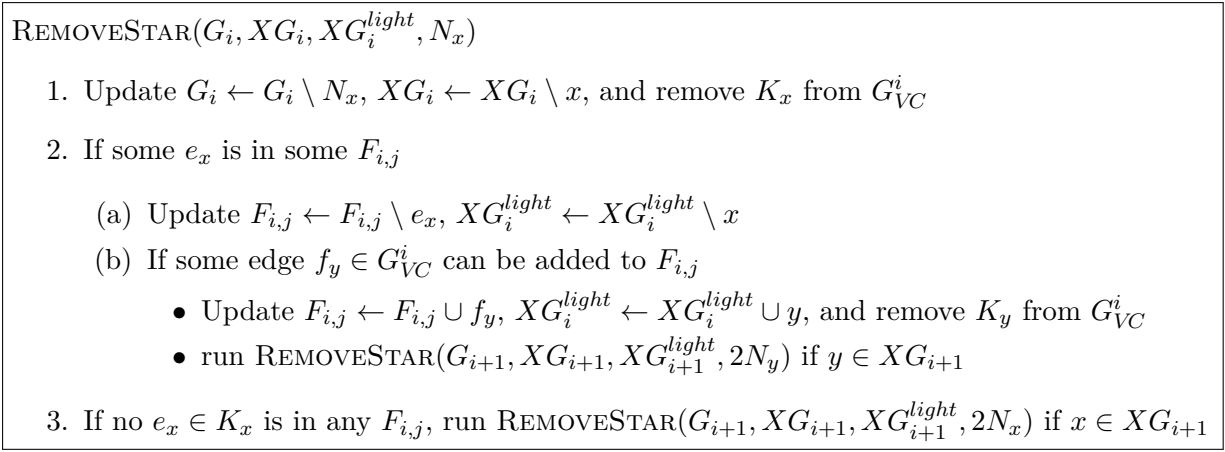

	\begin{algbox}
		$\textsc{RemoveStar}(G_i, XG_i, XG_i^{light}, N_x)$
		\begin{enumerate}
			\item Update $G_i \gets G_i \setminus N_x$, $XG_i \gets XG_i \setminus x$, and remove $K_x$ from $G_{\VC}^i$
			\item If some $e_x$ is in some $F_{i,j}$
			\begin{enumerate}
				\item Update $F_{i,j} \gets F_{i,j} \setminus e_x$, $XG^{light}_i \gets XG^{light}_i \setminus x$
				\item If some edge $f_y \in G_{\VC}^i$ can be added to $F_{i,j}$
				\begin{itemize}
					\item Update $F_{i,j} \gets F_{i,j} \cup f_y$, $XG^{light}_{i} \gets XG^{light}_{i} \cup y$, and remove $K_y$ from $G_{\VC}^i$
					\item run $\textsc{RemoveStar}(G_{i+1}, XG_{i+1}, XG_{i+1}^{light}, 2N_y)$ if $y \in XG_{i+1}$
				\end{itemize}
			\end{enumerate}
			\item If no $e_x \in K_x$ is in any $F_{i,j}$,  run $\textsc{RemoveStar}(G_{i+1}, XG_{i+1}, XG_{i+1}^{light}, 2N_x)$ if $x \in XG_{i+1}$
		\end{enumerate}
	\end{algbox}
	
	\caption{Remove $N_x$ from $G_i$}
	\label{fig:removeStar}
	
\end{figure}

\begin{lemma}
	\label{lem:removeStar}
	$\textsc{RemoveStar}(G_i, XG_i, XG_i^{light}, N_x)$ removes $N_x$ from $G_i$ while maintaining $t$-clique forest $F_i$
	
\end{lemma}

\begin{proof}
	If we had $F_i \cap K_x = e_x$, then $x$ was in  $XG^{light}_i$, so $e_x$ must be removed from some $F_{i,j}$ and $x$ must be removed from $XG^{light}_i$. $F_i$ was a $t$-clique forest and $G_{\VC}^i = \bigcup_{x \in XG^{heavy}_i} K_x$ (as was noted), implying that multiple edges in $G_{\VC}^i$ cannot be added to $F_{i,j}$ without creating a cycle. If $f_y$ is added to $F_{i,j}$ then $y$ is added to $XG^{light}_i$ and $F_i \cap K_y = f_y$. Therefore, $F_i$ is still a $t$-clique forest, and because $y \in XG^{light}_i$, it is now necessary to remove $2N_y$ from $G_{i+1}$ if $y \in XG_{i+1}$.

	If we have $F_i \cap K_x = \emptyset$, then $x \in XG^{heavy}_i$ and $F_i$ is still a $t$-clique forest. Further $XG_{i+1} \subseteq XG_i$, so it is necessary to remove $2N_x$ from $G_{i+1}$ if $x \in XG_{i+1}$.

\end{proof}

Furthermore, we still maintain $G_{\VC}^i = \bigcup_{x \in XG^{heavy}_i} K_x$

\begin{lemma}
	\label{lem:boundedMaintenance}
	For any addition/deletion of some $x$ from $XG_0$ and $N_x$ from $G_0$, maintaining $H$ takes $O(t \cdot \poly(\log{n},\epsilon^{-1}))$ time
	
\end{lemma}

\begin{proof}
	Checking each forest for an edge insertion/deletion takes $O(t\log{n})$ time. It follows almost immediately from the analysis in Subsection~\ref{subsec:dataStructures} that the rest of the computation in one iteration of $\textsc{InsertStar}$ and $\textsc{RemoveStar}$ takes $O(\poly(\log{n},\epsilon^{-1}))$ time. Furthermore, both can make at most one recursive call to themselves, so adding/deleting $N_x$ from $G_0$ takes $O(l\cdot t\cdot \poly(\log{n},\epsilon^{-1}))$ time where $l = O(\log{n})$.

\end{proof}

{\paragraph{Proof of Theorem~\ref{thm:dynMaintenance}}: Any edge insertion/deletion in $\tilde{G}$ requires $O(\poly(\log{n},\epsilon^{-1}))$ update time for $\widehat{G}$ and $\VC$ from Lemma~\ref{lem:widehatMaintanence}. Therefore, there are at most $O(\poly(\log{n},\epsilon^{-1}))$ additions/deletions of some $N_x$ to some $\widehat{G}_i$, which will require $O(t \cdot \poly(\log{n},\epsilon^{-1}))$ update time from Lemma~\ref{lem:boundedMaintenance}, where $t = O(\poly(\log{n}, \epsilon^{-1}))$. Thus, the full dynamic update process of all data structures takes $O(\poly(\log{n}, \epsilon^{-1}))$ time per dynamic update of $\tilde{G}$.

\printbibliography

\newpage
\appendix

\section{Omitted Proofs of \Cref{sec:spectral sparsifier general}}\label{apx:proofs spectral sparsifier}

In the following we give the omitted proofs of section \Cref{sec:spectral sparsifier general}, which mainly use standard arguments.

\sparsifyingstep*

\begin{proof}
Let \begin{equation*}
R = \frac{\epsilon^2}{3 (c+1) \ln{n}}.
\end{equation*}
For every edge $ e \in G \setminus B $, let $ X_e $ be the random variable that is $ 4 \ww_{G} (e) \cdot \lap_e $ with probability $ 1/4 $ and $ 0 $ with probability $ 3/4 $.
We further set $ \lap_{B^{(j)}} $ for every $ 1 \leq j \leq \lceil 1 / R \rceil $ as follows:
\begin{equation*}
	\lap_{B_i^{(j)}} =
	\begin{cases}
		R \cdot \lap_{B_i} & \text{if $ 1 \leq j \leq \lfloor 1 / R \rfloor $} \\
		\lap_{B_i} - \lfloor 1 / R \rfloor R \cdot \lap_{B_i} & \text{if $ j = \lceil 1 / R \rceil $}
	\end{cases}
\end{equation*}
Note that this definition simply guarantees that $ \sum_{j=1}^{\lceil 1 / R \rceil} \lap_{B^{(j)}} = \lap_{B_i} $ and $ \lap_{B_i^{(j)}} \leq R \cdot \lap_{B_i} $ for every $ 1 \leq j \leq \lceil 1 / R \rceil $.
We now want to apply \Cref{thm:matrix concentration} with the random variables $ Y = \sum_{e \in G \setminus B} X_e + \sum_{j=1}^{\lceil 1 / R \rceil} \lap_{B^{(j)}} $ and $ Z = \lap_{G} $.
Observe that
\begin{eqnarray*}
	\expec{Y} & = & E \left[ \sum_{e \in G \setminus B} X_e + \sum_{j=1}^{\lceil 1 / R \rceil} \lap_{B^{(j)}} \right]\\ & = & \sum_{e \in G \setminus B} \expec{X_e} + \sum_{j=1}^{\lceil 1 / R \rceil} \lap_{B^{(j)}} \\
	& = & \sum_{e \in G \setminus B} \lap_e + \lap_{B} =  \lap_{G} = Z \, .
\end{eqnarray*}
For every edge $ e \in G \setminus B $, using \Cref{lem:stretch resistance correspondence}, we have
\begin{equation*}
	 X_e \preceq 4 \ww_{G} (e) \cdot \lap_e \preceq \frac{\alpha}{t} \cdot \lap_G \leq R \cdot \lap_G \, .
\end{equation*}
Furthermore, using $ B\preceq G $, we have
\begin{equation*}
	\lap_{B_i^{(j)}} \leq R \cdot \lap_{B_i} \preceq R \cdot \lap_{G_{i-1}}
\end{equation*}
for every $ 1 \leq j \leq \lceil 1 / R \rceil $.
Thus, the preconditions of \Cref{thm:matrix concentration} are satisfied.
We conclude that we have $ \lap_{G_H} \preceq (1 + \epsilon ) \lap_{G} $ with probability at least
\begin{equation*}
n \cdot \exp (- \epsilon^2 / 2R) \geq n \cdot \exp ((c+1) \ln{n}) = 1 / n^{c+1} \, .
\end{equation*}

A symmetric argument can be used for $ (1 - \epsilon ) \lap_{G} \preceq \lap_{H} $.
\end{proof}

\spectralsparsifiercorrectness*

\begin{proof}
Note that since $ H = \bigcup_{i=1}^k B_i \cup G_k $ we have
\begin{equation*}
	\lap_H = \lap_{G_k} + \sum_{i=1}^k \lap_{B_i} \, .
\end{equation*}
We now prove by induction on $ j $ that $ \lap_{G_k} + \sum_{i=k-j+1}^k \lap_{B_i} \preceq (1 + \epsilon/(2k))^j \lap_{G_{k-j}} $.
This claim is trivially true for $ j = 0 $.
For $ 1 \leq j \leq k $, we use the induction hypothesis and \Cref{lem:sparsifying step}, which both hold with high probability, to get
\begin{align*}
	\lap_{G_k} + \sum_{i=k-j+1}^k \lap_{B_i} &= \lap_{G_k} + \sum_{i=k-j+2}^k \lap_{B_i} + \lap_{B_{k-j+1}} \\
		&\preceq (1 + \epsilon/(2k))^{j-1} \lap_{G_{k-j+1}} + \lap_{B_{k-j+1}} \\
		&\preceq (1 + \epsilon/(2k))^{j-1} (\lap_{G_{k-j+1}} + \lap_{B_{k-j+1}}) \\
		&\preceq (1 + \epsilon/(2k))^j \lap_{G_{k-j}} \, .
\end{align*}
We now have $ \lap_H \preceq (1 + \epsilon/(2k))^k \lap_G $ with high probability by setting $ j = k $.
Using symmetric arguments we can prove $ (1 - \epsilon/(2k))^k \lap_G \preceq \lap_H $.
Since $ (1 - \epsilon/(2k))^k \geq 1 - \epsilon $ and $ (1 + \epsilon/(2k))^k \leq 1 + \epsilon $, the claim follows.
\end{proof}

\spectralsparsifiersize*

\begin{proof}
We will show that, with probability $1-2n^{-c+1}$, every iteration
$j$ computes a graph $G_{j+1}$ with half the number of edges in $G_j$.
By a union bound, the probability that this fails to be true
for any $j<n$ is at most $2n^{-c}$. This implies all claims.

We use the following  standard Chernoff bound:
Let $ X = \sum_{k=1}^N X_k $, where $ X_k = 1 $ with probability $ p_k $ and $ X_k = 0 $ with probability $ 1 - p_k $, and all $ X_k $ are independent. Let $ \mu = \expec{X} = \sum_{k=1}^N p_k $.
Then $ \prob{X \geq (1 + \delta) \mu} \leq \exp(- \frac{\delta^2}{2+\delta} \mu) $ for all $ \delta > 0 $.

We apply this bound on the output of \textsc{Light-Spectral-Sparsify} for every $j$.
Concretely, we assign a random variable to each edge $e$ of $G_j$, with $ X_e = 1 $ if and only if $ e $ is added to $ G_{j+1} $. Then $ \expec{X} = N / 4 $. By construction, the number of edges in $G_j$ is
$ N  \geq  (c+1) \log{n} $. Applying the Chernoff bound with $ \delta = 2 $ we get
\begin{equation*}
\prob{X \geq 2N} \leq \frac{1}{e^{N/4}} \leq \frac{1}{e^{((c+1) \log{n}) / 4}} = \frac{1}{e^{1/4} n^{c+1}} \leq \frac{1}{2 n^{c+1}} \, .
\end{equation*}
\end{proof}

\section{Guarantees of Combinatorial Reductions}
\label{sec:min_cut_proofs}

We show some of the structural results necessary for the
reductions in Sections~\ref{sec:dynamic min cut}, \ref{sec:vertSparsify}, and \ref{sec:onePlusEpsilon}.
We first show the guarantees of $K_x$:

\begin{proof} (of Theorem~\ref{thm:schurComplement})
	For any $x \in X$ and $S_{\VC} \subset \VC$, let $\ww_{K_x}(S_{\VC})$ denote the weight of cutting $S_{\VC}$ in $K_x$. Consequently, for any $S_{\VC} \subset \VC$, $\Delta_{G_{\VC}}(S_{\VC}) = \Delta_{G \setminus X}(S_{\VC}) + \sum_{x \in X} \ww_{K_x}(S_{\VC})$, and it suffices to show that for all $x \in X$, $\frac{1}{2} \ww^{(x)}(S_{\VC}) \leq \ww_{K_x}(S_{\VC}) \leq \ww^{(x)}(S_{\VC})$.
\end{proof}

\begin{lemma}
	For any $x \in X$ and $S \subset \VC$, we have $\frac{1}{2} \ww_{K_x}(S) \leq \ww^{(x)}(S) \leq \ww_{K_x}(S)$
\end{lemma}

\begin{proof}
	Without loss of generality, assume $\ww(x,S) \leq \ww(x,\VC\setminus S)$, so $\ww^{(x)}(S) = \ww(x,S) = \sum_{u \in S \cap N(x)} \ww (x,u)$ 
	
	\[\ww_{K_x}(S) = \sum_{u \in S \cap N(x)} \sum_{v \in (\VC \setminus S) \cap N(x)} \frac{\ww (x,u)\ww (x,v)}{\sum_{i \in N(x)} \ww (x,i)} = \sum_{u \in S \cap N(x)} \ww (x,u) \frac{\ww(x,\VC\setminus S)}{\sum_{i \in N(x)} \ww (x,i)}\] where by definition $\sum_{i \in N(x)} \ww (x,i) = \ww(x,S) + \ww(x,\VC\setminus S)$ and so by assumption \[\frac{1}{2} \leq \frac{\ww(x,\VC\setminus S)}{\sum_{i \in N(x)} \ww (x,i)} \leq 1\]
\end{proof}

\begin{proof}(of Lemma~\ref{lem:elimWeight})
	Each edge in $ (u,v)_x \in G_{\VC}$ has weight \[ \ww_{(u,v)_x} = 
	\frac{\ww (x,v) \ww (x,u)}{\sum_{i \in N(x)} \ww (x,i)}
	\], $\ww (x,v) \ww (x,u) \geq \gamma^2$ and $\sum_{i \in N(x)} w_{(x,i)} \leq \gamma U d$. Also, $\sum_{i \in N(x)} \ww (x,i) \geq \max\{\ww (x,v) \ww (x,u) \},$ implying \[ \frac{\ww (x,v) \ww (x,u)}{\sum_{i \in N(x)} \ww (x,i)} \leq
	\frac{\max\{\ww (x,v) \ww (x,u)\}^2}{\sum_{i \in N(x)} \ww (x,i)} \leq \max\{\ww (x,v) \ww (x,u)\}
	\]
\end{proof}

Next we bound the size of the vertex cover formed
by removing all leaves, compared to the optimum.

\begin{proof}(of Lemma~\ref{lem:treeApprox})
	From~\cite{GuptaS09,DabneyDH09}, given a tree $T_0$ with root $r_0$, leaves $l(T_0)$, and parents of the leaves $p(T_0)$, the greedy algorithm of taking $p(T_0)$ and iterating on $T_1 = T_0 \setminus \{l(T_0) \cup p(T_0)\}$, with $r_1 = r_0$ or $r_1$ arbitrary if $r_0 \in p(T_0)$, will give a minimum vertex cover of $T_0$. If $T_1$ is a forest, iterate on each tree of the forest, where $r_0$ is the root of whichever tree it is contained in, and the remaining trees are arbitrarily rooted. Assume that if $T_i = r_i$ for some $i$, then $p(T_i) = \emptyset$.
	
	Set $T = T_0$ and $r = r_0$, and suppose $T_0$ can be decomposed into $T_0 \ldots T_d$ as above. Therefore, $\bigcup_{i=0}^d p(T_i)$ is a minimum vertex cover, and $\VC$ is $p(T_d) \cup \bigcup_{i=0}^{d-1} (p(T_i) \cup l(T_{i+1}))$
	
	By construction, all $p(T_i)$ and $l(T_j)$ are disjoint, and we claim that $|p(T_i)| \geq |l(T_{i+1})|$ for all $i$. Assume $T_i$ is a tree, and this will clearly still hold if $T_i$ is a collection of disjoint trees. Each vertex in $l(T_{i+1})$ was not a leaf in $T_i$ and is now a leaf in $T_{i+1}$. Further, $r_i \notin l(T_{i+1})$ because if $r_i \in T_{i+1}$, then $r_{i+1} = r_i$. Therefore, each vertex in $l(T_{i+1})$ must have had its degree reduced by removing $l(T_i)$ and $p(T_i)$. A vertex in $l(T_{i+1})$ cannot be connected to a vertex in $l(T_i)$ because then it would be in $p(T_i)$. Consequently, it must be connected to some vertex in $p(T_i)$, and if $|p(T_i)| < |l(T_{i+1})|$, then two vertices in $l(T_{i+1})$ must be connected to the same vertex in $p(T_i)$, creating a cycle in $T_i$, giving a contradiction. Thus	\[|\VC| = p(T_d) + \sum_{i=0}^{d-1} (|p(T_i)| + |l(T_{i+1})|) \leq p(T_d) + \sum_{i=0}^{d-1} 2|p(T_i)| \leq 2 \sum_{i=0}^d |p(T_i)| = 2|\MVC| \]

\end{proof}

\end{document}